\documentclass{lmcs} 
\pdfoutput=1
\usepackage[utf8]{inputenc}

\usepackage{lastpage}
\lmcsdoi{20}{3}{1}
\lmcsheading{}{\pageref{LastPage}}{}{}%
{Jul.~04,~2023}{Jul.~02,~2024}{}

\pdfoutput=1 

\keywords{First Order Logic with Data, Satisfiability Problem}

\usepackage{hyperref}

\usepackage{amsmath}
\usepackage{amssymb}
\usepackage{mathtools}
\usepackage[
    linecolor=blue,
    bordercolor=blue,
    backgroundcolor=white]{todonotes}
\usepackage{environ}
\usepackage{stmaryrd}
\usepackage{subcaption}
\makeatletter         
\def\figurecaption#1#2{\noindent\hangindent 40pt
                       \hbox to 36pt {\small\sl #1 \hfil}
                       \ignorespaces {\small #2}}

\long\def\@makecaption#1#2{
  \vskip 10pt 
  \settowidth{\@tempdima}{#2}
  \ifdim\@tempdima>0pt
       \setbox\@tempboxa\hbox{#1: #2}
     \else
       \setbox\@tempboxa\hbox{#1 #2}
   \fi
   \ifdim \wd\@tempboxa >\hsize               
       \begin{list}{#1:}{
       \settowidth{\labelwidth}{#1:}
       \setlength{\leftmargin}{\labelwidth}
       \addtolength{\leftmargin}{\labelsep}
        }\item #2 \end{list}\par   
     \else                                    
       \hbox to\hsize{\hfil\box\@tempboxa\hfil}  
   \fi}
\makeatother
\usepackage{array}

\usepackage{upgreek}

\usepackage{graphicx}

\usepackage{cite}

\usepackage{multirow}

 \theoremstyle{plain}
 \newtheorem{theorem}[thm]{Theorem}
 \newtheorem{proposition}[thm]{Proposition}
 \newtheorem{lemma}[thm]{Lemma}
 \theoremstyle{definition}
\newtheorem{example}[thm]{Example}

\newtheorem{remark}[thm]{Remark}

\DeclareMathOperator{\N}{\mathbb{N}}
\renewcommand{\phi}{\varphi}
\newcommand{\vp}{\varphi}
\newcommand{\vpeta}{\vp_\eta}
\newcommand{\vpetai}[1]{\vp_\eta^{#1}}

\newcommand{\Unary}{\Sigma}
\newcommand{\Unarycl}[1]{\Lambda_{#1}}
\newcommand{\unarycounter}[1]{\eta_{#1}}
\newcommand{\ediag}{\mathsf{ed}}
\newcommand{\Binary}{\Gamma}

\newcommand{\unary}{\sigma}

\newcommand{\uP}[1]{P_{#1}}

\def \AA{\mathfrak{A}}
\def \BB{\mathfrak{B}}
\def \GG{\mathcal{G}}

\newcommand{\nbd}{\mathit{Da}} 

\newcommand{\ifunct}{f_1}
\newcommand{\ofunct}{f_2}
\newcommand{\funi}{\ifunct}
\newcommand{\funo}{\ofunct}
\newcommand{\funct}[1]{f_{#1}}
\newcommand{\Data}[1]{\textup{Data}[#1]}
\newcommand{\nData}[2]{\textup{Data}[{#1},{#2}]}
\newcommand{\Values}[1]{\Valuessub{#1}{A}}
\newcommand{\Valuessub}[2]{\mathit{Val}_{#1}(#2)}

\newcommand{\f}[1]{f_{#1}}

\newcommand{\rel}{\sim}
\newcommand{\relsaa}[5]{#4 \mathrel{_{#1}{\sim^{#3}_{#2}}} #5}
\newcommand{\relsaaord}[3]{{_{#1}{\sim^{#3}_{#2}}}}
\newcommand{\relsaord}[2]{{_{#1}{\sim_{#2}}}}
\newcommand{\rels}[4]{#3 \mathrel{_{#1}{\sim}{_{#2}}} #4}

\newcommand{\distaa}[3]{\mathit{d}^{#3}(#1,#2)}
\newcommand{\vprojr}[3]{#1|_{#2}^{#3}}

\newcommand{\datagraph}[1]{\GG(#1)}
\newcommand{\gaifmanish}[1]{\datagraph{#1}}
\newcommand{\Vertex}[1]{\mathit{V}_{#1}}
\newcommand{\Edge}[1]{\mathit{E}_{#1}}
\newcommand{\Ball}[3]{\mathit{B}_{#1}^{#3}(#2)}

\newcommand{\FO}{\textup{FO}}

\newcommand{\ndFO}[2]{\textup{dFO}[{#1},{#2}]}
\newcommand{\ndFOvar}[3]{\textup{dFO}^{#3}[{#1},{#2}]}
\newcommand{\extndFOvar}[3]{\textup{ext-dFO}^{#3}[{#1},{#2}]}
\newcommand{\ndFOr}{\textup{dFO}}
\newcommand{\ndFOrvar}[1]{\textup{dFO}^{#1}}
\newcommand{\extndFOrvar}[1]{\textup{ext-dFO}^{#1}}
\newcommand{\rndFO}[3]{#3\textup{-Loc-dFO}[#1,#2]}
\newcommand{\rndFOr}[1]{#1\textup{-Loc-dFO}}
\newcommand{\eFO}[3]{\exists\textup{-}{#3}\textup{-Loc-dFO}[{#1},{#2}]}
\newcommand{\eFOr}[1]{\exists\textup{-}{#1}\textup{-Loc-dFO}}
\newcommand{\qfFO}[3]{\textup{qf-}{#3}\textup{-Loc-dFO}[{#1},{#2}]}

\newcommand{\Pform}[2]{{#1}(#2)}

\newcommand{\Var}{\mathcal{V}}
\newcommand{\Intrepl}[2]{{I[#1/#2]}}
\newcommand{\Intrepltwo}[4]{{I[#1/#2][#3/#4]}}

\makeatletter
\newcommand{\problemtitle}[1]{\gdef\@problemtitle{#1}}
\newcommand{\probleminput}[1]{\gdef\@probleminput{#1}}
\newcommand{\problemquestion}[1]{\gdef\@problemquestion{#1}}
\NewEnviron{decproblem}{
  \problemtitle{}\probleminput{}\problemquestion{}
  \BODY
  \par\addvspace{0.8\baselineskip}
  \noindent
  \fbox{\centering
  \begin{tabular}{rl}
    \multicolumn{2}{l}{\normalsize \@problemtitle} \\
    \hline
    \rule{0pt}{3ex}
    \normalsize \textbf{Input:} & \normalsize \@probleminput \\[0.5ex]
    \normalsize \textbf{Question:} & \normalsize \@problemquestion
  \end{tabular}}
  \par\addvspace{0.8\baselineskip}
}

\DeclareMathOperator{\donc}{\;\Rightarrow\;}
\DeclareMathOperator{\et}{\wedge}
\DeclareMathOperator{\Et}{\bigwedge}
\DeclareMathOperator{\ou}{\vee}
\DeclareMathOperator{\Ou}{\bigvee}

\makeatletter
\newsavebox{\@brx}
\newcommand{\llangle}[1][]{\savebox{\@brx}{\(\m@th{#1\langle}\)}%
  \mathopen{\copy\@brx\kern-0.5\wd\@brx\usebox{\@brx}}}
\newcommand{\rrangle}[1][]{\savebox{\@brx}{\(\m@th{#1\rangle}\)}%
  \mathclose{\copy\@brx\kern-0.5\wd\@brx\usebox{\@brx}}}
\makeatother

\newcommand{\locformr}[3]{\llangle {#1} \rrangle_{#2}^{ #3}}

\newcommand{\DataSat}[1]{\ensuremath{\textsc{DataSat}(#1)}}
\newcommand{\nDataSat}[2]{\ensuremath{\textsc{DataSat}(#1,#2)}}

\newcommand{\im}[1]{\mathit{Im}(#1)}

\newcommand{\phit}[1]{\sem{#1}}

\newcommand{\sem}[1]{\ensuremath{[\![#1]\!]}}
\newcommand{\AAas}{\sem{\AA}_{\tuple{a}}}
\newcommand{\tuple}[1]{\vec{#1}}

\newcommand{\Unaryp}{\Unary'}
\newcommand{\Udeci}[1]{\Lambda_{#1}}
\newcommand{\udd}[3]{a_{#1}[{#2},{#3}]}

\newcommand{\inj}{f_{\textup{new}}}

\newcommand{\fp}{f^p}

\newcommand{\phiBun}[1]{\phi_{#1,\Ball{1}{a_p}{}}^{}}
\newcommand{\phiBunbis}{\phi_{\Ball{1}{a_p}{}}^{}}
\newcommand{\phiBdeux}{\phi_{\Ball{2}{a_p}{}}^{}}
\newcommand{\phiBdsu}[1]{\phi_{#1,\Ball{2}{a_p}{}\setminus\Ball{1}{a_p}{}}^{}}
\newcommand{\phiun}{\phi_{j,k,a_p}^{r=1}}
\newcommand{\psiun}{\psi_{j,k,a_p}^{r=1}}
\newcommand{\phideux}{\phi_{j,k,a_p}^{r=2}}

\newcommand{\psideux}{\psi_{j,k,a_p}^{r>1}}

\newcommand{\phitran}{\phi_{\mathit{tran}}}
\newcommand{\phirefl}{\phi_{\mathit{refl}}}
\newcommand{\phiuniq}{\phi_{\mathit{uniq}}}

\newcommand{\phiwf}{\phi_{\mathit{wf}}}
\newcommand{\dout}{d_{\mathit{out}}}
\newcommand{\T}[1]{\ensuremath{[\![#1]\!]}}
\newcommand{\Tp}[1]{\ensuremath{[\![#1]\!]}_{x_p}}
\newcommand{\Tbis}[1]{\ensuremath{[\![#1]\!]'}}
\newcommand{\Tpbis}[1]{\ensuremath{[\![#1]\!]'}_{x_p}}

\newcommand{\psitran}{\psi_{\mathit{tran}}}
\newcommand{\psirefl}{\psi_{\mathit{refl}}}
\newcommand{\psiwf}{\psi_{\mathit{wf}}}

\newcommand{\uge}{\mathsf{ge}}
\newcommand{\addge}[1]{#1_\uge}
\newcommand{\minusge}[1]{#1_{\setminus\uge}}
\newcommand{\AAge}{\addge{\AA}}
\newcommand{\AAminusge}{\minusge{\AA}}

\newcommand{\jmax}{j_{\max}}

\def \DD{\mathcal{D}}
\def \Z{\mathbb{Z}}
\newcommand{\Sign}{\Sigma}
\newcommand{\Zmod}[1]{\Z\,\mathrm{mod}\,#1}
\newcommand{\Tile}{\ensuremath{\textsc{Tiles}}}
\newcommand{\Hform}{\mathsf{H}}
\newcommand{\Vform}{\mathsf{V}}
\newcommand{\Wform}{\mathsf{W}}
\newcommand{\vpcomplete}{\vp_{\mathit{complete}}}
\newcommand{\vpprog}{\vp_{\mathit{progress}}}
\newcommand{\UH}[1]{X_{#1}}
\newcommand{\UV}[1]{Y_{#1}}
\newcommand{\grid}[1]{{\mathfrak{G}_{#1}}}

\newcommand{\Unarygrid}{\Unary_{\mathit{grid}}}
\newcommand{\gridcarrier}[1]{G_{#1}}
\newcommand{\vpgrid}{\vp_{\mathit{grid}}}
\newcommand{\query}[2]{\sem{#1}_{#2}}
\newcommand{\vpdo}{\vp_\DD}
\newcommand{\vpH}{\vp_H}
\newcommand{\vpHzz}{\vpH^{00}}
\newcommand{\vpHoz}{\vpH^{10}}
\newcommand{\vpHzo}{\vpH^{01}}
\newcommand{\vpHoo}{\vpH^{11}}
\newcommand{\vpHij}{\vpH^{ij}}
\newcommand{\vpV}{\vp_V}
\newcommand{\vpVzz}{\vpV^{00}}
\newcommand{\vpVoz}{\vpV^{10}}
\newcommand{\vpVzo}{\vpV^{01}}
\newcommand{\vpVoo}{\vpV^{11}}

\newcommand{\vpW}{\vp_W}
\newcommand{\vpWzz}{\vpW^{00}}
\newcommand{\vpWoz}{\vpW^{10}}
\newcommand{\vpWzo}{\vpW^{01}}
\newcommand{\vpWoo}{\vpW^{11}}

\newcommand{\vpcompleteee}{\vpcomplete^{\mathit{3\text{-}loc}}}
\newcommand{\vpproggg}{\vpprog^{\mathit{3\text{-}loc}}}
\newcommand{\vpgriddd}{\vpgrid^{\mathit{3\text{-}loc}}}
\newcommand{\vpDD}{\mathit{\vp_{\DD}}}

\newcommand{\vpcompletebis}{\vp'_{\mathit{complete}}}
\newcommand{\vpcompletee}{\vpcomplete^{\mathit{2\text{-}loc}}}
\newcommand{\vpprogg}{\vpprog^{\mathit{2\text{-}loc}}}
\newcommand{\vpgridd}{\mathit{\vp_{grid}^{2\text{-}loc}}}
\newcommand{\vpDDbis}{\mathit{\vp'_{\DD}}}

\newcommand{\gap}{0.25}  
\newcommand{\cut}{.4}   
\newcommand{\cont}{.45}   
\newcommand{\elsize}{1.5pt} 
\newcommand{\singletonradius}{.4cm} 
\newcommand{\connectionwidth}{3pt} 
\newcommand{\connectiononetwo}[2]{\draw[shift={(#1,#2)},line width=\connectionwidth,line cap = butt]  (-\gap,.5) arc [start angle=90, end angle=180, radius=.5-\gap];}
\newcommand{\xmax}{6}
\newcommand{\ymax}{\xmax}
\tikzstyle{styleclasseone}=[color=green!85!black, pattern={north east lines},pattern color=green!30]
\tikzstyle{styleclassetwo}=[color=purple!90,pattern={north west lines},pattern color=purple!20]
\newcommand{\roundedsquare}[3]{
\draw[#3]
				 (#1+\cut,#2-\gap) -- (#1+1-\cut,#2-\gap)
		 .. controls (#1+1-\cut+\cont,#2-\gap)  and (#1+1+\gap,#2+\cut-\cont) .. 
					(#1+1+\gap,#2+\cut) -- (#1+1+\gap,#2+1-\cut)
		.. controls (#1+1+\gap,#2+1-\cut+\cont) and (#1+1-\cut+\cont,#2+1+\gap) ..
					(#1+1-\cut,#2+1+\gap) -- (#1+\cut,#2+1+\gap)
		.. controls (#1+\cut-\cont,#2+1+\gap) and (#1-\gap,#2+1-\cut+\cont) ..
					 (#1-\gap,#2+1-\cut) -- (#1-\gap,#2+\cut)
		..controls (#1-\gap,#2+\cut-\cont) and (#1+\cut-\cont,#2-\gap) ..
					(#1+\cut,#2-\gap);
}
\newcommand{\classesingleton}[3]{\draw[shift={(#1,#2)},#3] (0,0) circle (\singletonradius);}
\newcommand{\drawelement}[2]{\fill (#1,#2) circle (\elsize);}

\newcommand{\s}[2]{\mathrel{_{#1}{\sim}{_{#2}}}}

\definecolor{bananamania}{rgb}{0.98, 0.91, 0.71}
\definecolor{darkseagreen}{rgb}{0.56, 0.74, 0.56}

\usepackage{tikz}
\usetikzlibrary{shapes}
\usetikzlibrary{patterns}
\usetikzlibrary{arrows.meta}
\usetikzlibrary{arrows}
\usetikzlibrary{decorations.pathmorphing}
\tikzstyle{data}=[rectangle split,rectangle split parts=2,draw,text centered, rectangle split part fill={blue!20,blue!20}]
\tikzstyle{dataredblue}=[rectangle split,rectangle split parts=2,draw,text centered, rectangle split part fill={red!30,blue!20}]
\tikzstyle{databluered}=[rectangle split,rectangle split parts=2,draw,text centered, rectangle split part fill={blue!20,red!30}]
\tikzstyle{datablueblue}=[rectangle split,rectangle split parts=2,draw,text centered, rectangle split part fill={blue!20,blue!20}]
\tikzstyle{dataredred}=[rectangle split,rectangle split
parts=2,draw,text centered, rectangle split part fill={red!30,red!30}]
\tikzstyle{datawhite}=[rectangle split,rectangle split parts=2,draw,text centered]
\tikzstyle{dataone}=[rectangle split,rectangle split parts=1,draw,text centered]
\tikzstyle{datablueyellow}=[rectangle split,rectangle split parts=2,draw,text centered, rectangle split part fill={blue!20,bananamania}]
\tikzstyle{datayellowblue}=[rectangle split,rectangle split parts=2,draw,text centered, rectangle split part fill={bananamania,blue!20}]
\tikzstyle{datagreengreen}=[rectangle split,rectangle split parts=2,draw,text centered, rectangle split part fill={darkseagreen!70,darkseagreen!70}]
\tikzstyle{databluegreen}=[rectangle split,rectangle split parts=2,draw,text centered, rectangle split part fill={blue!20,darkseagreen!70}]
\tikzstyle{datagreenred}=[rectangle split,rectangle split parts=2,draw,text centered, rectangle split part fill={darkseagreen!70,red!30}]
\tikzstyle{datagreenyellow}=[rectangle split,rectangle split parts=2,draw,text centered, rectangle split part fill={darkseagreen!70,bananamania}]
\tikzstyle{dataredyellow}=[rectangle split,rectangle split parts=2,draw,text centered, rectangle split part fill={red!30,bananamania}]
\tikzstyle{datayellowred}=[rectangle split,rectangle split parts=2,draw,text centered, rectangle split part fill={bananamania,red!30}]
\tikzstyle{datayellowgreen}=[rectangle split,rectangle split parts=2,draw,text centered, rectangle split part fill={bananamania,darkseagreen!70}]
\tikzstyle{datagraygray}=[rectangle split,rectangle split parts=2,draw,text centered, rectangle split part fill={gray!20,gray!20}]
\tikzstyle{datayellowyellow}=[rectangle split,rectangle split
parts=2,draw,text centered, rectangle split part
fill={bananamania,bananamania}]
\tikzstyle{datanone}=[rectangle split,rectangle split parts=2]

\usetikzlibrary{arrows}

\theoremstyle{plain}

\begin{document}

\title[On the Satisfiability of Local
  First-Order Logics with Data]{On the Satisfiability of Local\texorpdfstring{\\}{}
  First-Order Logics with Data }


\author[B. Bollig]{Benedikt Bollig}[a]
\author[A. Sangnier]{Arnaud Sangnier\lmcsorcid{0000-0002-6731-0340
}}[b]
\author[O. Stietel]{Olivier Stietel}[a,c]

\address{CNRS, LMF, ENS Paris-Saclay, Universit{\'e} Paris-Saclay,
  France}

\address{DIBRIS, Universit\`a di Genova, Italy}

\address{IRIF, Universit\'e  Paris Cit\'e, CNRS, France}


\begin{abstract}
We study first-order logic over unordered structures whose elements carry a finite number of data values from an infinite domain. Data values can be compared wrt.\ equality. As the satisfiability problem for this logic is undecidable in general, we introduce a family of local fragments. They restrict quantification to the neighbourhood of a given reference point that is bounded by some radius. Our first main result establishes decidability of the satisfiability problem for the local radius-1 fragment in presence of one  ``diagonal relation''. On the other hand, extending the radius leads to undecidability. In a second part, we provide the precise decidability and complexity landscape of the satisfiability problem for the existential fragments of local logic, which are parameterized by the number of data values carried by each element and the radius of the considered neighbourhoods. Altogether, we draw a landscape of formalisms that are suitable for the specification of systems with data and open up new avenues for future research.
\end{abstract}

\maketitle

\section{Introduction}

Data logics have been introduced to reason about structures whose elements are labeled with a value from an infinite alphabet \cite{Segoufin06}. The idea is to extend classic mathematical structures by a mapping that associates with every element of the universe a value from an infinite domain. When comparing data values only for equality, this view is equivalent to extending the underlying signature by a binary relation symbol whose interpretation is restricted to an equivalence relation. Among potential applications are XML reasoning or the specification of concurrent systems and distributed algorithms. Expressive decidable fragments include notably two-variable logics over data words and data trees \cite{BojanczykDMSS11,BojanczykMSS09}. The decidability frontier is fragile, though. Extensions to two data values, for example, quickly lead to an undecidable satisfiability problem. From a modeling point of view, those extensions still play an important role.
When specifying the \emph{input-output behavior} of distributed algorithms \cite{Fokkink2013,Lynch:1996}, processes get an input value and produce an output value, which requires two data values per process. In leader election or renaming algorithms, for instance, a process gets its unique identifier as input, and it should eventually output the identifier of a common leader (leader election) or a unique identifier from a restricted name space (renaming).

In this paper, we consider a natural extension of first-order logic over unordered structures whose elements carry two or more data values from an infinite domain. There are two major differences between most existing formalisms and our language. While previous data logics are usually interpreted over words or trees, we consider unordered structures (or multisets). When each element of such a structure represents a process, we therefore do not assume a particular processes architecture, but rather consider clouds of computing units.
Moreover, decidable data logics are usually limited to one value per element, which would not be sufficient to model an input-output relation. Hence, our models are algebraic structures consisting of a universe and functions assigning to each element a tuple of integers. We remark that, for many distributed algorithms, the precise data values are not relevant, but whether or not they are the same is.
Like \cite{BojanczykDMSS11,BojanczykMSS09}, we thus add binary relations that allow us to test if two data values are identical and, for example, whether all output values were already present in the collection of input values (as required for leader election).

The first fundamental question that arises is whether a given specification is consistent. This leads us to the satisfiability problem. While the general logic considered here turns out to be undecidable already in several restricted settings, we show that interesting fragments preserve decidability. The fragments we consider are \emph{local} logics in the sense that data values can only be compared within the neighborhood of a (quantified) reference element of the universe. In other words, comparisons of data values are restricted to elements whose distance to some reference point is bounded by a given radius.

We first consider a fragment over structures with two data values where the first value at the reference point can be compared with any second value in the neighborhood in terms of what we call the \emph{diagonal relation}. In particular, we do not allow the symmetrical relation. However, we do not restrict comparisons of first values with each other in a neighborhood, nor do we restrict comparisons of second values with each other. Note that adding only one diagonal relation still constitutes an extension of the (decidable) two-variable first-order logic with two equivalence relations \cite{Kieronski12,Kieronski05,KieronskiT09}: equivalence classes consist of those elements with the same first value, respectively, second value. In fact, our main technical contribution is a reduction to this two-variable logic. The reduction requires a careful relabelling of the underlying structures so as to be able to express the diagonal relation in terms of the two equivalence relations. In addition, the reduction takes care of the fact that our logic does not restrict the number of variables. We can actually count elements up to some threshold and express, for instance, that at most five processes crash (in the context of distributed algorithms). This is a priori not possible in two-variable logic. We also show that this local fragment has an undecidable satisfiability problem for any radius strictly greater than 1.

In a second part, we study orthogonal local fragments where global quantification is restricted to being existential (while quantification inside a local property is still unrestricted). We obtain decidability for (i) radius 1 and an arbitrary number of data values, and for (ii) radius 2 and two data values. In all cases, we provide tight complexity upper and lower bounds. These results mark the exact decidability frontier of the existential fragments: satisfiability is undecidable as soon as we consider radius 3 in presence of two data values, or radius 2 together with three data values.

\paragraph{Related Work.}

As already mentioned, data logics over word and tree structures were studied in \cite{BojanczykMSS09,BojanczykDMSS11}. In particular, the authors showed that two-variable first-order logic on words has a decidable satisfiability problem. Other types of data logics allow \emph{two} data values to be associated with an element \cite{Kieronski05,KieronskiT09}, though they do not assume a linearly ordered or tree-like universe. Again, satisfiability turned out to be decidable for the two-variable fragment of first-order logic.

Orthogonal extensions for multiple data values include
shuffle expressions for nested data \cite{BjorklundB07} and temporal logics \cite{KaraSZ10,DeckerHLT14}.
Other generalizations of data logics allow for an order on data values \cite{ManuelZ13,Tan14}.
The application of formal methods in the context of distributed algorithms is a rather recent but promising approach (cf.\ for a survey \cite{KonnovVW15}). A particular branch is the area of parameterized systems, which, rather than on data, focuses on the (unbounded) number of processes as the parameter \cite{Bloem15,Esparza14}. Other related work includes \cite{EmersonN03}, which considers temporal logics involving quantification over 
processes but without data, while \cite{AiswaryaBG18} introduces an (undecidable) variant of propositional dynamic logic that allows one to reason about 
totally ordered process identifiers in ring architectures. First-order logics for \emph{synthesizing} distributed algorithms were considered in \cite{BolligBR19,GrumbachW09}.
A counting extension of two-variable first-order logic over finite data words with one data value per position has been studied in \cite{Bednarczyk020}.

\paragraph{Outline.}

The paper is structured as follows.
In Section~\ref{sec:structures}, we recall important notions such as structures and first-order logic, and we introduce local first-order logic. Section~\ref{sec:two-data} presents decidability and undecidability results for our first fragment.
In Section~\ref{sec:existential}, we study the existential fragments and present their decidability frontier as well as complexity results. We conclude in Section~\ref{sec:conclusion}.

\noindent
This is an extended unified version of two conference papers presented at FSTTCS'21 and GandALF'22 \cite{BSS-FSTTCS-2021, BSS-GandALF-2022}. It provides a homogeneous presentation and full proofs. This work was partly supported by the project ANR FREDDA (ANR-17-CE40-0013).


%

\section{Structures and First-Order Logic}
\label{sec:structures}

\subsection{Data Structures}      
We define here the class of models we are interested in. It consists of sets of nodes containing data values with the assumption that each node is labeled by a set of predicates and carries the same number of values. We consider hence $\Unary$ a finite set of unary relation symbols (sometimes
called unary predicates) and an integer $\nbd \geq 0$. A \emph{$\nbd$-data structure} over $\Unary$ is a tuple 
$\AA=(A,(P_{\unary})_{\unary \in \Unary},\f{1},\ldots,\f{\nbd})$
(in the following, we simply write $(A,(P_{\unary}),\f{1},\ldots,\f{\nbd})$)
where $A$ is a nonempty finite set,
$P_\unary \subseteq A$ for all $\unary \in \Unary$, and
$\f{i}$s are mappings $A \to \N$.
Intuitively $A$ represents the set of nodes and $f_i(a)$ is the $i$-th data value carried by $a$ for each node $a \in A$.
For $X\subseteq A$, we let $\Valuessub{\AA}{X} = \{\f{i}(a) \mid a \in X, i\in\{1,\ldots,\nbd\}\}$.
The set of all $\nbd$-data structures over $\Unary$
is denoted by $\nData{\nbd}{\Unary}$.

While this representation is often very convenient to deal with
data values, a more standard way of
representing mathematical structures is in terms of binary
relations.
For every $(i,j) \in \{1,\ldots,\nbd\} \times \{1,\ldots,\nbd\}$, the mappings
$\f{1},\ldots,\f{\nbd}$ determine a binary relation
${\relsaaord{i}{j}{\AA}} \subseteq A \times A$
as follows:
$\relsaa{i}{j}{\AA}{a}{b}$ iff $\funct{i}(a) = \funct{j}(b)$.
We may omit the superscript $\AA$ if it is clear from the context
and if $\nbd=1$, as there will be only one relation, we way may write $\rel$ for $\relsaord{1}{1}$.

\subsection{First-Order Logic}
Let $\nbd\geq 0$ be an integer, $\Binary\subseteq \{\rels{i}{j}{}{} \mid
i,j \in \{1,\ldots,\nbd\}\}$ a set of binary relation symbols, which
determines the binary relation symbols at our disposal, and  $\Var = \{x,y,\ldots\}$ a
countably infinite set of variables. The set $\ndFO{\nbd}{\Unary,\Binary}$ of first-order formulas interpreted over $\nbd$-data structures
over $\Unary$ is inductively given by the grammar:
$$
\vp ::= \Pform{\unary}{x} \mid \rels{}{}{x}{y} \mid x=y \mid \vp \vee
\vp \mid \neg \vp \mid \exists x.\vp
$$
where $x$ and $y$ range over $\Var$, $\unary$ ranges over $\Unary$,
and $\rel \in \Gamma$.
We use standard abbreviations such as $\wedge$ for conjunction and
$\Rightarrow$ for implication.
We write $\phi(x_1,\ldots,x_k)$ to indicate that the free variables of $\phi$ are among
$x_1,\ldots,x_k$. We call $\phi$ a \emph{sentence} if it does not
contain free variables. Moreover, we use $\Binary_\nbd$ to represent the
set of binary relation symbols  $\{\rels{i}{j}{}{} \mid
i,j \in \{1,\ldots,\nbd\}\}$

For $\AA=(A,(P_{\unary}),\f{1},\ldots,\f{\nbd}) \in \nData{\nbd}{\Unary}$ and a formula $\phi\in\ndFO{\nbd}{\Unary,\Binary}$,
the satisfaction relation $\AA \models_I \phi$ is defined wrt.\
an interpretation function $I: \Var \to A$. The
purpose of $I$ is to assign an interpretation to every (free)
variable of $\phi$ so that $\phi$ can be assigned a truth value.
For $x \in \Var$ and $a \in A$, the interpretation function $\Intrepl{x}{a}$
maps $x$ to $a$ and coincides
with $I$ on all other variables.
We then define:
\begin{center}
\begin{tabular}{l}
$\AA \models_I \Pform{\unary}{x}$ if $I(x) \in P_{\unary}$ \\
$\AA \models_I \phi_1 \vee \phi_2$ if $\AA \models_I \phi_1$ or $\AA \models_I \phi_2$\\
$\AA \models_I \rels{i}{j}{x}{y}$ if $\relsaa{i}{j}{\AA}{I(x)}{I(y)}$\\
$\AA \models_I \neg \phi$ if $\AA \not\models_I \phi$\\
$\AA \models_I x = y$ if $I(x) = I(y)$\\
$\AA \models_I \exists x.\phi$ if there is $a \in A$ s.t. $\AA \models_{\Intrepl{x}{a}} \phi$
\end{tabular}
\end{center}

Finally, for a data structure $\AA=(A,(P_{\unary}),\f{1},\ldots,\f{\nbd})$,  a formula $\phi(x_1,\ldots,x_k)$ and $a_1,\ldots,a_k\in A$,
we write $\AA\models\phi(a_1,\ldots a_k)$ if there exists an interpretation function $I$ such that $\AA\models_{I[x_1/a_1]\ldots[x_k/a_k]} \phi$.  In particular, for a sentence $\phi$, we write $\AA\models\phi$ if there exists an interpretation function $I$ such that $\AA \models_I  \phi$.

\begin{example}\label{ex:leader-election}
Assume a unary predicate $\mathrm{leader} \in \Unary$.
The following formula from $\ndFO{2}{\Unary,\Binary_2}$ expresses
correctness of a leader-election algorithm: (i)~there is a unique
process that has been elected leader, and (ii)~all processes agree,
in terms of their output values (their second data), on
the identity (the first data) of the leader: 
$ \exists x. (\mathrm{leader}(x) \et \forall y. \big(\mathrm{leader}(y)
\Rightarrow y=x)\big) \et \forall y. \exists x. (\mathrm{leader}(x)
\et \rels{1}{2}{x}{y})$. We assume here that for each element
represents a process,the 
first data is used to characterize its input value and the 
second data its output value.
\end{example}

Note that every choice of $\Binary$ gives rise to a particular logic, whose
formulas are interpreted over data structures over $\Unary$. We will focus on the
satisfiability problem for these logics.
Let $\mathcal{F}$ denote a generic class of first-order formulas, parameterized
by $D$, $\Unary$ and $\Binary$.
In particular, for $\mathcal{F} = \ndFOr$,
we have that $\mathcal{F}[\nbd,\Unary,\Binary]$ is the class
$\ndFO{\nbd}{\Unary,\Binary}$. The satisfiability problem for $\mathcal{F}$ wrt.\ data structures
is defined as follows:

\begin{center}
\begin{decproblem}
	\problemtitle{\nDataSat{\mathcal{F}}{\nbd,\Gamma}}
  \probleminput{A finite set $\Unary$ and a sentence $\vp \in \mathcal{F}[\nbd,\Unary,\Binary]$.}
	\problemquestion{Is there $\AA \in \nData{\nbd}{\Unary}$ such that $\AA \models \vp$\,?}
\end{decproblem}
\end{center}

The following negative result (see \cite[Theorem~1]{Janiczak-Undecidability-fm53}) calls for restrictions of the general logic.
\begin{thmC}[\cite{Janiczak-Undecidability-fm53}]\label{thm:undec-general}
	The problem $\nDataSat{\ndFOr}{2,\{\relsaord{1}{1},\relsaord{2}{2}\}}$ is undecidable, even when we require that $\Unary = \emptyset$.
  \end{thmC}

  \subsection{Local First-Order Logic}
  We are interested in logics that combine the advantages of the logics considered so far,
while preserving decidability. With this in mind, we will study \emph{local} logics,
where the scope of quantification
is restricted to the view (or neighborhood) of a given element.

The view of an element $a$ includes all elements whose distance to $a$ is bounded by a given radius.
It is formalized using the notion of a Gaifman graph (for an
introduction, see~\cite{Libkin04}). In fact, we use a variant that is
suitable for our setting and that we call \emph{data graph}. Fix sets
$\Unary$ and $\Binary$. Given a data structure  $\AA=(A,(P_{\unary}),\f{1},\ldots,\f{\nbd}) \in \nData{\nbd}{\Unary}$, we define its \emph{data graph} $\gaifmanish{\AA}=(\Vertex{\gaifmanish{\AA}},\Edge{\gaifmanish{\AA}})$ with set of vertices $\Vertex{\gaifmanish{\AA}} = A \times\{1,\ldots,\nbd\}$ and set of edges
$\Edge{\gaifmanish{\AA}} = \{ ((a,i),(b,j)) \in
\Vertex{\gaifmanish{\AA}} \times \Vertex{\gaifmanish{\AA}} \mid a=b
\mbox{ and } i \neq j, \mbox{ or } \rels{i}{j}{}{} \in \Gamma \mbox{
  and } \rels{i}{j}{a}{b} \}$. Note that the definition of the edges
of this graph depends of the set $\Gamma$, however to simplify the
figures we do not always write it and make sure it will always be clear
from the context. Figures \ref{fig:gaifman1}(A) and \ref{fig:gaifman2}(A) provide  examples of the graph
$\gaifmanish{\AA}$ for two data structure with $2$ data values. In the
first case, the set of considered binary relation symbols is
$\{\rels{1}{1}{}{},\rels{2}{2}{}{},\rels{1}{2}{}{} \}$ whereas in the
second case it is
$\Binary_2=\{\rels{1}{1}{}{},\rels{2}{2}{}{},\rels{1}{2}{}{},\rels{2}{1}{}{}\}$. We
point out that the graph $\gaifmanish{\AA}$  is directed and indeed we
see in the first Figure \ref{fig:gaifman1}(A) that since
$\rels{2}{1}{}{}$ is not considered, we have some unidirectional edges which we depict dashed. 


\newcommand{\selfconnectionright}[1]{\draw[<->, line width=0.7pt] (#1.one east) .. controls +(.4,0) and +(.4,0) .. (#1.two east);}
\newcommand{\selfconnectionleft}[1]{\draw[<->, line width=0.7pt]
  (#1.one west) .. controls +(-.4,0) and +(-.4,0) .. (#1.two west);}

\begin{figure}[htbp]
\centering
{
	\begin{subfigure}[b]{0.48\textwidth}
\begin{tikzpicture}[node distance=2cm]
	\node [data, label=below:$a$]                (A)    {1 \nodepart{second} 2};
	\node [dataredblue, above right of=A]    (B)    {3 \nodepart{second} 1};
	\node [dataredblue, below right of=B]    (C)    {3 \nodepart{second} 2};
	\node [databluered, above left of=A]    (D)    {1 \nodepart{second} 4};
	\node [dataredred, above right of=C]    (E)    {3 \nodepart{second} 5};

	\draw[line width=0.7pt,->, dashed] (A.north east) .. controls +(.5,.25) and +(0,-.75).. (B.two south);
	\draw[line width=0.7pt,->, dashed] (D.north east).. controls +(1,.5) and +(-1,-.5).. (B.south west);
	\draw[line width=0.7pt, <->] (B.one east) .. controls +(1,0) and +(-1,0).. (C.one west);
	\draw[line width=0.7pt, <->] (D.one east) .. controls +(1,0) and +(-1,0).. (A.one west);
	\draw[line width=0.7pt,<->] (B.north east) .. controls +(.5,.25) and +(-.5,.25).. (E.north west);
	\draw[line width=0.7pt,<->,in=45] (C.one north) .. controls +(0,.5) and +(-1 ,0).. (E.one west);
	\draw[line width=0.7pt,<->] (A.south east) .. controls +(1,-.5) and +(-.5,0).. (C.two west);

	\selfconnectionright{A};
	\selfconnectionleft{B};
	\selfconnectionright{C};
	\selfconnectionleft{D};
	\selfconnectionright{E};
  \end{tikzpicture}
  \caption{A data structure $\AA$ and  $\gaifmanish{\AA}$ \\when $\Binary=\{\rels{1}{1}{}{},\rels{2}{2}{}{},\rels{1}{2}{}{} \}$\label{fig:gaifman1-a}}

	\end{subfigure}}
	~~\vrule~~~
{
	\begin{subfigure}[b]{0.32\textwidth}

\begin{tikzpicture}[node distance=2cm]
	\node [data, label=below:$a$]                (A)    {1 \nodepart{second} 2};
	\node [dataredblue, above right of=A]    (B)    {8 \nodepart{second} 1};
	\node [dataredblue, below right of=B]    (C)    {9 \nodepart{second} 2};
	\node [databluered, above left of=A]    (D)    {1 \nodepart{second} 6};
	\node [above right of=C]    (E)    {};

	\draw[line width=0.7pt,->, dashed] (A.north east) .. controls +(.5,.25) and +(0,-.75).. (B.two south);
	\draw[line width=0.7pt,->, dashed] (D.north east).. controls +(1,.5) and +(-1,-.5).. (B.south west);
	\draw[line width=0.7pt, <->] (D.one east) .. controls +(1,0) and +(-1,0).. (A.one west);
	\draw[line width=0.7pt,<->] (A.south east) .. controls +(1,-.5) and +(-.5,0).. (C.two west);

	\selfconnectionright{A};
	\selfconnectionleft{B};
	\selfconnectionright{C};
	\selfconnectionleft{D};

\end{tikzpicture}
\caption{$\vprojr{\AA}{a}{1}$: the $1$ view of $a$}
\label{fig:gaifman1-b}
	\end{subfigure}
	}
\caption{ \label{fig:gaifman1}}
\end{figure}

\begin{figure}[htbp]
\centering
	\begin{subfigure}[b]{0.48\textwidth}
\begin{tikzpicture}[node distance=2cm]
	\node [data, label=below left:$a$]                (A)    {1
      \nodepart{two} 2 };
	\node [data, above left of=A,xshift=-1em,label=below:$b$]    (B)    {1 \nodepart{second} 3};
	\node [data, above right of=A,xshift=1em,label=below right:$c$]    (C)    {3 \nodepart{second} 2};
	\node [dataredred, below left of=A,label=below:$d$]    (D)    {5 \nodepart{second} 6};
	\node [dataredred, above right of=B,xshift=1em, label=below:$e$]
    (E)    {4 \nodepart{second} 3};
    \node [data, below right of=A, label=below:$f$]    (F)    {2 \nodepart{second} 7};

	\draw[line width=0.7pt,<->] (A.one north) .. controls +(0,.5) and
    +(.5,0).. (B.one east);
	\draw[line width=0.7pt,<->] (B.two east) .. controls +(2,-0.5) and
    +(-2,.5).. (C.one west);
    \draw[line width=0.7pt,<->] (E.two east) .. controls +(0,0) and
    +(0,0.5).. (C.one north);
    \draw[line width=0.7pt,<->] (E.south west) .. controls +(0,0) and
    +(0.5,.2).. (B.two east);
	\draw[line width=0.7pt,<->] (A.south east) .. controls +(1,-.5)
    and +(0,0).. (C.south west);
    \draw[line width=0.7pt,<->] (A.south) .. controls +(0,-0.5) and
    +(0,0).. (F.one west);
    \draw[line width=0.7pt,<->] (F.north) .. controls +(0,0) and
    +(0,0).. (C.south);
	\selfconnectionright{A};
	\selfconnectionleft{B};
	\selfconnectionright{C};
	\selfconnectionleft{D};
	\selfconnectionleft{E};
	\selfconnectionright{F};

  \end{tikzpicture}
\caption{A data structure $\AA$ and  $\gaifmanish{\AA}$ \\when $\Binary=\Binary_2$.}
\label{fig:gaifman2-a}
	\end{subfigure}
	\unskip\ \vrule\ \hspace{2em}
	\begin{subfigure}[b]{0.32\textwidth}
\begin{tikzpicture}[node distance=2cm]
		\node [data, label=below left:$a$]                (A)    {1
      \nodepart{two} 2 };
	\node [data, above left of=A,xshift=-1em,label=below:$b$]    (B)    {1 \nodepart{second} 3};
	\node [data, above right of=A,xshift=1em,label=below right:$c$]    (C)    {3 \nodepart{second} 2};
    \node [data, below right of=A, label=below:$f$]    (F)    {2 \nodepart{second} 7};

  \end{tikzpicture}
\caption{$\vprojr{\AA}{a}{2}$: the $2$ view of $a$}
\label{fig:gaifman2-b}
\end{subfigure}
\caption{
\label{fig:gaifman2}}
\end{figure}

We then define the distance $\distaa{(a,i)}{(b,j)}{\AA} \in \N \cup
\{\infty\}$ between two elements $(a,i)$ and $(b,j)$ from $A
\times\{1,\ldots,\nbd\}$ as the length of the shortest \emph{directed}
path from $(a,i)$ to $(b,j)$ in $\gaifmanish{\AA}$. In fact, as the
graph is directed, the distance function might not be symmetric. For $a \in A$ and $r \in \N$, the \emph{radius-$r$-ball around} $a$ is
the set $\Ball{r}{a}{\AA} = \{ (b,j)\in\Vertex{\gaifmanish{\AA}} \mid
\distaa{(a,i)}{(b,j)}{\AA}\leq r $ for some $i \in
\{1,\ldots,\nbd\}\}$. This ball contains the elements of
$\Vertex{\gaifmanish{\AA}}$ that can be reached from
$(a,1),\ldots,(a,\nbd)$ through a directed path of length at most $r$.
On Figure~\ref{fig:gaifman1}(A) the blue nodes represent
$\Ball{1}{a}{\AA}$ and on Figure~\ref{fig:gaifman2}(A) they represent
$\Ball{2}{a}{\AA}$. 

Let $\inj: A \times \{1,\ldots,\nbd\} \to \N \setminus \Values{\AA}$
be an injective mapping. The \emph{$r$-view of $a$ in $\AA$} is the
structure $\vprojr{\AA}{a}{r} =
(A',(P'_{\unary}),\f{1}',\ldots,\f{n}') \in \nData{\nbd}{\Unary}$. Its
universe is $A' = \{b \in A \mid (b,i) \in \Ball{r}{a}{\AA}$ for some
$i \in \{1,\ldots,\nbd\}\}$. Moreover, $\funct{i}'(b)= \funct{i}(b)$ if $(b,i)
\in\Ball{r}{a}{\AA}$, and $\funct{i}'(b)= \inj((b,i))$
otherwise. Finally, $P_{\unary}'$ is the restriction of $P_{\unary}$
to $A'$.  To illustrate this definition, we use again
Figures~\ref{fig:gaifman1} and \ref{fig:gaifman2}. On Figure
\ref{fig:gaifman1}(B), the structure $\vprojr{\AA}{a}{1}$ is given by the four
elements that contain at least one blue node.
However, the values of the red nodes have to be replaced by
pairwise distinct fresh values not contained in $\{1,\ldots,5\}$.
Note that the precise values do not matter. On Figure~\ref{fig:gaifman2}(B),
the structure $\vprojr{\AA}{a}{2}$ is depicted, in that case there is
no need of fresh value.

We are now ready to present the logic $\rndFO{\nbd}{\Unary,\Binary}{r}$, where $r \in \N$, interpreted over structures from
$\nData{\nbd}{\Unary}$. It is given by the grammar
\begin{align*}
	\vp ~&::=~ \locformr{\psi}{x}{r} \;\mid\; x=y \;\mid\; \exists x.\vp \;\mid\;  \vp \vee \vp \;\mid\;  \neg \vp
\end{align*}
where $\psi$ is a formula from $\ndFO{\nbd}{\Unary,\Binary}$
with (at most) one free variable $x$.This logic uses the \emph{local
modality} $\locformr{\psi}{x}{r}$ to specify that the formula $\psi$
should be interpreted over the $r$-view of the element associated to
the variable $x$.
For $\AA \in \nData{\nbd}{\Unary}$ and an interpretation function $I$, we have
$\AA \models_I \locformr{\psi}{x}{r}$ iff $\vprojr{\AA}{I(x)}{r} \models_I \psi$.

\begin{example}
Let us illustrate what can be specified by formulas in
$\rndFO{2}{\Unary,\Binary_2}{1}$. We can rewrite the formula from Example~\ref{ex:leader-election} so that
it falls into our fragment as follows:
$\exists x. (\locformr{\mathrm{leader}(x)}{x}{1} \et \forall
y. \linebreak[0](\locformr{\mathrm{leader}(y)}{y}{1} \Rightarrow x=y))
\et \forall
y. \linebreak[0] \locformr{\exists x. \mathrm{leader}(x) \et \rels{2}{1}{y}{x}
}{y}{1} $. Indeed note that in the formula of Example
~\ref{ex:leader-election}, when we write $\forall y. \exists x. (\mathrm{leader}(x)
\et \rels{1}{2}{x}{y})$, then for all elements $a$, the element labelled
with $\mathrm{leader}$ has to appear in the $1$-view of $a$ since they respectively
share the same value on  their second and fist value.
The next formula specifies an algorithm in which all processes suggest a value and choose a new value among those that have been suggested at least twice:
  $\forall x.\locformr{\exists
    y.\exists z. y \neq z \et \rels{2}{1}{x}{y} \et \rels{2}{1}{x}{z} }{x}{1} $. We can also specify partial renaming, i.e., two output values agree only if their input values are the same:
  $\forall x.\locformr{\forall y.(\rels{2}{2}{x}{y}\donc\rels{1}{1}{x}{y}}{x}{1}$.
Conversely, $\forall x.\locformr{\forall
  y.(\rels{1}{1}{x}{y}\donc\rels{2}{2}{x}{y}}{x}{1}$ specifies partial
fusion of equivalences classes.
\end{example}

\subsection{Previously Known Results}
We list in this subsection previously known results related to the
satisfiability problem we are studying and that we shall rely on.

\subsubsection{0 and 1 Data Value}

First, note that
formulas in $\ndFO{0}{\Unary,\emptyset}$ (i.e. where no data is considered)
correspond to first order logic formulas with a set of predicates and
equality test as a unique relation. As mentioned in  Chapter 6.2.1 of
\cite{borger-classical-springer97}, these formulas belong to the
\emph{L\"owenheim class with equality} also called as the relational
monadic formulas, and their satisfiability problem is in
\textsc{NEXP}. Furthermore, thanks to \cite{etessami-first-ic02} (Theorem 11), we know  that this latter problem is \textsc{NEXP}-hard even
if one considers formulas which use only two variables.

\begin{thmC}[\cite{borger-classical-springer97,etessami-first-ic02}]\label{thm:0fo}
 $\nDataSat{\ndFOr}{0,\emptyset}$ is \textsc{NEXP}-complete.
\end{thmC}

In \cite{Mundhenk09}, the authors study the satisfiability problem for
Hybrid logic over Kripke structures where the transition relation is
an equivalence relation, and they show that it is
\textsc{N2EXP}-complete. Furthermore in  \cite{Fitting12}, it is shown
that Hybrid logic can be translated to first-order logic in
polynomial time and this holds as well for  the converse
translation. Since $1$-data structures can be interpreted as Kripke
structures with one equivalence relation, altogether this allows us to
obtain the following preliminary result about the satisfiability
problem of $\ndFO{1}{\Unary,\{\relsaord{1}{1}\}}$.

\begin{thmC}[\cite{Mundhenk09,Fitting12}]\label{thm:1fo}
 $\nDataSat{\ndFOr}{1,\{\relsaord{1}{1}\}}$ is \textsc{N2EXP}-complete.
\end{thmC}

 \subsubsection{A Normal Form}

  When $\Binary = \emptyset$,
satisfiability of monadic first-order logic is decidable~\cite[Corollary 6.2.2]{Gurevich97}
and the logic actually has a useful normal form.
Let $\phi(x_1,\ldots,x_n,y) \in \ndFO{\nbd}{\Unary,\emptyset}$ (note
that since the set of binary relation symbols is empty, the number $\nbd$ of
data values does not matter here)  and $k \ge 1$
be a natural number. We use $\exists^{\ge k}y.
\phi(x_1,\ldots,x_n,y)$ as an abbreviation for
$
\exists y_1 \ldots \exists y_k.
\bigwedge_{1 \le i < j \le k} \neg{(y_i = y_j)} \wedge \bigwedge_{1 \le i \le k} \phi(x_1,\linebreak[0]\ldots,x_n,y_i)$.
Thus, $\exists^{\ge k}y.\phi$ says that there are at least $k$
distinct elements $y$ that verify $\phi$. 
We call a formula of the form $\exists^{\ge k}y.\phi$ a \emph{threshold formula}.
We also use
$\exists^{= k} y. \phi$ as an abbreviation for
$\exists^{\ge k} y. \phi \wedge \neg \exists^{\ge k+1} y. \phi$.

When $\Binary = \emptyset$, the out-degree of every element is $0$ so that,
over this particular signature, we deal with structures of bounded
degree.
The following lemma will turn out to be useful. It is due to Hanf's locality theorem \cite{Han65,Libkin04}
for structures of bounded degree (cf.\ \cite[Theorem~2.4]{BolligK12}).

\begin{lemma}\label{lem:threshold}
Every formula from $\ndFO{\nbd}{\Unary,\emptyset}$ with one free variable $x$
is effectively equivalent to a Boolean combination of
formulas of the form $\unary(x)$ with $\unary \in \Unary$ and
threshold formulas of the form
$\exists^{\ge k} y. \phi_U(y)$ where $U \subseteq \Unary$ and
$\phi_U(y) = \bigwedge_{\unary \in U} \unary(y) \wedge \bigwedge_{\unary \in \Unary \setminus U} \neg \unary(y)$.
\end{lemma}

\subsubsection{Extended Two-Variable First-Order Logic}

One way to skirt the undecidability result of Theorem
\ref{thm:undec-general} is to restrict the number of allowed variables
in the formulas. We denote by $\ndFOvar{\nbd}{\Unary,\Binary}{2}$ the
two-variable fragment of $\ndFO{\nbd}{\Unary,\Binary}$, i.e. the
set of formulas in $\ndFO{\nbd}{\Unary,\Binary}$ which use only two variables (usually $x$ and $y$). In a two-variable
formula, however, each of the two variables can be used arbitrarily
often. The satisfiability problem of two-variable logic over arbitrary finite
structures with two equivalence relations is decidable \cite[Theorem~15]{KieronskiT09}.
By a straightforward reduction to this problem, we obtain :

\begin{thmC}[\cite{KieronskiT09}]\label{thm:twoFO}
  The problem $\nDataSat{\ndFOrvar{2}}{2,\{\relsaord{1}{1},\relsaord{2}{2}\}}$ is decidable.
\end{thmC}

Actually, this result can be generalized to \emph{extended}
two-variable first-order logic. A formula belongs to $\extndFOvar{\nbd}{\Unary,\Binary}{2}$ if it is of the form
$\vp\et\psi$ where $\vp\in\ndFO{\nbd}{\Unary,\emptyset}$ and
$\psi\in\ndFOvar{\nbd}{\Unary,\Binary}{2}$. To obtain the next
result, the idea consists in first translating the formula
$\vp\in \ndFO{\nbd}{\Unary, \emptyset}$ to a two-variable formula thanks to new unary predicates.

\begin{proposition}\label{prop:exttwoFO}
The problem $\nDataSat{\extndFOrvar{2}}{2,\{\relsaord{1}{1},\relsaord{2}{2}\}}$ is decidable.
\end{proposition}

\begin{proof}
  We first show that one can reduce the first-order part with $\Binary = \emptyset$
to a two-variable formula. Let $\vp$ be a sentence in
$\ndFO{\nbd}{\Unary, \emptyset}$. We show we can effectively construct $\vp' \in \ndFOvar{\nbd}{\Unary',\emptyset}{2}$ with $\Unary\subseteq\Unary'$ such that $\vp$ is satisfiable iff $\vp'$ is satisfiable.
		Furthermore, if a structure $\AA$ satisfies $\vp$, then we can add an interpretation of the predicates in $\Unary'\setminus\Unary$ to $\AA$ to get a model for $\vp'$.
		Conversely, if a structure $\AA'$ satisfies $\vp'$, then
        forgetting the interpretation of the predicates in
        $\Unary'\setminus\Unary$ in $\AA'$ give us a model for $\vp$.

        We apply Lemma~\ref{lem:threshold} to $\vp$ and then obtain $\vp''$.
As there is no free variable in $\vp$, the formula $\vp''$ is a boolean combination of formulas of the form $\exists^{\ge k} y. \phi_U(y)$ where $U \subseteq \Unary$.
Let $M$ be the maximal such $k$ (if there is no threshold formula, $\vp''$ is either $\mathit{true}$ or $\mathit{false}$).
We define a set of unary predicates $\Unarycl{M}=\{\unarycounter{i}\mid 1\leq i \leq M\}$
and let $\Unary' = \Unary \cup \Unarycl{M}$. Intuitively the unary
predicates in $\Unarycl{M}$ will be used to count up to $M$ the
elements labeled by the same predicates in $\Unary$ associating each element of the
    domain with a unique numerical index.
The following formulas will specify the meaning of the elements of $\Unarycl{M}$. First, let $\phi_{same}(x,y)=\bigwedge_{\sigma \in \Unary} \sigma(x) \leftrightarrow \sigma(y)$. With this, we define:
$$
\begin{array}{rcl}
    \vpetai{1} &:= & \displaystyle \forall x.\bigvee_{i \in [1,M]} \big(\unarycounter{i}(x) \wedge \bigwedge_{j \in [1,M]\setminus \{i\}} \neg \unarycounter{j}(x)\big) \\
  &&\\
  \vpetai{2} &:= &\displaystyle \forall x.\bigwedge_{i \in [1,M-1]} \big(\unarycounter{i}(x) \rightarrow \neg \exists y. (x \neq y \wedge \phi_{same}(x,y)   \wedge \unarycounter{i}(y))\big) \\
                &&\\
    \vpetai{3} &:= &\displaystyle \forall x.\bigwedge_{i \in [2,M]} \big(\unarycounter{i}(x) \rightarrow (\exists y. \phi_{same}(x,y) \wedge \unarycounter{i-1}(y))\big)\\
\end{array}
$$
We then denote $\vpeta :=  \vpetai{1} \et \vpetai{2} \et
\vpetai{3}\in\ndFOvar{\nbd}{\Unary',\emptyset}{2}$. We provide an
informal explanation of the previous formulas:
\begin{itemize}
\item the formula $\vpetai{1}$ ensures that each element is labeled by
  a single label from $\Unarycl{M}$;
\item the formula $\vpetai{2}$ ensures that there are no two elements 
  labeled by the same predicates in $\Unary$ which are labeled
  by the same element $\unarycounter{i}$ with $i \leq M-1$;
\item  the formula $\vpetai{3}$ ensures that if an element is labeled
  with $\unarycounter{i}$ with $i \geq 2$ then there must exist
  another element with the same labels in $\Unary$ labeled by $\unarycounter{i-1}$.
\end{itemize}
Then, for a model $\AA\in\nData{\nbd}{\Unary'}$ of $\vpeta$ with carrier set $A$, an element $a\in A$, and an integer $1\leq i \leq M$, we have that $a \in P_{\unarycounter{i}}$ iff $|\{b\in A \mid$ for all $\unary\in\Unary$, $a \in P_\unary$ iff $b \in P_\unary\}|\geq i$.
Then in $\vp''$, we replace all threshold formulas $\exists^{\ge k}
y. \phi_U(y)$  with $\exists y.\vp_U(y)\et\unarycounter{k}(y)$ in
order to obtain $\vp'''\in \ndFOvar{\nbd}{\Unary',\emptyset}{2}$.
Finally we take $\vp'$ as $\vp'''\et\vpeta$.

We can now proceed with the proof of the Proposition. Let $\vp\et\psi$ be a sentence such that $\vp\in\ndFO{\nbd}{\Unary,\emptyset}$ and $\psi\in\ndFOvar{\nbd}{\Unary,\Binary}{2}$.
	We determine $\Unary' \supseteq \Unary$ and $\vp'$ in $\ndFOvar{\nbd}{\Unary,\Binary}{2}$ according to what we said above. Then, by Theorem~\ref{thm:twoFO}, it only remains to show that $\vp\et\psi$ is satisfiable iff $\vp'\et\psi$ is satisfiable.
	
Suppose there is $\AA \in \nData{\nbd}{\Unary}$ such that $\AA \models \vp\et\psi$. According to what we said, we can add propositions from
$\Unary'\setminus\Unary$ to $\AA$ to get a data structure $\AA'$ such that $\AA' \models \vp'$.
As $\psi$ does not speak about propositions in $\Unary'\setminus\Unary$, we have
$\AA' \models \psi$ and, therefore, $\AA' \models \vp'\et\psi$. 

Conversely, let $\AA' \in \nData{\nbd}{\Unary'}$ such that $\AA' \models \vp'\et\psi$.
Then, ``forgetting''
in $\AA'$ all labels in $\Unary'\setminus\Unary$ yields a structure $\AA$ such that $\AA \models \vp$.
As we still have $\AA \models \psi$, we conclude $\AA \models \vp\et\psi$.
\end{proof}

\section{Dealing with Two Data Values}
\label{sec:two-data}
\newcommand{\USet}{U}
\newcommand{\BSet}{R}
\newcommand{\eqx}{\mathsf{eq}}
\newcommand{\CSign}{\mathbb{C}}
\newcommand{\CountC}{\mathtt{C}}
\newcommand{\newcc}[4]{\Lbag #1, #2, #4\Rbag}
\newcommand{\myEnv}[6]{\mathtt{Env}_{#1,#2,#3}(#4,#5,#6)}
\newcommand{\myCountConst}[3]{\mathtt{C}_{#3}}
\newcommand{\phiint}{\phi_{\mathit{same}}^{\mathit{int}}}
\newcommand{\Ueq}{\Theta}
\newcommand{\Ueqed}{\Ueq_{\ediag}}
\newcommand{\addediag}[1]{#1 + \ediag}
\newcommand{\removeediag}[1]{#1 \setminus \ediag}
\newcommand{\extSigma}{\Theta}
\newcommand{\phied}{\xi_\ediag}
\newcommand{\philabels}[1]{\phi_{#1}}
\renewcommand{\alpha}{\upalpha}
\renewcommand{\beta}{\upbeta}
\renewcommand{\gamma}{\upgamma}

\subsection{Decidability with Radius One and One Diagonal Relation} 
We will show in this section that, when considering structures with two data values,  $\DataSat{\rndFOr{1},2,\linebreak[4]\{\relsaord{1}{1},\relsaord{2}{2},\relsaord{1}{2}\}}$ (or, symmetrically, $\nDataSat{\rndFOr{1}}{2,\{\relsaord{1}{1},\linebreak[0]\relsaord{2}{2},\relsaord{2}{1}\}}$), i.e. the local fragment of radius $1$ with a set of binary symbols restricted to  $\{\relsaord{1}{1},\relsaord{2}{2},\relsaord{1}{2}\}$ is decidable. To this end, we will give a
reduction to the problem
$\nDataSat{\extndFOrvar{2}}{2,\{\relsaord{1}{1},\relsaord{2}{2}\}}$
and rely on Proposition \ref{prop:exttwoFO}. Before we give the formal details of this reduction, we present here the main
steps it is  based on:
\begin{itemize}
\item We transform a formula $\phi \in \rndFO{2}{\Unary,\Binary}{1}$,
  in a formula $\chi \in \ndFO{2}{\Unary\cup \{\eqx\} \cup \myCountConst{\Unary}{\Binary}{M},\emptyset}$
  such that $\phi$ is satisfiable iff $\chi$ has a well-typed
  model. In a well-typed model, each element $a$ is labelled with some
  unary predicates in   $\myCountConst{\Unary}{\Binary}{M}$. These
  predicates count (up to a constant $M$) for each subset $R$ of
  binary relations in $\Binary$ the number of elements sharing the
  same unary predicates (in $\Unary$) and in relation with $a$ for all
  relations in $R$. For  instance, such a labelling can indicate that there are
   two elements having the same first data value as $a$. To perform
   this translation, we rely on the form of formulas in
   $\rndFO{2}{\Unary,\Binary}{1}$ and use Lemma \ref{lem:threshold}
   ($\eqx$ is a special predicate to detect elements with the same two
   values).
 \item We then  need to ensure that a given model  is well-typed. We
   can achieve this by adding predicates to count up to a constant $M$
   (we label a single element with $1$, a single element with $2$,
   etc)  and then use a formula with two variables to check the correct
   labelling. We could indeed build a formula in $
   \ndFOvar{2}{\Unary \cup
     \myCountConst{\Unary}{\Binary}{M} \cup \{\eqx\} \cup
     \Lambda_M,\Binary}{2}$ (where $\Lambda_M$ are extra  unary
   predicates to count elements with the same properties) to verify  if
   a model is a well typed. The problem is that to rely on the result
   of Proposition \ref{prop:exttwoFO}, we need the formula not to use
   the diagonal relation $\relsaord{1}{2}$.
 \item To get rid of the diagonal relations, we show that we can add
   for each value $v$ in our model, an extra element that we label with a
   special predicate $\ediag$ and that has the value $v$ in its two
   fields. We then demonstrate how to ensure diagonal relations by
   making a detour via these extra elements.  We call this extended
   model well-diagonalized and provide a two
   variable formula in $\ndFOvar{2}{\Unary \cup
     \myCountConst{\Unary}{\Binary}{M} \cup \{\eqx,\ediag\} \cup
     \Lambda_M,\Binary_{df}}{2} $ to check that a model is well
   diagonalized.
\end{itemize}

We now move to the formal reduction.
Henceforth, we fix a finite set $\Sigma$ as well
as $\Gamma = \{\relsaord{1}{1},\relsaord{2}{2},\relsaord{1}{2}\}$ and
we let the \emph{diagonal-free set} be
$\Binary_{df}=\{\relsaord{1}{1},\relsaord{2}{2}\}$.
Moreover, we let $\extSigma$ range over arbitrary finite sets such that
$\Sigma \subseteq \extSigma$ and $\Theta \cap \{\eqx,\ediag\} = \emptyset$,
where $\eqx$ and $\ediag$ are special unary symbols that are introduced below.

We start with some crucial notion.
Suppose $\Binary' \subseteq \Binary$ (which will later be instantiated by either
$\Binary_{df}$ or $\Binary$).
Consider a data structure $\AA=(A,(P_{\unary}),\ifunct,\ofunct) \in \nData{2}{\extSigma}$ with $\Sigma \subseteq \extSigma$.
Given $U \subseteq \Sigma$ and a nonempty set $R \subseteq \Gamma'$,
the \emph{environment} of $a \in A$ is defined as
$$
\myEnv{\AA}{\Sigma}{\Binary'}{a}{\USet}{\BSet}
 = \bigl\{b \in A \mid \USet=\{\unary \in \Unary \mid b \in  P_{\unary}\} \mbox{ and }
 \BSet=\{\rels{i}{j}{}{} \in \Binary' \mid \relsaa{i}{j}{\AA}{a}{b} \}\bigr\}
 $$
Thus, it contains the elements
that carry exactly the labels from $U$ (relative to $\Sigma$) and to which $a$ is related precisely in terms of the relations in $R$
(relative to~$\Gamma'$).

\begin{example}
Consider $\AA \in \nData{2}{\Sigma}$
from Figure~\ref{fig:well-typed}(a) where $\Sigma = \emptyset$.
Then, the set $\myEnv{\AA}{\Sigma}{\Gamma}{a}{\emptyset}{\{\relsaord{1}{1},\relsaord{1}{2}\}} =
\myEnv{\AA}{\Sigma}{\Binary_{df}}{a}{\emptyset}{\{\relsaord{1}{1}\}}$ contains
exactly the yellow elements (with data-value pairs $(1,1)$), and
$\myEnv{\AA}{\Sigma}{\Gamma}{a}{\emptyset}{\{\relsaord{1}{2}\}}$ contains the two blue
elements (with data-value pairs $(2,1)$ and $(3,1)$).
\end{example}

\newcommand{\cconetwo}{
\parbox{0.5cm}{
\begin{tikzpicture}[scale=0.25]
    \draw [thick] (0,1) -- (1,0);
\end{tikzpicture}}}

\newcommand{\cconeone}{
\parbox{0.5cm}{
\begin{tikzpicture}[scale=0.25]
    \draw [thick] (0,1) -- (1,1);
    \draw [draw=none] (0,0) -- (1,0);        
\end{tikzpicture}}}

\newcommand{\cctwotwo}{
\parbox{0.5cm}{
\begin{tikzpicture}[scale=0.25]
    \draw [thick] (0,0) -- (1,0);
    \draw [draw=none] (0,1) -- (1,1);    
\end{tikzpicture}}}

\newcommand{\ccoott}{
\parbox{0.5cm}{
\begin{tikzpicture}[scale=0.25]
    \draw [thick] (0,1) -- (1,1);
    \draw [thick] (0,0) -- (1,0);    
\end{tikzpicture}}}

\newcommand{\ccottt}{
\parbox{0.5cm}{
\begin{tikzpicture}[scale=0.25]
    \draw [thick] (0,1) -- (1,0);
    \draw [thick] (0,0) -- (1,0);    
\end{tikzpicture}}}

\newcommand{\ccooot}{
\parbox{0.5cm}{
\begin{tikzpicture}[scale=0.25]
    \draw [thick] (0,1) -- (1,1);
    \draw [thick] (0,1) -- (1,0);    
\end{tikzpicture}}}

\newcommand{\ccall}{
\parbox{0.5cm}{
\begin{tikzpicture}[scale=0.25]
    \draw [thick] (0,1) -- (1,1);
    \draw [thick] (0,1) -- (1,0);
    \draw [thick] (0,0) -- (1,0);    
\end{tikzpicture}}}

\begin{figure*}[t]
\scalebox{0.85}{
\hspace{-4pt}
\begin{tabular}{l}
\begin{tikzpicture}[node distance=2.5cm]
  \pgfsetxvec{\pgfpoint{1.1cm}{0cm}}
  \pgfsetyvec{\pgfpoint{0cm}{1.1cm}}
  \pgfsetzvec{\pgfpoint{0cm}{1.1cm}}

\put(-5,0){
	\node []	(a) at (-2.7,1)	{(a)};
	\node []	(a) at (0,-1.8)	{};
	\node [dataredred, label=above:$a$]	(A) at (0.1,0.1)	{1 \nodepart{second} 2};
	\node [datayellowyellow]    		(B) at (-1.7,-0.6)	{1 \nodepart{second} 1};
	\node [datayellowyellow]			(C) at (-0.7,-0.5)	{1 \nodepart{second} 1};
	\node [datayellowyellow]			(D) at (1,-0.6)	{1 \nodepart{second} 1};
	\node [datagreengreen]				(E) at (-1,0.8)	{2 \nodepart{second} 2};
	\node [datagraygray]				(F) at (1,0.8)	{3 \nodepart{second} 4};
	\node [datablueblue]				(G) at (-2.0,0.7)	{2 \nodepart{second} 1};
	\node [datablueblue]				(H) at (1.8,0)	{3 \nodepart{second} 1};
	}
\end{tikzpicture}
~~\vline~~
\begin{tikzpicture}[node distance=2.5cm]
  \pgfsetxvec{\pgfpoint{1.1cm}{0cm}}
  \pgfsetyvec{\pgfpoint{0cm}{1.1cm}}
  \pgfsetzvec{\pgfpoint{0cm}{1.1cm}}

	\node []	(b) at (-2.7,1)	{(b)};
	\node []	(a) at (0,-1.8)	{};	
	\node [dataredred, label=above:$a$]	(A) at (0.1,0.1)	{1 \nodepart{second} 2};
	\node [dataredred, label=below:$\substack{\relsaord{1}{1}\\\relsaord{2}{2}}$]	(A) at (0.1,0.1)	{1 \nodepart{second} 2};
	\node [datayellowyellow, label=below:$\substack{\relsaord{1}{1}\\\relsaord{1}{2}}$]    		(B) at (-1.7,-0.8)	{1 \nodepart{second} 1};
	\node [datayellowyellow, label=below:$\substack{\relsaord{1}{1}\\\relsaord{1}{2}}$]			(C) at (-0.7,-0.7)	{1 \nodepart{second} 1};
	\node [datayellowyellow, label=below:$\substack{\relsaord{1}{1}\\\relsaord{1}{2}}$]			(D) at (1,-0.8)	{1 \nodepart{second} 1};
	\node [datagreengreen, label=below:$\substack{\relsaord{2}{2}}$]				(E) at (-1,0.8)	{2 \nodepart{second} 2};
	\node [datagraygray]				(F) at (1,0.8)	{3 \nodepart{second} 4};
	\node [datablueblue, label=below:$\substack{\relsaord{1}{2}}$]				(G) at (-2.0,0.7)	{2 \nodepart{second} 1};
	\node [datablueblue, label=below:$\substack{\relsaord{1}{2}}$]				(H) at (1.8,0)	{3 \nodepart{second} 1};
\end{tikzpicture}
~~\vline~~
\begin{tikzpicture}[node distance=2.5cm]
  \pgfsetxvec{\pgfpoint{1.1cm}{0cm}}
  \pgfsetyvec{\pgfpoint{0cm}{1.1cm}}
  \pgfsetzvec{\pgfpoint{0cm}{1.1cm}}

	\node []	(c) at (-2.7,1)	{(c)};
	\node []	(a) at (0,-1.8)	{};	
	\node [dataredred]	(A) at (0.1,1)	{1 \nodepart{second} 2};
	\node [draw, rectangle, minimum width=1.4cm, minimum height=3.5cm, rounded corners=0.1cm, label=above:$a$] (cc) at (0.1,-0.05) {};
	
	\node [fill=darkseagreen!70, draw, rectangle, minimum width=1.1cm] (const) at (0.1,0.2) {\scalebox{0.8}{$\hspace{1em} \ge 1$}};
	\node [] (const) at (-0.1,0.2) {\cctwotwo};
	
	\node [fill=red!30,draw, rectangle, minimum width=1.1cm] (const) at (0.1,-0.3) {\scalebox{0.8}{$\hspace{1em} \ge 1$}};	
	\node [] (const) at (-0.1,-0.3) {\ccoott};	
	
	\node [fill=blue!20, draw, rectangle, minimum width=1.1cm] (const) at (0.1,-0.8) {\scalebox{0.8}{$\hspace{1em} \ge 2$}};
	\node [] (const) at (-0.1,-0.8) {\cconetwo};
		
	\node [fill=bananamania!90,draw, rectangle, minimum width=1.1cm] (const) at (0.1,-1.3) {\scalebox{0.8}{$\hspace{1em} \ge 3$}};
	\node [] (const) at (-0.1,-1.3) {\ccooot};
		
	\node [datayellowyellow]    		(B) at (-2,-0.8)	{1 \nodepart{second} 1};
	\node [datayellowyellow]			(C) at (-1.2,-0.7)	{1 \nodepart{second} 1};
	\node [datayellowyellow]			(D) at (1.5,-0.8)	{1 \nodepart{second} 1};
	\node [datagreengreen]				(E) at (-1,0.8)	{2 \nodepart{second} 2};
	\node [datagraygray]				(F) at (1.3,0.8)	{3 \nodepart{second} 4};
	\node [datablueblue]				(G) at (-2.0,0.7)	{2 \nodepart{second} 1};
	\node [datablueblue]				(H) at (2,0.4)	{3 \nodepart{second} 1};
\end{tikzpicture}
\end{tabular}
}\\
\hrule
{
~\\[1ex]
{
\begin{minipage}{0.96\textwidth}
\hspace{-2em}
\scalebox{0.85}{
\begin{tabular}{c}
{
\begin{tikzpicture}[node distance=2.5cm, scale=1.1]
  \pgfsetxvec{\pgfpoint{1.1cm}{0cm}}
  \pgfsetyvec{\pgfpoint{0cm}{1.1cm}}
  \pgfsetzvec{\pgfpoint{0cm}{1.1cm}}

	\node []	(c) at (-1,1.3)	{(d)};
\put(0,0){
	\node [draw, fill=bananamania!90, rectangle, minimum width=1.4cm, minimum height=3.9cm, rounded corners=0.1cm] (cc) at (0.1,0) {};
	\node [datayellowyellow]	(A) at (0.1,1)	{1 \nodepart{second} 1};	
	
	\node [draw, rectangle, minimum width=1.1cm] (const) at (0.1,0.2) {\scalebox{0.8}{$\hspace{1em} \ge 1$}};
	\node [] (const) at (-0.1,0.2) {\cconeone};
	
	\node [draw, rectangle, minimum width=1.1cm] (const) at (0.1,-0.3) {\scalebox{0.8}{$\hspace{1em} \ge 2$}};	
	\node [] (const) at (-0.1,-0.3) {\ccottt};	
	
	\node [draw, rectangle, minimum width=1.1cm] (const) at (0.1,-0.8) {\scalebox{0.8}{$\hspace{1em} \ge 3$}};
	\node [] (const) at (-0.1,-0.8) {\ccall};
		
	\node [rectangle, minimum width=1.1cm] (const) at (0.1,-1.3) {\scalebox{0.8}{$~\eqx$}};
}
\put(50,0){
	\node [draw, fill=bananamania!90, rectangle, minimum width=1.4cm, minimum height=3.9cm, rounded corners=0.1cm] (cc) at (0.1,0) {};
	\node [datayellowyellow]	(A) at (0.1,1)	{1 \nodepart{second} 1};	
	
	\node [draw, rectangle, minimum width=1.1cm] (const) at (0.1,0.2) {\scalebox{0.8}{$\hspace{1em} \ge 1$}};
	\node [] (const) at (-0.1,0.2) {\cconeone};
	
	\node [draw, rectangle, minimum width=1.1cm] (const) at (0.1,-0.3) {\scalebox{0.8}{$\hspace{1em} \ge 2$}};	
	\node [] (const) at (-0.1,-0.3) {\ccottt};	
	
	\node [draw, rectangle, minimum width=1.1cm] (const) at (0.1,-0.8) {\scalebox{0.8}{$\hspace{1em} \ge 3$}};
	\node [] (const) at (-0.1,-0.8) {\ccall};
		
	\node [rectangle, minimum width=1.1cm] (const) at (0.1,-1.3) {\scalebox{0.8}{$~\eqx$}};
}
\put(100,0){
	\node [draw, fill=bananamania!90, rectangle, minimum width=1.4cm, minimum height=3.9cm, rounded corners=0.1cm] (cc) at (0.1,0) {};
	\node [datayellowyellow]	(A) at (0.1,1)	{1 \nodepart{second} 1};	
	
	\node [draw, rectangle, minimum width=1.1cm] (const) at (0.1,0.2) {\scalebox{0.8}{$\hspace{1em} \ge 1$}};
	\node [] (const) at (-0.1,0.2) {\cconeone};
	
	\node [draw, rectangle, minimum width=1.1cm] (const) at (0.1,-0.3) {\scalebox{0.8}{$\hspace{1em} \ge 2$}};	
	\node [] (const) at (-0.1,-0.3) {\ccottt};	
	
	\node [draw, rectangle, minimum width=1.1cm] (const) at (0.1,-0.8) {\scalebox{0.8}{$\hspace{1em} \ge 3$}};
	\node [] (const) at (-0.1,-0.8) {\ccall};
		
	\node [rectangle, minimum width=1.1cm] (const) at (0.1,-1.3) {\scalebox{0.8}{$~\eqx$}};
}
\put(150,0){
	\node [draw, fill=darkseagreen!70, rectangle, minimum width=1.4cm, minimum height=3.9cm, rounded corners=0.1cm] (cc) at (0.1,0) {};
	\node [datagreengreen]	(A) at (0.1,1)	{2 \nodepart{second} 2};	
	
	\node [draw, rectangle, minimum width=1.1cm] (const) at (0.1,0.2) {\scalebox{0.8}{$\hspace{1em} \ge 1$}};
	\node [] (const) at (-0.1,0.2) {\cconeone};
	
	\node [draw, rectangle, minimum width=1.1cm] (const) at (0.1,-0.3) {\scalebox{0.8}{$\hspace{1em} \ge 1$}};	
	\node [] (const) at (-0.1,-0.3) {\ccottt};	
	
	\node [draw, rectangle, minimum width=1.1cm] (const) at (0.1,-0.8) {\scalebox{0.8}{$\hspace{1em} \ge 1$}};
	\node [] (const) at (-0.1,-0.8) {\ccall};
		
	\node [rectangle, minimum width=1.1cm] (const) at (0.1,-1.3) {\scalebox{0.8}{$~\eqx$}};
}
\put(200,0){
	\node [draw, fill=red!30, rectangle, minimum width=1.4cm, minimum height=3.9cm, rounded corners=0.1cm] (cc) at (0.1,0) {};
	\node [dataredred]	(A) at (0.1,1)	{1 \nodepart{second} 2};	
	
	\node [draw, rectangle, minimum width=1.1cm] (const) at (0.1,0.2) {\scalebox{0.8}{$\hspace{1em} \ge 1$}};
	\node [] (const) at (-0.1,0.2) {\cctwotwo};
	
	\node [draw, rectangle, minimum width=1.1cm] (const) at (0.1,-0.3) {\scalebox{0.8}{$\hspace{1em} \ge 1$}};	
	\node [] (const) at (-0.1,-0.3) {\ccoott};	
	
	\node [draw, rectangle, minimum width=1.1cm] (const) at (0.1,-0.8) {\scalebox{0.8}{$\hspace{1em} \ge 2$}};
	\node [] (const) at (-0.1,-0.8) {\cconetwo};
		
	\node [draw, rectangle, minimum width=1.1cm] (const) at (0.1,-1.3) {\scalebox{0.8}{$\hspace{1em} \ge 3$}};
	\node [] (const) at (-0.1,-1.3) {\ccooot};
}
\put(250,0){
	\node [draw, fill=blue!20, rectangle, minimum width=1.4cm, minimum height=3.9cm, rounded corners=0.1cm] (cc) at (0.1,0.0) {};
	\node [datablueblue]	(A) at (0.1,1)	{2 \nodepart{second} 1};	
	
	\node [draw, rectangle, minimum width=1.1cm] (const) at (0.1,0.2) {\scalebox{0.8}{$\hspace{1em} \ge 3$}};
	\node [] (const) at (-0.1,0.2) {\cctwotwo};
	
	\node [draw, rectangle, minimum width=1.1cm] (const) at (0.1,-0.3) {\scalebox{0.8}{$\hspace{1em} \ge 1$}};	
	\node [] (const) at (-0.1,-0.3) {\ccoott};	
	
	\node [draw, rectangle, minimum width=1.1cm] (const) at (0.1,-0.8) {\scalebox{0.8}{$\hspace{1em} \ge 1$}};
	\node [] (const) at (-0.1,-0.8) {\cconetwo};
		
	\node [draw, rectangle, minimum width=1.1cm] (const) at (0.1,-1.3) {\scalebox{0.8}{$\hspace{1em} \ge 1$}};
	\node [] (const) at (-0.1,-1.3) {\ccooot};
}
\put(300,0){
	\node [draw, fill=blue!20, rectangle, minimum width=1.4cm, minimum height=3.28cm, rounded corners=0.1cm] (cc) at (0.1,0.26) {};
	\node [datablueblue]	(A) at (0.1,1)	{3 \nodepart{second} 1};	
	
	\node [draw, rectangle, minimum width=1.1cm] (const) at (0.1,0.2) {\scalebox{0.8}{$\hspace{1em} \ge 1$}};
	\node [] (const) at (-0.1,0.2) {\cconeone};
	
	\node [draw, rectangle, minimum width=1.1cm] (const) at (0.1,-0.3) {\scalebox{0.8}{$\hspace{1em} \ge 3$}};	
	\node [] (const) at (-0.1,-0.3) {\cctwotwo};	
	
	\node [draw, rectangle, minimum width=1.1cm] (const) at (0.1,-0.8) {\scalebox{0.8}{$\hspace{1em} \ge 1$}};
	\node [] (const) at (-0.1,-0.8) {\ccoott};
}
\put(350,0){
	\node [draw, fill=gray!20, rectangle, minimum width=1.4cm, minimum height=2.7cm, rounded corners=0.1cm] (cc) at (0.1,0.5) {};
	\node [datagraygray]	(A) at (0.1,1)	{3 \nodepart{second} 4};	
	
	\node [draw, rectangle, minimum width=1.1cm] (const) at (0.1,0.2) {\scalebox{0.8}{$\hspace{1em} \ge 1$}};
	\node [] (const) at (-0.1,0.2) {\cconeone};
	
	\node [draw, rectangle, minimum width=1.1cm] (const) at (0.1,-0.3) {\scalebox{0.8}{$\hspace{1em} \ge 1$}};	
	\node [] (const) at (-0.1,-0.3) {\ccoott};		
}
\node []	(c) at (12.1,-1)	{};	
\end{tikzpicture}
}
\scalebox{0.8}{
\begin{tikzpicture}[node distance=2cm, scale=1.1]
\put(-60,0){
\node [] (cc) at (0,0) {\cconeone: $\{\relsaord{1}{1}\}$};
\node [] (cc) at (0,-0.6) {\cctwotwo: $\{\relsaord{2}{2}\}$};	
\node [] (cc) at (0,-1.2) {\cconetwo: $\{\relsaord{1}{2}\}$};

\node [] (cc) at (0,-1.8) {\ccoott: $\{\relsaord{1}{1},\relsaord{2}{2}\}$};	
\node [] (cc) at (0,-2.4) {\ccooot: $\{\relsaord{1}{1},\relsaord{1}{2}\}$};
\node [] (cc) at (0,-3.0) {\ccottt: $\{\relsaord{2}{2},\relsaord{1}{2}\}$};	

\node [] (cc) at (0,-3.6) {\ccall: $\{\relsaord{1}{1},\relsaord{2}{2},\relsaord{1}{2}\}$};
}
\end{tikzpicture}
}
\end{tabular}
}
\end{minipage}
}
}
	\caption{(a)~A data structure over $\Unary = \emptyset$. (b)~Adding unary predicates for a given element $a$. (c)~Adding counting constraints to $a$. (d)~A well-typed data structure from $\Data{\{\eqx\} \cup \myCountConst{}{}{3}}$.
\label{fig:well-typed}}
\end{figure*}

\newcommand{\AddPred}{\Theta}

Let us now go through the reduction step by step.\\

\noindent{\bf Step 1: Transform Binary into Unary Relations}\\

In the first step, we get rid of the binary relations by representing them as unary ones.
In fact, in a formula $\locformr{\psi}{x}{1}$ from $\rndFO{2}{\Unary,\Binary}{1}$,
$\psi$ only talks about elements that are directly related to
$a = I(x)$ in terms of pairs from $\Binary$.
We can hence rewrite $\psi$ into $\psi'$ so that all comparisons are wrt.\ $x$,
i.e., they are of the form $\rels{i}{j}{x}{y}$. 
Then, a pair $\rels{i}{j}{}{} \in \Binary$
can be seen as a unary predicate that holds at $b$ iff $\rels{i}{j}{a}{b}$.
In this way, we eliminate the binary relations and replace $\psi'$ with a
first-order formula $\psi''$ over unary predicates.

\begin{example}
Adding unary relations to a data structure for a given element $a$ is illustrated in
Figure~\ref{fig:well-typed}(b) (recall that $\Sigma = \emptyset$). We
can see for instance that we add to the yellow nodes the labels
$\rels{1}{1}{}{}$  and  $\rels{1}{2}{}{}$ as the first value of $a$ is
equal to the first and second values of these nodes.
\end{example}

Thanks to the unary predicates, we can now apply Lemma~\ref{lem:threshold}
(which was a consequence of locality of first-order logic over unary symbols only).
That is, to know whether
$\psi''$ holds when $x$ is interpreted as $a$, it is enough to know how often every unary predicate is present
in the environment of $a$, counted only up to some $M \ge 1$.
However, we will then give up the information of whether the two data values at
$a$ coincide or not. Therefore, we introduce a unary predicate $\eqx$, which
shall label those events whose two data values coincide.
Accordingly, we say that $\AA=(A,(P_{\unary}),\ifunct,\ofunct) \in \nData{2}{\Theta \cup \{\eqx\}}$
is \emph{eq-respecting} if, for all $a \in A$, we have
$a \in P_\eqx$ iff $\ifunct(a) = \ofunct(a)$.

Once we add this information to $a$, it is enough to know the size of
$\myEnv{\AA}{\Sigma}{\Binary}{a}{\USet}{\BSet}$ for every $\USet \subseteq \Unary$ and nonempty $\BSet \subseteq \Binary$, measured up to $M$.
To reason about these sizes, we introduce a unary predicate
$\newcc{U}{R}{\ge}{m}$ 
for all $U \subseteq \Sigma$, nonempty sets $R \subseteq \Gamma$,
and $m \in \{1,\ldots,M\}$ (which is interpreted as ``${\ge}\, m$'').
We also call such a predicate a \emph{counting constraint} and
denote the set of all counting constraints by $\myCountConst{\Unary}{\Binary}{M}$
(recall that we fixed $\Unary$ and $\Binary$).
For a finite set $\Theta$ with $\Sigma \subseteq \Theta$, we call
$\AA=(A,(P_{\unary}),\ifunct,\ofunct) \in \Data{\Theta \cup \myCountConst{}{}{M}}$
\emph{cc-respecting} if, for all $a \in A$,
we have $a \in P_{\newcc{U}{R}{\ge}{m}}$ iff $|\myEnv{\AA}{\Sigma}{\Binary}{a}{\USet}{\BSet}| \ge m$.

Finally, we call $\AA \in \nData{2}{\Theta \cup \{\eqx\} \cup \myCountConst{}{}{M}}$
\emph{well-typed} if it is eq-respecting and cc-respecting.

\begin{example}
In Figure~\ref{fig:well-typed}(c), where we suppose $M = 3$ and $\Unary = \emptyset$, the element $a$ satisfies
the counting constraints $\newcc{\emptyset}{\{\relsaord{2}{2}\}}{\ge}{1}$, $\newcc{\emptyset}{\{\relsaord{1}{1},\relsaord{2}{2}\}}{\ge}{1}$,
$\newcc{\emptyset}{\{\relsaord{1}{2}\}}{\ge}{2}$, and $\newcc{\emptyset}{\{\relsaord{1}{1},\relsaord{1}{2}\}}{\ge}{3}$,
as well as all inherited constraints for smaller constants (which we omitted).
We write $\newcc{\emptyset}{\BSet}{\ge}{m}$ as $\BSet \ge m$.
In fact, pairs from $\BSet$ are represented as black bars in the obvious way
(cf.\ Figure~\ref{fig:well-typed}(d));
moreover, for each constraint, the corresponding elements have the
same color. For instance, $a$ satisfies the counting constraint
$\newcc{\emptyset}{\{\relsaord{1}{1},\relsaord{1}{2}\}}{\ge}{3}$ as
they are at least three nodes (the yellow ones) having as first
and second values the same value as the first one of $a$.
Finally, the data structure from Figure~\ref{fig:well-typed}(d) is well-typed, i.e.,
eq- and cc-respecting. Again, we omit inherited constraints.
\end{example}

To summarize, we have the following reduction:

\begin{lemma}
\label{lem:new-fo-count}
  For each formula $\phi \in \rndFO{2}{\Unary,\Binary}{1}$, we can effectively compute $M \in \N$ and $\chi \in \ndFO{2}{\Unary\cup \{\eqx\} \cup \myCountConst{\Unary}{\Binary}{M},\emptyset}$
  such that $\phi$ is satisfiable iff $\chi$ has a well-typed model.
\end{lemma}

\begin{proof}
Consider $\locformr{\psi}{x}{1}$
where $\psi$ is a formula from $\ndFO{2}{\Unary,\Binary}$
with one free variable $x$.
Wlog., we assume that $x$ is not quantified in $\psi$.
We replace, in $\psi$, every occurrence of a formula $\rels{i}{j}{y}{z}$ with $y \neq x$ by
$$\bigvee_{\substack{k\in\{1,2\}\,|\\ \rels{k}{i}{}{}, \rels{k}{j}{}{}\in\Binary}} \rels{k}{i}{x}{y} ~\wedge~ \rels{k}{j}{x}{z}\,.$$
Call the resulting formula $\psi'$.
Replace, in $\psi'$,
every formula $\rels{i}{j}{x}{y}$ by $\relsaord{i}{j}(y)$ to obtain
an $\ndFO{2}{\Unary \cup \Binary,\emptyset}$ formula $\psi''$.
Consider $\AA=(A,(P_{\unary})_{\unary \in \Unary \cup \Binary},\ifunct,\ofunct,) \in \nData{2}{\Unary \cup \Binary}$ and an interpretation function $I$ such that,
for all $b \in A$ and $\relsaord{i}{j} \in \Binary$,
we have $b \in P_{\relsaord{i}{j}}$ iff $\rels{i}{j}{I(x)}{b}$.
Then, 
$$
\vprojr{\AA}{I(x)}{1} \models_I \psi(x)
~\Longleftrightarrow~
\vprojr{\AA}{I(x)}{1} \models_I \psi'(x)
~\Longleftrightarrow~
\vprojr{\AA}{I(x)}{1} \models_I \psi''(x)\,.
$$

According to Lemma~\ref{lem:threshold},
we can effectively transform $\psi''$ into an equivalent\linebreak$\ndFO{2}{\Unary \cup \Binary, \emptyset}$ formula $\hat\psi''$
that is a Boolean combination of
formulas of the form $\unary(x)$ with $\unary \in \Unary \cup \Gamma$ and
threshold formulas of the form
$\exists^{\ge k} y. \phi_U(y)$ where $U \subseteq \Unary \cup \Binary$ and
$\phi_U(y) = \bigwedge_{\unary \in U} \unary(y) \wedge \bigwedge_{\unary \in (\Unary \cup \Binary) \setminus U} \neg \unary(y)$.
Let $M$ be the maximal such $k$ (or $M = 0$ if there is no threshold formula).
Again, we assume that $x$ is not quantified in $\hat\psi''$.

We obtain the $\ndFO{2}{\Unary \cup \{\eqx\} \cup \myCountConst{\Unary}{\Binary} {M},\emptyset}$
formula $\chi'$ from $\hat\psi''$ by replacing
\begin{itemize}
\item $\relsaord{1}{2}(x)$ by $\eqx(x)$, and $\relsaord{1}{1}(x)$ and $\relsaord{2}{2}(x)$ by $\mathit{true}$,
\item $\exists^{\ge k} y.\phi_U(y)$ by
$\begin{cases}
\mathit{false} & \text{if } U \cap \Binary = \emptyset\\
\newcc{\USet \cap \Unary}{\USet \cap \Binary}{\ge}{k}(x)
& \text{if } U \cap \Binary \neq \emptyset
\end{cases}$
\end{itemize}
We can then eliminate redundant $\mathit{true}$ and $\mathit{false}$.
Suppose a well-typed data structure $\AA=(A,(P_{\unary}),\ifunct,\ofunct) \in \nData{2}{\Unary \cup \Binary \cup \{\eqx\} \cup \myCountConst{\Unary}{\Binary} {M}}$ and an interpretation function $I$ such that,
for all $b \in A$ and $\relsaord{i}{j} \in \Binary$, we have $b \in P_{\relsaord{i}{j}}$ iff $\rels{i}{j}{I(x)}{b}$.
Then, 
$$
\vprojr{\AA}{I(x)}{1} \models_I \hat\psi''(x)
~\Longleftrightarrow~
\vprojr{\AA}{I(x)}{1} \models_I \chi'(x)
\,.$$
Moreover, for $\USet \subseteq \Unary$, a nonempty set $\BSet \subseteq \Binary$, and $k \in \N$, we have
$$
\vprojr{\AA}{I(x)}{1} \models_I \newcc{\USet}{\BSet}{\ge}{k}(x)
~\Longleftrightarrow~
\AA \models_I \newcc{\USet}{\BSet}{\ge}{k}(x)
\,.$$
We deduce that,
for all well-typed $\AA \in \nData{2}{\Unary \cup \Binary \cup \{\eqx\} \cup \myCountConst{\Unary}{\Binary} {M}}$ and
interpretation functions $I$,
$$\AA \models_I \locformr{\psi}{x}{1}
~\Longleftrightarrow~
\AA \models_I \chi'(x)\,.$$
We can apply this translation to all the terms $\locformr{\psi}{x}{1}$ present in the formula $\phi \in \rndFO{2}{\Unary,\linebreak[0]\Binary}{1}$ to obtain a formula $\chi \in \ndFO{2}{\Unary\cup \{\eqx\} \cup \myCountConst{\Unary}{\Binary}{M},\emptyset}$. From what we have seen if $\chi$ has a well-typed model then $\phi$ is satisfiable (note to obtain a data structure in $\nData{2}{\Unary}$ from a well-typed model in  $\nData{2}{\Unary \cup \Binary \cup \{\eqx\} \cup \myCountConst{\Unary}{\Binary} {M}}$, we simply forget the unary predicates not in $\Unary$). Furthermore from a data structure in $\nData{2}{\Unary}$, we can easily build a well-typed model in $\nData{2}{\Unary \cup \Binary \cup \{\eqx\} \cup \myCountConst{\Unary}{\Binary} {M}}$ (by respecting the typing rules) and this is allows to conclude that if $\phi$ is satisfiable then $\chi$ has a well-typed model.
\end{proof}

\noindent \textbf{Step 2: Well-Diagonalized Structures}\\

In $\myCountConst{}{}{M}$, we still have the diagonal
relation $(1,2) \in \Gamma$.
Our goal is to get rid of it
so that we only deal with the diagonal-free set
$\Binary_{df}=\{\relsaord{1}{1},\relsaord{2}{2}\}$. The idea is again to
extend a given structure $\AA$, but now we add new elements,
one for each value $n \in \Values{\AA}$, which we tag with a unary symbol
$\ediag$ and whose two data values are $n$. Diagonal equality
will be ensured through making a detour via these `diagonal' elements
(hence the name $\ediag$).

Formally, when we start from some $\AA=(A,(P_{\unary}),\ifunct,\ofunct) \in \nData{2}{\extSigma \cup \{\eqx\}}$,
the data structure $\addediag{\AA} \in \nData{2}{\extSigma \cup \{\eqx,\ediag\}}$
is defined as $(A',(P'_{\unary}),\ifunct',\ofunct')$ where 
$A' = A \uplus \Values{\AA}$, $f_i'(a) = f_i(a)$ for all $a \in A$ and $i \in \{1,2\}$,
$f_1'(a) = f_2'(a) = a$ for all $a \in \Values{\AA}$,
$P'_{\unary}=P_{\unary}$ for all $\unary \in \Ueq\setminus\{\eqx\} $,
$P'_{\ediag} = \Values{\AA}$, and
$P'_{\eqx}=P_{\eqx}\cup \Values{\AA}$.

\begin{example}
The structure $\addediag{\AA}$ is illustrated in Figure~\ref{fig:diagonal}(a), with $\extSigma = \emptyset$.
\end{example}

\begin{figure*}[t]
\centering
\scalebox{0.81}{
\begin{tabular}{c}
\begin{tikzpicture}[node distance=2cm]
  \pgfsetxvec{\pgfpoint{1.1cm}{0cm}}
  \pgfsetyvec{\pgfpoint{0cm}{1.1cm}}
  \pgfsetzvec{\pgfpoint{0cm}{1.1cm}}

	\node []	(b) at (-5.25,1.8)	{(a)};
	\node []	(b) at (-4,1.6)	{$\addediag{\AA}$};	
	\node []	(b) at (0.2,1.6)	{$\AA$};		

	\node [datayellowyellow, label=above:$b_1$]	(A) at (-4.5,0.3)	{1 \nodepart{second} 1};
	\node [datanone, label=below:$\ediag$]	(A) at (-4.5,0.3)	{1 \nodepart{second} 1};
	\node [datanone, label=below:$\eqx$]	(A) at (-4.5,-0.1)	{};	

	\node [datagreengreen, label=above:$b_2$]	(A) at (-3.8,0.3)	{2 \nodepart{second} 2};
	\node [datanone, label=below:$\ediag$]	(A) at (-3.8,0.3)	{2 \nodepart{second} 2};
	\node [datanone, label=below:$\eqx$]	(A) at (-3.8,-0.1)	{};	

	\node [datablueblue, label=above:$b_3$]	(A) at (-3.1,0.3)	{3 \nodepart{second} 3};
	\node [datanone, label=below:$\ediag$]	(A) at (-3.1,0.3)	{3 \nodepart{second} 3};
	\node [datanone, label=below:$\eqx$]	(A) at (-3.1,-0.1)	{};	

	\node [dataredred, label=above:$b_4$]	(A) at (-2.4,0.3)	{4 \nodepart{second} 4};
	\node [datanone, label=below:$\ediag$]	(A) at (-2.4,0.3)	{4 \nodepart{second} 4};
	\node [datanone, label=below:$\eqx$]	(A) at (-2.4,-0.1)	{};	


	\node [draw, rectangle, minimum width=4.5cm, minimum height=4cm, rounded corners=0.2cm] (cc) at (0.1,0.2) {};
	\node [draw, rectangle, minimum width=7.8cm, minimum height=4.3cm, rounded corners=0.2cm] (cc) at (-1.3,0.2) {};	


	\node [datayellowgreen, label=above:$a_7$]	(A) at (0.2,0.2)	{1 \nodepart{second} 2};

	\node [datayellowyellow, label=above:$a_1$]    		(B) at (-1.5,-0.4)	{1 \nodepart{second} 1};
	\node [datanone, label=below:$\eqx$]    		(B) at (-1.5,-0.4)	{1 \nodepart{second} 1};	

	\node [datayellowyellow, label=above:$a_2$]    		(B) at (-0.6,-0.7)	{1 \nodepart{second} 1};
	\node [datanone, label=below:$\eqx$]    		(B) at (-0.6,-0.7)	{1 \nodepart{second} 1};	

	\node [datayellowyellow, label=above:$a_3$]    		(B) at (1,-0.7)	{1 \nodepart{second} 1};
	\node [datanone, label=below:$\eqx$]    		(B) at (1,-0.7)	{1 \nodepart{second} 1};	

	\node [datagreengreen, label=above:$a_6$]				(E) at (-0.6,1.1)	{2 \nodepart{second} 2};
	\node [datanone, label=below:$\eqx$]    		(B) at (-0.6,1.1)	{2 \nodepart{second} 2};	

	\node [databluered, label=above:$a_8$]				(F) at (1,0.8)	{3 \nodepart{second} 4};
	\node [datagreenyellow, label=above:$a_5$]				(G) at (-1.5,1.1)	{2 \nodepart{second} 1};

	\node [datablueyellow, label=above:$a_4$]				(H) at (1.8,0)	{3 \nodepart{second} 1};
\end{tikzpicture}
\end{tabular}
}
	\vrule
\scalebox{0.81}{
\begin{tabular}{c}
\begin{tikzpicture}[node distance=2cm]
  \pgfsetxvec{\pgfpoint{1.1cm}{0cm}}
  \pgfsetyvec{\pgfpoint{0cm}{1.1cm}}
  \pgfsetzvec{\pgfpoint{0cm}{1.1cm}}

	\node []	(b) at (-5.25,1.8)	{(b)};

	\node [datablueyellow, label=above:$b_1$]	(A) at (-4.5,0.3)	{3 \nodepart{second} 1};
	\node [datanone, label=below:$\ediag$]	(A) at (-4.5,0.3)	{3 \nodepart{second} 1};
	\node [datanone, label=below:$\eqx$]	(A) at (-4.5,-0.1)	{};	

	\node [datayellowgreen, label=above:$b_2$]	(A) at (-3.8,0.3)	{1 \nodepart{second} 2};
	\node [datanone, label=below:$\ediag$]	(A) at (-3.8,0.3)	{1 \nodepart{second} 2};
	\node [datanone, label=below:$\eqx$]	(A) at (-3.8,-0.1)	{};	

	\node [dataredblue, label=above:$b_3$]	(A) at (-3.1,0.3)	{4 \nodepart{second} 3};
	\node [datanone, label=below:$\ediag$]	(A) at (-3.1,0.3)	{4 \nodepart{second} 3};
	\node [datanone, label=below:$\eqx$]	(A) at (-3.1,-0.1)	{};	

	\node [datagreenred, label=above:$b_4$]	(A) at (-2.4,0.3)	{2 \nodepart{second} 4};
	\node [datanone, label=below:$\ediag$]	(A) at (-2.4,0.3)	{2 \nodepart{second} 4};
	\node [datanone, label=below:$\eqx$]	(A) at (-2.4,-0.1)	{};	


	\node [draw, rectangle, minimum width=4.5cm, minimum height=4cm, rounded corners=0.2cm] (cc) at (0.1,0.2) {};
	\node [draw, rectangle, minimum width=7.8cm, minimum height=4.3cm, rounded corners=0.2cm] (cc) at (-1.3,0.2) {};


	\node [databluegreen, label=above:$a_7$]	(A) at (0.2,0.2)	{3 \nodepart{second} 2};

	\node [datablueyellow, label=above:$a_1$]    		(B) at (-1.5,-0.4)	{3 \nodepart{second} 1};
	\node [datanone, label=below:$\eqx$]    		(B) at (-1.5,-0.4)	{3 \nodepart{second} 1};	

	\node [datablueyellow, label=above:$a_2$]    		(B) at (-0.6,-0.7)	{3 \nodepart{second} 1};
	\node [datanone, label=below:$\eqx$]    		(B) at (-0.6,-0.7)	{3 \nodepart{second} 1};	

	\node [datablueyellow, label=above:$a_3$]    		(B) at (1,-0.7)	{3 \nodepart{second} 1};
	\node [datanone, label=below:$\eqx$]    		(B) at (1,-0.7)	{3 \nodepart{second} 1};	

	\node [datayellowgreen, label=above:$a_6$]				(E) at (-0.6,1.1)	{1 \nodepart{second} 2};
	\node [datanone, label=below:$\eqx$]    		(B) at (-0.6,1.1)	{1 \nodepart{second} 2};	

	\node [dataredred, label=above:$a_8$]				(F) at (1,0.8)	{4 \nodepart{second} 4};
	\node [datayellowyellow, label=above:$a_5$]				(G) at (-1.5,1.1)	{1 \nodepart{second} 1};
	\node [dataredyellow, label=above:$a_4$]				(H) at (1.8,0)	{4 \nodepart{second} 1};
\end{tikzpicture}
\end{tabular}
}
\caption{(a) Adding diagonal elements. (a)$\leftarrow$(b) Making a
  data structure eq-respecting with the
following permutation on the first element $1\mapsto 2, 2
 \mapsto 4, 3\mapsto 1, 4 \mapsto 3$.\label{fig:diagonal}} 
\end{figure*}

With this, we say that $\BB \in \nData{2}{\extSigma \cup \{\eqx,\ediag\}}$ is \emph{well-diagonalized} if
it is of the form $\addediag{\AA}$ for some \emph{eq-respecting} $\AA \in \nData{2}{\extSigma \cup \{\eqx\}}$.
Note that then $\BB$ is eq-respecting, too.

\begin{example}
The data structure $\addediag{\AA}$ from Figure~\ref{fig:diagonal}(a) is well-diagonalized.
The one from Figure~\ref{fig:diagonal}(b) is not well-diagonalized
(in particular, it is not eq-respecting).
\end{example}

We will need a way to ensure that the considered data structures are well-diagonalized. To this end, we introduce the following
sentence from $\ndFOvar{2}{\extSigma \cup \{\eqx,\ediag\},\Binary_{df}}{2}$ :
$$
\begin{array}{lcl}
  \phied^\Ueq :=& &
  \bigwedge_{i \in \{1,2\}} \forall x.\exists y.(\ediag(y) \wedge \rels{i}{i}{x}{y}) \wedge \bigl(\forall x.\forall y.(\ediag(x) \wedge \ediag(y) \wedge \rels{i}{i}{x}{y}) \rightarrow x=y\bigr) \\[0.7ex]
  & \wedge\!\!\!\!\! &\forall x.\eqx(x) \leftrightarrow \exists y.(\ediag(y) \wedge \rels{1}{1}{x}{y} \wedge \rels{2}{2}{x}{y})\\[0.7ex]
  & \wedge\!\!\!\!\! &\forall x.\ediag(x) \rightarrow
  \bigwedge_{\unary \in \Ueq} \neg\unary(x)
\end{array}
$$

Informally, this formula ensures the following properties:
\begin{itemize}
  \item for each element, and each of its value, there exists an
    element labeled by $\ediag$ having the same value at the same
    position;
  \item there is no two distinct elements labeled by $\ediag$ and
    having an equal value at the same position;
  \item for all elements labeled by $\eqx$, there exists an element
    labeled by $\ediag$ having the two same values at the same positions;
   \item the elements labeled by $\ediag$ are not labeled by any other
     labels from $\extSigma$.
\end{itemize}

Every structure that is well-diagonalized satisfies $\phied^\Ueq$.
The converse is not true in general. In particular, a model of $\phied^\Ueq$ is not necessarily eq-respecting.
However, if a structure satisfies a formula $\phi \in
\ndFO{2}{\extSigma \cup \{\eqx,\ediag\},\Binary_{df}}$, then it is
possible to perform a permutation on the first (or the second) values
of its elements while preserving $\phi$. This allows us to get:

\begin{lemma}\label{lem-diagonalization}
  Let $\BB \in \nData{2}{\extSigma \cup \{\eqx,\ediag\}}$ and $\phi \in \ndFO{2}{\extSigma \cup \{\eqx,\ediag\},\Binary_{df}}$. If $\BB \models \phi \wedge \phied^\Ueq$, then there exists an eq-respecting $\AA \in \nData{2}{\extSigma \cup \{\eqx\}}$ such that $\addediag{\AA} \models \phi$. 
\end{lemma}

\begin{proof}
  Let $\BB=(A,(P_{\unary}),\ifunct,\ofunct)$ in $\nData{2}{\extSigma \cup \{\eqx,\ediag\}}$ such that $\BB \models \phi \wedge \phied$. We define the sets $I=\{n \in \N \mid \exists b \in P_\ediag . \ifunct(b)=n\}$ and $O=\{n \in \N \mid \exists b \in P_\ediag . \ofunct(b)=n\}$. Since
$$\BB \models \bigwedge_{i \in \{1,2\}} \bigl[\big(\forall x.\exists y. \ediag(y) \wedge \rels{i}{i}{x}{y}\big) \wedge \big(\forall x.\forall y.(\ediag(x) \wedge \ediag(y) \wedge \rels{i}{i}{x}{y}) \rightarrow x=y \big)\bigr]\,,$$
  we deduce that $|I|=|O|$ and furthermore the mapping $\pi : I
  \mapsto O$ defined by $\pi(n)=m$ iff there exists $b \in P_\ediag$
  such that $\ifunct(b)=n$ and $\ofunct(b)=m$ is a well-defined
  bijection. It is well defined because there is a single element $b$
  in $P_\ediag$ such that $\ifunct(b)=n$ and it is a bijection because
  for all $m \in O$, there is a single $b \in P_\ediag$ such that
  $\ofunct(b)=m$.   We can consequently extend $\pi$ to be a
  permutation from $\N$ to $\N$. We then take the model $\AA'=(A,(P_{\unary}),\pi
  \circ \ifunct,\ofunct)$.  Since $\phi \wedge \phied \in \ndFO{2}{\extSigma \cup
  \{\eqx,\ediag\},\Binary_{df}}$ with $\Binary_{df}=\{\relsaord{1}{1},\relsaord{2}{2}\}$
  and since $\BB \models \phi \wedge \phied$, we deduce that $\AA' \models \phi
  \wedge \phied$
  because performing a permutation on the first data values of the
  elements of $\BB$ does not affect the satisfaction of $\phi \wedge \phied$ (this
  is a consequence of the fact that there is no comparison between the
  first values and the second values of the elements). The
  satisfaction  of $\phied$ by $\AA'$ allows us to deduce that $\AA'$ is
  well-diagonalized. We can in  fact safely remove  from $\AA'$ the elements
  of $P_\ediag$ to obtain a structure $\AA \in \nData{2}{\extSigma \cup
    \{\eqx\}}$ which is eq-respecting ( this is due to the fact that $\AA' \models \forall
  x.\eqx(x) \leftrightarrow \exists y.(\ediag(y) \wedge
  \rels{1}{1}{x}{y} \wedge \rels{2}{2}{x}{y}) \wedge \forall x.\ediag(x) \rightarrow
  \bigwedge_{\unary \in \Ueq \setminus \{\eqx\}} \neg\unary(x)$) and
  such that $\AA'=\addediag{\AA}$ .
\end{proof}

\begin{example}
Consider Figure~\ref{fig:diagonal} and let $\Theta = \emptyset$.
The data structure from Figure~\ref{fig:diagonal}(b)
satisfies $\phied^\Ueq$, though it is not
well-diagonalized. Suppose it also satisfies
$\phi \in \ndFO{2}{\{\eqx,\ediag\},\Binary_{df}}$.
By permutation of the first data values ( $1\mapsto 2, 2
 \mapsto 4, 3\mapsto 1, 4 \mapsto 3$), we obtain the
well-diagonalized data structure in Figure~\ref{fig:diagonal}(a).
As $\phi$ does not talk about the diagonal relation,
satisfaction of $\phi$ is preserved.
\end{example}

\newcommand{\tred}[1]{\sem{#1}_{{+}\ediag}}

Finally, we can inductively translate $\phi \in \ndFO{2}{\extSigma \cup \{\eqx\},\emptyset}$ into a formula $\tred{\phi} \in \ndFO{2}{\extSigma \cup \{\eqx,\ediag\},\emptyset}$ that avoids the extra `diagonal' elements: $\tred{\Pform{\unary}{x}} = \Pform{\unary}{x}$,
$\tred{ x=y } = (x=y) $, $\tred{\exists x.\vp}=\exists x.(\neg\Pform{\ediag}{x} \wedge \tred{\vp})$, $\tred{ \vp \vee \vp'}=\tred{\vp} \vee \tred{\vp'}$, and $\tred{\neg \vp}=\neg \tred{\vp}$.
We immediately obtain:

\begin{lemma}\label{lem:new-Tdiag}
 Let $\AA \in \nData{2}{\extSigma \cup \{\eqx\}}$ and
 $\phi \in \ndFO{2}{\extSigma \cup \{\eqx\},\emptyset}$ be a sentence.
  We have $\AA \models \phi$ iff $\addediag{\AA} \models \tred{\phi}$.
\end{lemma}

\noindent \textbf{Step 3: Getting Rid Of the Diagonal Relation}\\

We will now exploit
well-diagonalized data structures to reason about environments
relative to $\Gamma$ in terms of environments relative to $\Gamma_{df}$.
Recall that $\extSigma$ ranges over finite sets such that $\Sigma
\subseteq \extSigma$.

\begin{lemma}\label{lem:new-env-with-no-diag}
Let $\AA=(A,,(P_{\unary}),\ifunct,\ofunct) \in \nData{2}{\extSigma \cup \{\eqx\}}$
be eq-respecting and $\BB = \addediag{\AA}$. Moreover, let $a \in A$, $U \subseteq \Sigma$, and $R \subseteq \Gamma$ be a nonempty set. We have
$\mathtt{Env}_{\AA,\Unary,\Binary}(a,\USet,\BSet) =$
\[
\begin{cases}
\mathtt{Env}_{\BB,\Unary,\Binary_{df}}(a,\USet,\Binary_{df}) \setminus P_\ediag & \text{if } a \in P_\eqx \text{ and } \makebox[10em][l]{$\BSet=\Binary$} (1)
\\
\mathtt{Env}_{\BB,\Unary,\Binary_{df}}(a,\USet,\Binary_{df}) & \text{if } a \notin P_\eqx \text{ and } \makebox[10em][l]{$\BSet=\Binary_{df}$} (2)
\\
\mathtt{Env}_{\BB,\Unary,\Binary_{df}}(a,\USet,\{\relsaord{1}{1}\}) \cap (P_\eqx\setminus P_\ediag) & \text{if } a \notin P_\eqx \text{ and } \makebox[10em][l]{$\BSet=\{\relsaord{1}{1},\relsaord{1}{2}\}$} (3)
\\
\mathtt{Env}_{\BB,\Unary,\Binary_{df}}(a,\USet,\{\relsaord{2}{2}\}) & \text{if } a \in P_\eqx \text{ and } \makebox[10em][l]{$\BSet=\{\relsaord{2}{2},\relsaord{1}{2}\}$} (4)
\\
\mathtt{Env}_{\BB,\Unary,\Binary_{df}}(a,\USet,\{\relsaord{2}{2}\}) \setminus P_\ediag & \text{if } a \notin P_\eqx \text{ and } \makebox[10em][l]{$\BSet=\{\relsaord{2}{2}\}$} (5)
\\
\mathtt{Env}_{\BB,\Unary,\Binary_{df}}(a,\USet,\{\relsaord{1}{1}\}) \setminus P_\eqx & \text{if } \hspace{5.35em} \makebox[10em][l]{$\BSet=\{\relsaord{1}{1}\}$} (6)
\\
\mathtt{Env}_{\BB,\Unary,\Binary_{df}}(d,\USet,\{\relsaord{2}{2}\}) & \text{if } a \notin P_\eqx \text{ and } \makebox[10em][l]{$\BSet=\{\relsaord{1}{2}\}$} (7)\\[-0.5ex]
\quad \text{for the unique } d \in P_\ediag \text{ such that } \relsaa{1}{1}{\BB}{d}{a}
\\
\emptyset & \text{otherwise } \hspace{11.9em} (8)
\end{cases}
\]
\end{lemma}

\begin{example}
\label{ex:env-with-no-diag}
Before to provide the formal proof, let us go through the different cases of Lemma~\ref{lem:new-env-with-no-diag} using Figure~\ref{fig:diagonal}(a),
and letting $\Sigma = \extSigma = \emptyset$.
\begin{enumerate}[(1)]
\item Let $a = a_1$ and $R = \Gamma$. Then, $\mathtt{Env}_{\AA,\Unary,\Binary}(a,\emptyset,\BSet) = \{a_1,a_2,a_3\}$.
We also have that
$\mathtt{Env}_{\BB,\Unary,\Binary_{df}}(a,\emptyset,\Binary_{df}) = \{a_1,a_2,a_3,b_1\}$: These are the elements
that coincide with $a$ \emph{exactly} on the first and the on the second data value when we dismiss the diagonal relation.
Of course, as we consider $\BB$, this includes $b_1$, which we have to exclude.
Thus, $\mathtt{Env}_{\AA,\Unary,\Binary}(a,\emptyset,\BSet) = \mathtt{Env}_{\BB,\Unary,\Binary_{df}}(a,\emptyset,\Binary_{df}) \setminus P_\ediag$.
\item Let $a = a_7$ and $\BSet=\Binary_{df}$. Then, $\mathtt{Env}_{\AA,\Unary,\Binary}(a,\emptyset,\BSet) = \mathtt{Env}_{\BB,\Unary,\Binary_{df}}(a,\emptyset,\Binary_{df}) = \{a_7\}$. Since $a \not\in P_\eqx$, it actually does not matter whether we include the diagonal relation or not.

\item Let $a = a_7$ and $\BSet=\{\relsaord{1}{1},\relsaord{1}{2}\}$. Then, $\mathtt{Env}_{\AA,\Unary,\Binary}(a,\emptyset,\BSet) = \{a_1,a_2,a_3\}$.
So how do we get this set in $\BB$ without referring to the diagonal relation?
The idea is to use only $\relsaord{1}{1} \in \Gamma_{df}$ and to ensure data equality by restricting to elements in $P_{\eqx}$ (again excluding $P_\ediag$).
Indeed, we have $\mathtt{Env}_{\BB,\Unary,\Binary_{df}}(a,\emptyset,\{\relsaord{1}{1}\}) \cap (P_\eqx\setminus P_\ediag) = \{a_1,a_2,a_3,b_1\} \cap (\{a_1,a_2,a_3,b_1\} \setminus \{b_1\}) = \{a_1,a_2,a_3\}$.

\item Let $a = a_1$ and $\BSet=\{\relsaord{2}{2},\relsaord{1}{2}\}$. Then, $\mathtt{Env}_{\AA,\Unary,\Binary}(a,\emptyset,\BSet) = \{a_4,a_5\}$.
So we are looking for elements that have $1$ as the second data value and a first data value different from~1,
and this set is exactly $\mathtt{Env}_{\BB,\Unary,\Binary_{df}}(a,\emptyset,\{\relsaord{2}{2}\})$.

\item Let $a = a_5$ and $\BSet=\{\relsaord{2}{2}\}$. Then, $\mathtt{Env}_{\AA,\Unary,\Binary}(a,\emptyset,\BSet) = \{a_1,a_2,a_3,a_4\}$,
which is the set of elements that have 1 as the second data value and a first data value different from~2.
Thus, this is exactly $\mathtt{Env}_{\BB,\Unary,\Binary_{df}}(a,\emptyset,\{\relsaord{2}{2}\}) \setminus P_\ediag$ (i.e., after discarding $b_1 \in P_\ediag$).

\item Let $a = a_4$ and $\BSet=\{\relsaord{1}{1}\}$. We have $\mathtt{Env}_{\AA,\Unary,\Binary}(a,\emptyset,\BSet) = \{a_8\}$.
Looking at $\BB$ and discarding the diagonal relation would also include $b_3$ and any element with data-value
pair $(3,3)$. Discarding $P_\eqx$, we obtain
$\mathtt{Env}_{\BB,\Unary,\Binary_{df}}(a,\emptyset,\{\relsaord{1}{1}\}) \setminus P_\eqx = \{a_8,b_3\} \setminus \{b_3\} = \{a_8\}$.

\item Let $a = a_7$ and $\BSet=\{\relsaord{1}{2}\}$. Then, $\mathtt{Env}_{\AA,\Unary,\Binary}(a,\emptyset,\BSet) = \{a_4,a_5\}$,
which is the set of elements whose second data value is 1 and whose first data value is different from 1.
The idea is now to change the reference point. Take the unique $d \in P_\ediag$ such that
$\relsaa{1}{1}{\BB}{d}{a}$. Thus, $d = b_1$. The set
$\mathtt{Env}_{\BB,\Unary,\Binary_{df}}(b_1,\emptyset,\{\relsaord{2}{2}\})$ gives us exactly the elements that
have 1 as the second data value and a first value different from 1, as desired.
\end{enumerate}
\end{example}

\begin{proof}
   Let $\AA=(A,(P_{\unary}),\ifunct,\ofunct) \in \nData{2}{\extSigma \cup \{\eqx\}}$
be eq-respecting and $\BB = \addediag{\AA}$. We consider $a \in A$, $U \subseteq \Sigma$, and $R \subseteq \Gamma$ be a nonempty set. 
  Note that by definition of $\mathtt{Env}$, we have
  $\mathtt{Env}_{\AA,\Unary,\Binary}(a,\USet,R)=\mathtt{Env}_{\BB,\Unary,\Binary}(a,\USet,R)
  \setminus P_\ediag$ and when $R \neq \{\relsaord{1}{2}\}$, we have
  $\mathtt{Env}_{\BB,\Unary,\Binary_{df}}(a,\USet,R \setminus
  \{(1,2)\})=\mathtt{Env}_{\BB,\Unary,\Binary}(a,\USet,R \cup \{\relsaord{1}{2}\})
  \cup \mathtt{Env}_{\BB,\Unary,\Binary}(a,\USet,R
  \setminus\{\relsaord{1}{2}\})$. We will use these two equalities in the rest of
  the proof.
	We now perform a case analysis on $\BSet$. 
	\begin{enumerate}
	\item 
		Assume $\BSet=\Binary=\{\relsaord{1}{1},\relsaord{2}{2},\relsaord{1}{2}\}$. First we suppose
        that $a
        \notin P_\eqx$. Since $\AA$ is eq-respecting it implies that
        $\ifunct(a)\neq \ofunct(a)$. Now assume there exists $b \in
        \mathtt{Env}_{\AA,\Unary,\Binary}(a,\USet,\BSet)$, this means
        that $\relsaa{1}{1}{\AA}{a}{b}$ and $\relsaa{2}{2}{\AA}{a}{b}$
        and $\relsaa{1}{2}{\AA}{a}{b}$. Hence we have
        $\ofunct(a)=\ofunct(b)$ and  $\ifunct(a)=\ofunct(b)$, which is
        a contradiction. Consequently
        $\mathtt{Env}_{\AA,\Unary,\Binary}(a,\USet,\BSet)=\emptyset$. We
        now suppose that $a \in
        P_\eqx$. In that case, since $\AA$ is eq-respecting, we have
        $\mathtt{Env}_{\BB,\Unary,\Binary}(a,\USet,R \setminus
        \{\relsaord{1}{2}\})=\emptyset$. In fact if
        $\relsaa{2}{2}{\AA}{a}{b}$ for some $b$ then
        $\relsaa{1}{2}{\AA}{a}{b}$ as $a \in P_\eqx$. Hence we have  $\mathtt{Env}_{\AA,\Unary,\Binary}(a,\USet,R)=\mathtt{Env}_{\BB,\Unary,\Binary}(a,\USet,R)
  \setminus P_\ediag =\mathtt{Env}_{\BB,\Unary,\Binary_{df}}(a,\USet,\Binary_{df}) \setminus P_\ediag $.

\item 	Assume $\BSet=\Binary_{df}=\{\relsaord{1}{1},\relsaord{2}{2}\}$. By  a similar
  reasoning as  the previous  case, if  $a \in P_\eqx$, we have
  necessarily  $\mathtt{Env}_{\AA,\Unary,\Binary}(a,\USet,\BSet) =
  \emptyset$. Now suppose $a \notin P_\eqx$. Thanks to this
  hypothesis, we know that
  $P_\ediag \cap
  \mathtt{Env}_{\BB,\Unary,\Binary_{df}}(a,\USet,\Binary_{df})
  =\emptyset$ and that $\mathtt{Env}_{\BB,\Unary,\Binary}(a,\USet,\{\relsaord{1}{1},\relsaord{2}{2},\linebreak[0](1,2)\})=\emptyset$. Hence we obtain directly $\mathtt{Env}_{\AA,\Unary,\Binary}(a,\USet,\Binary_{df}
  )= \mathtt{Env}_{\BB,\Unary,\Binary_{df}}(a,\USet,\Binary_{df})$.

  \item Assume $\BSet=\{\relsaord{1}{1},\relsaord{1}{2}\}$. Again it is obvious that if
    $a \in P_\eqx$, we have $\mathtt{Env}_{\AA,\Unary,\Binary}(a,\USet,\BSet) =
  \emptyset$. We suppose that  $a \notin P_\eqx$. Note that we have that
  $\mathtt{Env}_{\BB,\Unary,\Binary}(a,\USet,R) \subseteq
  P_\eqx$ and $\mathtt{Env}_{\BB,\Unary,\Binary}(a,\USet,\{\relsaord{1}{1}\}) \cap P_\eqx =\emptyset$ . Since $\mathtt{Env}_{\BB,\Unary,\Binary_{df}}(a,\USet,\{\relsaord{1}{1}s\})=
  \mathtt{Env}_{\BB,\Unary,\Binary}(a,\USet,R) \cup \mathtt{Env}_{\BB,\Unary,\Binary}(a,\USet,\{(1,1)\})$, we deduce that
  $\mathtt{Env}_{\BB,\Unary,\Binary_{df}}(a,\USet,
  \{(1,1)\}) \cap P_\eqx=\mathtt{Env}_{\BB,\Unary,\Binary}(a,\USet,R)
  $. From which  we  can conclude
  $\mathtt{Env}_{\AA,\Unary,\Binary}(a,\USet,\BSet)
  =\mathtt{Env}_{\BB,\Unary,\Binary_{df}}(a,\USet,\{\relsaord{1}{1}\}) \cap
  P_\eqx\setminus P_\ediag$.
\item Assume $\BSet=\{\relsaord{2}{2},\relsaord{1}{2}s\}$. Again, it is obvious that if
    $a \notin P_\eqx$, we have $\mathtt{Env}_{\AA,\Unary,\Binary}(a,\USet,\BSet) =
  \emptyset$. We  now suppose that $a \in P_\eqx$. In that case, we have
  that
  $\mathtt{Env}_{\BB,\Unary,\Binary}(a,\USet,R\setminus\{\relsaord{1}{2}\}) =
  \emptyset$ and furthermore
  $\mathtt{Env}_{\BB,\Unary,\Binary}(a,\USet,R) \cap
  P_\ediag=\emptyset$. We can hence conclude  that $\mathtt{Env}_{\AA,\Unary,\Binary}(a,\USet,\BSet)
  =\mathtt{Env}_{\BB,\Unary,\Binary_{df}}(a,\USet,\{\relsaord{2}{2}\})$.
\item Assume $\BSet=\{\relsaord{2}{2}\}$. As before if
    $a \in P_\eqx$, we have $\mathtt{Env}_{\AA,\Unary,\Binary}(a,\USet,\BSet) =
  \emptyset$. We now suppose that $a \notin P_\eqx$. In that case, we
  have immediately
  $\mathtt{Env}_{\BB,\Unary,\Binary}(a,\USet,\{\relsaord{2}{2},\relsaord{1}{2}\})=\emptyset$
  and consequently $$\mathtt{Env}_{\AA,\Unary,\Binary}(a,\USet,\BSet)
  =\mathtt{Env}_{\BB,\Unary,\Binary}(a,\USet,\{\relsaord{2}{2}\}) \setminus P_\ediag
  =\mathtt{Env}_{\BB,\Unary,\Binary_{df}}(a,\USet,\{\relsaord{2}{2}\})\setminus
  P_\ediag\,.$$
\item Assume $\BSet=\{\relsaord{1}{1}\}$.  Remember that we have
  $\mathtt{Env}_{\BB,\Unary,\Binary_{df}}(a,\USet,
  \{(1,1)\})=\mathtt{Env}_{\BB,\Unary,\Binary}(a,\USet,\linebreak[0]\{\relsaord{1}{1},\relsaord{1}{2}\})
  \cup \mathtt{Env}_{\BB,\Unary,\Binary}(a,\USet,\{\relsaord{1}{1}\})$. But
  $\mathtt{Env}_{\BB,\Unary,\Binary}(a,\USet,\{\relsaord{1}{1},\relsaord{1}{2}\})
  \subseteq P_\eqx$ and, moreover,
  $\mathtt{Env}_{\BB,\Unary,\Binary}(a,\USet,\{\relsaord{1}{1}\}) \cap
  P_\eqx=\emptyset$. We hence deduce that $\mathtt{Env}_{\BB,\Unary,\Binary_{df}}(a,\USet,
  \{\relsaord{1}{1}\}) \setminus
  P_\eqx=\mathtt{Env}_{\BB,\Unary,\Binary}(a,\USet,\{\relsaord{1}{1}\})$ and
  since $(\mathtt{Env}_{\BB,\Unary,\Binary_{df}}(a,\USet,
  \{\relsaord{1}{1}\}) \setminus
  P_\eqx) \cap P_\ediag=\emptyset$, we obtain $\mathtt{Env}_{\AA,\Unary,\Binary}(a,\USet,\BSet)
  =\mathtt{Env}_{\BB,\Unary,\Binary_{df}}(a,\USet,
  \{\relsaord{1}{1}\}) \setminus
  P_\eqx$.
\item Assume $\BSet=\{\relsaord{1}{2}\}$. Again it is obvious that if
    $a \in P_\eqx$, we have $\mathtt{Env}_{\AA,\Unary,\Binary}(a,\USet,\BSet) =
  \emptyset$. We now suppose that $a \notin P_\eqx$. By definition,
  since $\BB=\addediag{\AA}$, in $\BB$ there is a unique $d \in
  P_\ediag \text{ such that } \relsaa{1}{1}{\BB}{d}{a}$. We have then
  $\mathtt{Env}_{\BB,\Unary,\Binary}(a,\USet,R)=\mathtt{Env}_{\BB,\Unary,\Binary}(d,\USet,\{\relsaord{1}{2},\relsaord{2}{2}\})$. As
for the case 4., we deduce that
$\mathtt{Env}_{\BB,\Unary,\Binary}(d,\USet,\{\relsaord{1}{2},\relsaord{2}{2}\})=\mathtt{Env}_{\BB,\Unary,\Binary_{df}}(d,\USet,\{\relsaord{2}{2}\})$. Hence $\mathtt{Env}_{\BB,\Unary,\Binary}(a,\USet,R)=\mathtt{Env}_{\BB,\Unary,\Binary_{df}}(d,\USet,\{\relsaord{2}{2}\})$. \qedhere
	\end{enumerate}
\end{proof}

Let us wrap up:
By Lemmas~\ref{lem:new-fo-count} and \ref{lem:new-env-with-no-diag},
we end up with checking counting constraints in an extended data structure
without using the diagonal relation.\\

\newcommand{\gammarel}[1]{\gamma_{#1}}
\newcommand{\alpharel}[2]{\alpha_{#1}^{#2}}
\newcommand{\betarel}[2]{\beta_{#1}^{#2}}
\newcommand{\gammaform}{\phi_\gamma}
\newcommand{\alphaform}{\phi_\gamma}
\newcommand{\labels}[1]{\ell_\Unary(#1)}

\noindent\textbf{Step 4: Counting in Two-Variable Logic}\\

The next step is to express these constraints using two-variable formulas.
Counting in two-variable logic is established using further unary predicates.
These additional predicates allow us to define a partitioning of the universe of
a structure into so-called \emph{intersections}.
Suppose $\AA=(A,(P_{\unary}),\ifunct,\ofunct) \in \nData{2}{\extSigma \cup \{\eqx,\ediag\}}$, where $\Sigma \subseteq \extSigma$.
Let $a \in A \setminus P_\ediag$ and define $\labels{a} = \{\sigma \in \Sigma \mid a \in P_\sigma\}$.
The \emph{intersection} of $a$ in $\AA$ is the set $\{b \in A \setminus P_\ediag \mid \rels{1}{1}{a}{b} \mathrel{\wedge} \rels{2}{2}{a}{b} \mathrel{\wedge} \labels{a} = \labels{b}\}$. 
A set is called an \emph{intersection} in $\AA$ if it is the intersection of some $a \in A \setminus P_\ediag$.

\begin{example}
Consider Figure~\ref{fig:intersection} and suppose $\Sigma = \{\mathsf{p}\}$.
The intersections of the given data structure are gray-shaded.
\end{example}

\begin{figure*}[t]
\centering
\scalebox{0.85}{
\begin{tabular}{c}
\begin{tikzpicture}[node distance=2cm]
  \pgfsetxvec{\pgfpoint{1.2cm}{0cm}}
  \pgfsetyvec{\pgfpoint{0cm}{1.2cm}}
  \pgfsetzvec{\pgfpoint{0cm}{1.2cm}}

	\node []	(b) at (10,0)	{};

	\node [datayellowyellow, label=below:$\ediag$]	(A) at (-4.5,0.3)	{1 \nodepart{second} 1};
	\node [datanone, label=below:$\eqx$]	(A) at (-4.5,-0.1)	{};	

	\node [datagreengreen, label=below:$\ediag$]	(A) at (-3.8,0.3)	{2 \nodepart{second} 2};
	\node [datanone, label=below:$\eqx$]	(A) at (-3.8,-0.1)	{};	

	\node [datablueblue, label=below:$\ediag$]	(A) at (-3.1,0.3)	{3 \nodepart{second} 3};
	\node [datanone, label=below:$\eqx$]	(A) at (-3.1,-0.1)	{};	

	\node [dataredred, label=below:$\ediag$]	(A) at (-2.4,0.3)	{4 \nodepart{second} 4};
	\node [datanone, label=below:$\eqx$]	(A) at (-2.4,-0.1)	{};	


	\node [draw, fill=gray!10, rectangle, minimum width=13.9cm, minimum height=3cm, rounded corners=0.2cm] (cc) at (3.9,-0.25) {};	


\put(-10,0){
	\node [label=right:$\textcolor{red}{\alpharel{3}{1}}$, draw, fill=gray!20, rectangle, minimum width=3.5cm, minimum height=2.7cm, rounded corners=0.2cm] (cc) at (0.05,-0.25) {};	
	
	\node [datayellowyellow]	(A) at (-1,0.3)	{1 \nodepart{second} 1};
	\node [datanone, label=below:$\mathsf{p}$]	(A) at (-1,0.2)	{};
	\node [datanone, label=below:$\eqx$]	(A) at (-1,-0.1)	{};
	\node [datanone, label=below:$\textcolor{blue}{\gammarel{1}}$]	(A) at (-1,-0.5)	{};		

	\node [datayellowyellow]	(A) at (-0.3,0.3)	{1 \nodepart{second} 1};
	\node [datanone, label=below:$\mathsf{p}$]	(A) at (-0.3,0.2)	{};
	\node [datanone, label=below:$\eqx$]	(A) at (-0.3,-0.1)	{};
	\node [datanone, label=below:$\textcolor{blue}{\gammarel{2}}$]	(A) at (-0.3,-0.5)	{};		

	\node [datayellowyellow]	(A) at (0.4,0.3)	{1 \nodepart{second} 1};
	\node [datanone, label=below:$\mathsf{p}$]	(A) at (0.4,0.2)	{};
	\node [datanone, label=below:$\eqx$]	(A) at (0.4,-0.1)	{};
	\node [datanone, label=below:$\textcolor{blue}{\gammarel{3}}$]	(A) at (0.4,-0.5)	{};		

	\node [datayellowyellow]	(A) at (1.1,0.3)	{1 \nodepart{second} 1};
	\node [datanone, label=below:$\mathsf{p}$]	(A) at (1.1,0.2)	{};
	\node [datanone, label=below:$\eqx$]	(A) at (1.1,-0.1)	{};
	\node [datanone, label=below:$\textcolor{blue}{\gammarel{3}}$]	(A) at (1.1,-0.5)	{};		
}


\put(116,0){
	\node [label=right:$\textcolor{red}{\alpharel{3}{2}}$, draw, fill=gray!20, rectangle, minimum width=2.7cm, minimum height=2.7cm, rounded corners=0.2cm] (cc) at (-0.3,-0.25) {};	
	
	\node [datayellowgreen]	(A) at (-1,0.3)	{1 \nodepart{second} 2};
	\node [datanone, label=below:$\mathsf{p}$]	(A) at (-1,0.2)	{};
	\node [datanone, label=below:$\textcolor{blue}{\gammarel{1}}$]	(A) at (-1,-0.2)	{};		

	\node [datayellowgreen]	(A) at (-0.3,0.3)	{1 \nodepart{second} 2};
	\node [datanone, label=below:$\mathsf{p}$]	(A) at (-0.3,0.2)	{};
	\node [datanone, label=below:$\textcolor{blue}{\gammarel{2}}$]	(A) at (-0.3,-0.2)	{};		

	\node [datayellowgreen]	(A) at (0.4,0.3)	{1 \nodepart{second} 2};
	\node [datanone, label=below:$\mathsf{p}$]	(A) at (0.4,0.2)	{};
	\node [datanone, label=below:$\textcolor{blue}{\gammarel{3}}$]	(A) at (0.4,-0.2)	{};		
}


\put(222,0){
	\node [label=right:$\textcolor{red}{\alpharel{2}{3}}$, draw, fill=gray!20, rectangle, minimum width=2cm, minimum height=2.7cm, rounded corners=0.2cm] (cc) at (-0.65,-0.25) {};	
	
	\node [datayellowblue]	(A) at (-1,0.3)	{1 \nodepart{second} 3};
	\node [datanone, label=below:$\mathsf{p}$]	(A) at (-1,0.2)	{};
	\node [datanone, label=below:$\textcolor{blue}{\gammarel{1}}$]	(A) at (-1,-0.2)	{};		

	\node [datayellowblue]	(A) at (-0.3,0.3)	{1 \nodepart{second} 3};
	\node [datanone, label=below:$\mathsf{p}$]	(A) at (-0.3,0.2)	{};
	\node [datanone, label=below:$\textcolor{blue}{\gammarel{2}}$]	(A) at (-0.3,-0.2)	{};		
}


\put(305,0){
	\node [label=right:$\textcolor{red}{\alpharel{2}{4}}$, draw, fill=gray!20, rectangle, minimum width=2cm, minimum height=2.7cm, rounded corners=0.2cm] (cc) at (-0.65,-0.25) {};	
	
	\node [datayellowred]	(A) at (-1,0.3)	{1 \nodepart{second} 4};
	\node [datanone, label=below:$\mathsf{p}$]	(A) at (-1,0.2)	{};
	\node [datanone, label=below:$\textcolor{blue}{\gammarel{1}}$]	(A) at (-1,-0.2)	{};		

	\node [datayellowred]	(A) at (-0.3,0.3)	{1 \nodepart{second} 4};
	\node [datanone, label=below:$\mathsf{p}$]	(A) at (-0.3,0.2)	{};
	\node [datanone, label=below:$\textcolor{blue}{\gammarel{2}}$]	(A) at (-0.3,-0.2)	{};		
}

\end{tikzpicture}
\end{tabular}
}
\caption{Counting intersections for $M = 3$ and elements with label $\mathsf{p}$
\label{fig:intersection}}
\end{figure*}

Let us introduce the various unary predicates, which will be assigned
to \emph{non-diagonal} elements. There are three types of them
(for the first two types, also see Figure~\ref{fig:intersection}):

\begin{enumerate}
\item The unary predicates $\Lambda_M^\gamma = \{\gammarel{1},\ldots,\gammarel{M}\}$ have the following intended meaning:
For all intersections $I$ and $i \in \{1,\ldots,M\}$, we have $|I| \ge i$ iff
there is $a \in I$ such that $a \in P_{\gammarel{i}}$. In other words, the presence
(or absence) of $\gammarel{i}$ in an intersection $I$ tells us whether $|I| \ge i$.

\item The predicates $\Lambda_M^\alpha = \{\alpharel{i}{j} \mid i \in \{1,\ldots,M\}$ and
$j \in \{1,\ldots,M+2\}\}$ have the following meaning:
If $a$ is labeled with $\alpharel{i}{j}$, then (i) there are at least $j$ intersections sharing the same first value and the same label set $\labels{a}$, and (ii) the intersection of $a$ has $i$ elements if $i \le M-1$ and at least $M$ elements if $i=M$. Hence, in $\alpharel{i}{j}$, index $i$ counts the elements inside an intersection, and $j$ labels up to $M+2$ different intersections.
We need to go beyond $M$ due to Lemma~\ref{lem:new-env-with-no-diag}:
When we remove certain elements (e.g., $P_\eqx$)
from an environment, we must be sure to still have sufficiently many to be able to count until $M$.

\item Labels from $\Lambda_M^\beta = \{\betarel{i}{j} \mid i \in \{1,\ldots,M\}$ and
$j \in \{1,\ldots,M+1\}\}$ will play a similar role as those in $\Lambda_M^\alpha$
but consider the second values of the elements instead of the first ones.
\end{enumerate}

\begin{example}
A suitable labeling for types $\gamma$ and $\alpha$ is illustrated in Figure~\ref{fig:intersection} for $M = 3$.
\end{example}

Let $\Lambda_M = \Lambda_M^\alpha \cup \Lambda_M^\beta \cup \Lambda_M^\gamma$ denote the set of all these unary predicates.
We shall now see how to build sentences
$\phi_\alpha, \phi_\beta, \phi_\gamma \in \ndFOvar{2}{\extSigma \cup \{\eqx,\ediag\} \cup \Lambda_M,\Binary_{df}}{2}$
that guarantee the respective properties.

To deal with the predicates in $\Lambda_M^\gamma$, we  first  define
the formula $\phiint = \rels{1}{1}{x}{y} \mathrel{\wedge}
\rels{2}{2}{x}{y} \mathrel{\wedge} \bigwedge_{\sigma \in \Unary}
\sigma(x) \leftrightarrow \sigma(y)$ and introduce the
following formulas:
\[
\begin{array}{lcl}
  \phi^1_\gamma (x)&:= & \displaystyle \bigvee_{i \in [1,M]} \big(\gamma_i(x) \wedge \bigwedge_{j \in [1,M]\setminus \{i\}} \neg \gamma_j(x)\big) \\
  &&\\
  \phi^2_\gamma(x) &:= &\displaystyle \bigwedge_{i \in [1,M-1]}
                         \big(\gamma_i(x) \rightarrow \neg \exists
                         y. ( x \neq y \wedge \phiint(x,y)  \wedge \gamma_i(y))\big) \\
                &&\\
  \phi^3_\gamma(x) &:= &\displaystyle \bigwedge_{i \in [2,M]} \big(\gamma_i(x) \rightarrow (\exists y. \phiint(x,y) \wedge \gamma_{i-1}(y))\big)\\
\end{array}
\]
We then let $\phi_\gamma:=\forall x.\big( \neg \ediag(x) \rightarrow
(\phi^1_\gamma(x) \wedge \phi^2_\gamma(x) \wedge \phi^3_\gamma(x)) \big)
\wedge( \ediag(x) \rightarrow \bigwedge_{\gamma \in
  \Lambda_M^\gamma}\neg \gamma(x))$. Thus, a data structure satisfies
$\phi_\gamma$ if no diagonal element is labelled with predicates in
$\Lambda_M^\gamma$ and (1)~all its non-diagonal elements are labelled with
exactly one predicate in $\Lambda_M^\gamma$ (see $\phi^1_\gamma$), (2)
if $i \le M-1$, then there are no two 
$\gamma_i$-labelled elements with the same labels of $\Unary $ and in the same intersection (see  $\phi^2_\gamma$), and
(3) if $i \ge 2$, then for all $\gamma_i$-labelled elements, there exists an $\gamma_{i-1}$-labelled element with the same labels of $\Unary$ and in the same intersection (see  $\phi^3_\gamma$).

\begin{lemma}\label{lem:gamma-count}

Let $\AA=(A,(P_{\unary}),\ifunct,\ofunct,) \in \nData{2}{\Sigma \cup \{\eqx\} \cup \myCountConst{}{}{M} \cup \Lambda_M}$
  be eq-respecting
  and such that $\addediag{\AA} \models \phi_\gamma $.
We consider $a \in A$ and $\gamma_i \in \Lambda_M$ and $E_a=\{b \in A \mid \rels{1}{1}{a}{b} \mathrel{\wedge} \rels{2}{2}{a}{b} \mathrel{\wedge} \labels{a} = \labels{b}\}$. Then, $|E_a|\geq i$ iff there exists $b\in E_a$ such that $b\in P_{\gamma[i]}$.
\end{lemma}

\begin{proof}
  For any $b\in E_a$, as $\addediag{\AA}\models\phi^1_\gamma$ there is exactly one $j\in [1,M]$ such that $b\in P_{\gamma_j}$. 
  This allow us to build the function $f:E_a\to[1,M]$ which associates
  to any $b\in E_a$ such a $j$.
  Let $J= \{f(b)|b\in E_a\}$ denotes the image of $E_a$ under $f$.
  As $\addediag{\AA}\models\phi^3_\gamma$, for any $j\in [2,M]$ if $j\in J$ then $j-1 \in J$.
  And as $E_a\neq\emptyset$, there is $\jmax\in[1,M]$ such that $J=[1,\jmax]$.
  We can now rephrase our goal as $|E_a|\geq i$ iff $i\in J$.
  Assuming that $i\in J$, we have $i\leq\jmax$.
  As $f$ is a function, we have $|E_a|\geq |J|$. As $|J|=\jmax$, we have that $|E_a|\geq i$.
  Conversely, assuming that $|E_a|\geq i$.  Assume by contradiction that $i\notin J$, then $\jmax<i\leq M$.
  That is, for all $j\in J$, we have $j<M$.
  Since  $\addediag{\AA}\models \phi^2_\gamma$, all elements of $J$ have exactly one preimage.
  So $|E_a|=|J|=\jmax<i$, which contradicts the assumption.
\end{proof}

It is then easy to see that, in an intersection, if there is an
element $a$ labelled by $\gamma_i$ and no element labelled by
$\gamma_{i+1}$ for $i < M$, then the intersection has exactly $i$
elements; moreover, if there is a node $a$ labelled by $\gamma_M$ then
the intersection has at least $M$ elements.

We now show how we use the predicates in $\Lambda_M^\alpha$ and
introduce the following formulas (where $\phiint = \rels{1}{1}{x}{y} \mathrel{\wedge}
\rels{2}{2}{x}{y} \mathrel{\wedge} \bigwedge_{\sigma \in \Unary}
\sigma(x) \leftrightarrow \sigma(y)$ and $\phi_{same}= \bigwedge_{\sigma \in \Unary}
\sigma(x) \leftrightarrow \sigma(y)$):
$$
\begin{array}{lcl}
  \phi^1_\alpha(x) &:= &\displaystyle\bigvee_{\scriptstyle i \in [1,M]\atop \scriptstyle j \in [1,M+2]} \bigg(\alpha^j_i(x) \wedge \bigwedge_{ \scriptstyle k \in [1,M] \atop {\scriptstyle \ell \in [1,M+2] \atop \scriptstyle(k,\ell) \neq (i,j)}} \neg \alpha^\ell_k(x)\bigg) \\
                   &&\\
%

  \phi^2_\alpha(x) &:= & \displaystyle \bigwedge_{ \scriptstyle i \in
                         [1,M] \atop \scriptstyle j \in [1,M+2]}
                         \bigg( \alpha^j_i(x) \rightarrow \forall
                         y.\big(( \neg \ediag(y) \wedge \phiint(x,y)) \rightarrow \alpha^j_i(y) \big)\bigg)\\
                &&\\

  \phi^3_\alpha(x) &:= & \displaystyle \bigwedge_{ \scriptstyle i \in [1,M-1] \atop \scriptstyle j \in [1,M+2]} \left( \alpha^j_i(x) \rightarrow
  \left(
  \begin{array}{l}
  \phantom{\wedge\;}\exists y. \big(\phiint(x,y) \wedge \gamma_i(y)\big)\\[0.5ex]
  \wedge\; \neg\exists y.\big(\phiint(x,y) \wedge \gamma_{i+1}(y)\big)
  \end{array}
  \right)
  \right)\;\wedge \\&&\displaystyle  \bigwedge_{ \scriptstyle j \in [1,M+2]}\bigg(\alpha^j_M(x) \rightarrow \exists y. \big(\phiint(x,y) \wedge \gamma_M(y)\big) \bigg)\\ 
                   &&\\

      \phi^4_\alpha(x) &:= & \displaystyle \bigwedge_{ \scriptstyle i \in [1,M] \atop \scriptstyle j \in [1,M+1]} \bigg( \alpha^j_i(x) \rightarrow \forall y.
\bigg(
\left(
\begin{array}{l}
\neg \ediag(y) \wedge \phi_{same}(x,y)\\
\wedge\; \rels{1}{1}{x}{y} \wedge \neg (\rels{2}{2}{x}{y})
  \end{array}
  \right)
\rightarrow 
                           \bigwedge_{k \in [1,M]} \neg\alpha^j_k(y)   \bigg)\bigg)\\
                &&\\
\phi^5_\alpha(x)& := & \displaystyle \bigwedge_{ \scriptstyle i \in [1,M] \atop \scriptstyle j \in [2,M+2]} \bigg( \alpha^j_i(x) \rightarrow  \exists y. \big(\phi_{same}(x,y) \wedge \rels{1}{1}{x}{y} \wedge \bigvee_{k \in [1,M]}\alpha^{j-1}_k(y)\big)\bigg)\\
\end{array}
$$

We then define $\phi_\alpha:=\forall x.\big( (\neg \ediag(x))
\rightarrow (\phi^1_\alpha(x) \wedge \phi^2_\alpha(x) \wedge
\phi^3_\alpha(x) \wedge \phi^4_\alpha(x) \wedge \phi^5_\alpha(x))
\big) \wedge ( \ediag(x) \rightarrow \bigwedge_{\alpha \in
  \Lambda_M^\alpha}\neg \alpha(x))$. Note that $\phi_\alpha$ is a two-variable formula in $\ndFOvar{2}{\extSigma \cup \{\ediag\} \cup \Lambda_M,\Binary_{df}}{2}$. If a data structure satisfies $\phi_\alpha$, then no diagonal element is labelled with predicates in
$\Lambda_M^\alpha$ and all its non-diagonal elements are labelled with
exactly one predicate in $\Lambda_M^\alpha$ (see
$\phi^1_\alpha$). Furthermore, all non-diagonal elements in a same
intersection are labelled with the same $\alpha^j_i$ (see
$\phi^2_\alpha$), and there are exactly $i$ such elements in the
intersection if $i \le M-1$ and at least $M$ otherwise (see
$\phi^3_\alpha$). Finally, we want to identify up to $M+2$ different
intersections sharing the same first value and  we use the $j$ in
$\alpha^j_i$ for this matter. Formula $\phi^4_\alpha$ tells us that no
two non-diagonal elements with the same labels of $\Unary $ share the same index $j$ (for $j \le M+1$) if they do not belong to the same intersection and have the same first value. The formula $\phi^5_\alpha$ specifies that, if an element $a$ is labelled with $\alpha^j_i$, then there are at least $j$ different nonempty intersections with the same labels of $\Unary$ as $a$ sharing the same first values.

The next lemma formalizes the property of this labelling.

\begin{lemma}
  \label{lem:alphalab}
  We consider $\AA\!=\!(A,(P_{\unary}),\ifunct,\ofunct) \in \nData{2}{\Sigma \cup \{\eqx\} \cup \myCountConst{}{}{M} \cup \Lambda_M}$
  eq-respecting
  and such that $\addediag{\AA} \models \phi_\gamma \wedge \phi_\alpha
  $ and $a \in A$. Let $S_{a,\rels{1}{1}{}{}} = \{b \in A \mid \relsaa{1}{1}{\AA}{a}{b} \wedge \labels{a} = \labels{b} \}$ and  $S^{j}_{a,\rels{1}{1}{}{},i}=S_{a,\rels{1}{1}{}{}}\cap P_{\alpha^j_i}$ for all $i \in [1,M]$ and  $j \in [1,M+2]$. The following properties hold:
  \begin{enumerate}\itemsep=0.5ex
  \item We have $S_{a,\rels{1}{1}{}{}}=\bigcup_{i \in [1,M], j \in [1,M+2]} S^j_{a,\rels{1}{1}{}{},i}$.
  \item For all $j,\ell \in [1,M+2]$ and $i,k \in [1,M]$ such that $i \neq k$ or $j \neq l$, we have $S^j_{a,\rels{1}{1}{}{},i} \cap S^\ell_{a,\rels{1}{1}{}{},k}=\emptyset$.
    \item For all $j \in [1,M+1]$ and $i \in [1,M]$ such that $b,c \in S^j_{a,\rels{1}{1}{}{},i}$, we have $\rels{2}{2}{b}{c}$.
    \item For all $b,c \in S_{a,\rels{1}{1}{}{}}$ such that $\rels{2}{2}{b}{c}$, there exist $j \in [1,M+2]$ and $i \in [1,M]$ such that $b,c \in S^j_{a,\rels{1}{1}{}{},i}$.
   \item For all $j \in [1,M+2]$ and $i \in [1,M]$ such that $b \in S^j_{a,\rels{1}{1}{}{},i}$,
   we have
   \[\begin{cases}
   |\{c \in A \mid \relsaa{1}{1}{\AA}{b}{c} \wedge \relsaa{2}{2}{\AA}{b}{c} \wedge \labels{b} = \labels{c} \}|=i & \text{if } i \le M-1\\
   |\{c \in A \mid \relsaa{1}{1}{\AA}{b}{c} \wedge
   \relsaa{2}{2}{\AA}{b}{c} \wedge \labels{b} = \labels{c}  \}|\geq M
   & \text{otherwise.}
   \end{cases}\]
  \item For all $j \in [1,M+1]$, there exists at most one $i$ such that $S^{j,[i]}_{a,\rels{1}{1}{}{}} \neq \emptyset$.
   
  \item For all $j \in [2,M+2]$ and $i \in [1,M]$ such that $S^j_{a,\rels{1}{1}{}{},i}\neq \emptyset$, there exists $k \in [1,M]$ such $S^{j-1}_{a,\rels{1}{1}{}{},k}\neq \emptyset$.
  \end{enumerate}

\end{lemma}

\begin{proof}
  We prove the different statements:
  \begin{enumerate}\itemsep=0.5ex
  \item Thanks to the formula $\phi^1_\alpha(x)$ we have that $A=\bigcup_{i \in [1,M], j \in [1,M+2]}  P_{\alpha^j_i}$. Since $S_{a,\rels{1}{1}{}{}}=A \cap S_{a,\rels{1}{1}{}{}}$, we deduce that
  \[S_{a,\rels{1}{1}{}{}}=\big(\bigcup_{i \in [1,M], j \in [1,M+2]} P_{\alpha^j_i}\big)  \cap S_{a,\rels{1}{1}{}{}} =\bigcup_{i \in [1,M], j \in [1,M+2]} S^j_{a,\rels{1}{1}{}{},i}\,.\]
  \item This point can be directly deduced thanks to $\phi^1_\alpha(x)$.
  \item This point can be directly deduced thanks to $\phi^4_\alpha(x)$.
  \item Since $b \in S_{a,\rels{1}{1}{}{}}$, by 1.\ there exist $j \in [1,M+2]$ and $i \in [1,M]$ such that $b \in S^j_{a,\rels{1}{1}{}{},i}$. Furthermore, since $c \in S_{a,\rels{1}{1}{}{}}$, using formula $\phi^2_\alpha(x)$, we deduce that $c \in S^j_{a,\rels{1}{1}{}{},i}$.
  \item This point can be directly deduced thanks to formula $\phi^3_\alpha(x)$ and to Lemma~\ref{lem:gamma-count}.
  \item Assume there exist $i,i' \in [1,M]$ such that $i \neq i'$ and $S^j_{a,\rels{1}{1}{}{},i} \neq \emptyset$ and $S^j_{a,\rels{1}{1}{}{},i'} \neq \emptyset$. Let $b \in S^j_{a,\rels{1}{1}{}{},i}$ and $c \in S^j_{a,\rels{1}{1}{}{},i'} \neq \emptyset$. If $\relsaa{2}{2}{\AA}{b}{c}$, then, by 5., we necessarily have $i=i'$. Hence we deduce that $\relsaa{2}{2}{\AA}{b}{c}$ does not hold, and we can conclude thanks to formula $\phi^4_\alpha(x)$.
 \item This point can be directly deduced thanks to formula $\phi^5_\alpha(x)$.\qedhere
  \end{enumerate}
\end{proof}

While the predicates $\alpha^j_i$ deal with the relation
$\rels{1}{1}{}{}$, we now define a similar formula $\phi_\beta \in
\ndFOvar{2}{\extSigma \cup \{\ediag\} \cup \Lambda_M,\Binary_{df}}{2}$ for the
predicates in $\Lambda_M^\beta$ to count intersections connected by
the binary relation $\rels{2}{2}{}{}$.

\begin{remark}
The formula $\phi_\beta$ to deal with the predicates in $\Lambda_M^\beta$ is built
the exact same way as the formula  $\phi_\alpha$ but we provide its
definition and its properties in a detailed way for the sake of
completeness. The already convinced reader can continue right after
Lemma \ref{lem:betalab}.
\end{remark}

We introduce hence the
following formulas (we recall that $\phiint = \rels{1}{1}{x}{y} \mathrel{\wedge}
\rels{2}{2}{x}{y} \mathrel{\wedge} \bigwedge_{\sigma \in \Unary}
\sigma(x) \leftrightarrow \sigma(y)$ and $\phi_{same}= \bigwedge_{\sigma \in \Unary}
\sigma(x) \leftrightarrow \sigma(y)$):

$$
\begin{array}{lcl}
  \phi^1_\beta (x)&:= &\displaystyle \bigvee_{\scriptstyle i \in [1,M]\atop \scriptstyle j \in [1,M+1]} \bigg(\beta^j_i(x) \wedge \bigwedge_{ \scriptstyle k \in [1,M] \atop {\scriptstyle \ell \in [1,M+1] \atop \scriptstyle(k,\ell) \neq (i,j)}} \neg \beta^\ell_k(x)\bigg) \\
                &&\\
  \phi^2_\beta(x) &:= & \displaystyle \bigwedge_{ \scriptstyle i \in
                        [1,M] \atop \scriptstyle j \in [1,M+1]} \bigg(
                        \beta^j_i(x) \rightarrow \forall y.\big((\neg
                        \ediag(y) \wedge \phiint(x,y)) \rightarrow \beta^j_i(y) \big)\bigg)\\
                &&\\
  \phi^3_\beta(x) &:= & \displaystyle \bigwedge_{ \scriptstyle i \in [1,M-1] \atop \scriptstyle j \in [1,M+1]} \left( \beta^j_i(x) \rightarrow
  \left(
  \begin{array}{l}
  \phantom{\wedge\;}\exists y. \big(\phiint(x,y) \wedge \gamma_i(y)\big)\\
  \wedge\; \neg\exists y.\big(\phiint(x,y) \wedge \gamma_{i+1}(y)\big)
  \end{array}
  \right)
  \right)\;\wedge \\&&\displaystyle  \bigwedge_{ \scriptstyle j \in [1,M+1]}\bigg(\beta^j_M(x) \rightarrow \exists y. \big(\phiint(x,y) \wedge \gamma_M(y)\big) \bigg)\\ 
               &&\\
   \phi^4_\beta(x) &:= & \displaystyle \bigwedge_{ \scriptstyle i \in [1,M] \atop \scriptstyle j \in [1,M]} \bigg( \beta^j_i(x) \rightarrow \forall y.
   \bigg(
   \left(
\begin{array}{l}
\neg \ediag(y) \wedge 
   \phi_{same}(x,y)\\
   \wedge\; \neg (\rels{1}{1}{x}{y}) \wedge \rels{2}{2}{x}{y}
 \end{array}
 \right) \rightarrow \bigwedge_{k \in [1,M]} \neg\beta^j_k(y) \bigg)\bigg)\\
                &&\\
\phi^5_\beta(x) & := & \displaystyle \bigwedge_{ \scriptstyle i \in [1,M] \atop \scriptstyle j \in [2,M+1]} \bigg( \beta^j_i(x) \rightarrow   \exists y. \bigg(\phi_{same}(x,y) \wedge \rels{2}{2}{x}{y} \wedge \bigvee_{k \in [1,M+1]}\beta^{j-1}_k(y)\bigg)\bigg)\\
\end{array}
$$
We then define $\phi_\beta:=\forall x.\big( (\neg \ediag(x))
\rightarrow (\phi^1_\beta(x) \wedge \phi^2_\beta(x) \wedge
\phi^3_\beta(x)\wedge \phi^4_\beta(x))\big) \wedge ( \ediag(x) \rightarrow \bigwedge_{\beta \in
  \Lambda_M^\alpha}\neg \beta(x))$.\\

The following Lemma is the equivalent of the Lemma \ref{lem:alphalab}
for the relation $\rels{2}{2}{}{}$. Its proof is similar to the one of
the Lemma \ref{lem:alphalab}.

\begin{lemma}
  \label{lem:betalab}
We consider $\AA\!=\!(A,(P_{\unary}),\ifunct,\ofunct) \in \nData{2}{\Sigma \cup \{\eqx\} \cup \myCountConst{}{}{M} \cup \Lambda_M}$
  eq-respecting
  and such that $\addediag{\AA} \models \phi_\gamma \wedge \phi_\beta
  $ and $a \in A$. Let $S_{a,\rels{2}{2}{}{}} = \{b \in A \mid \relsaa{2}{2}{\AA}{a}{b} \wedge \labels{a} = \labels{b} \}$ and  $S^j_{a,\rels{2}{2}{}{},i}=S_{a,\rels{2}{2}{}{}}\cap P_{\beta^j_i}$ for all $i \in [1,M]$ and  $j \in [1,M+1]$. The following statements hold:
  \begin{enumerate}\itemsep=0.5ex
  \item We have $S_{a,\rels{2}{2}{}{}}=\bigcup_{i \in [1,M], j \in [1,M+1]} S^j_{a,\rels{2}{2}{}{,i}}$.
  \item For all $j,\ell \in [1,M+1]$ and $i,k \in [1,M]$ such that $i \neq k$ or $j \neq l$, we have $S^j_{a,\rels{2}{2}{}{},i} \cap S^\ell_{a,\rels{2}{2}{}{},k}=\emptyset$.
    \item For all $j \in [1,M+1]$ and $i \in [1,M]$ such that $b,c \in S^j_{a,\rels{2}{2}{}{},i}$, we have $\rels{1}{1}{b}{c}$.
    \item For all $b,c \in S_{a,\rels{2}{2}{}{}}$ such that $\rels{1}{1}{b}{c}$, there exists $j \in [1,M+1]$ and $i \in [1,M]$ such that $b,c \in S^j_{a,\rels{2}{2}{}{},i}$.
      \item For all $j \in [1,M+1]$ and $i \in [1,M]$ such that $b \in S^j_{a,\rels{2}{2}{}{},i}$, we have 
\[
\begin{cases}
|\{c \in A \mid \relsaa{1}{1}{\AA}{b}{c} \wedge
\relsaa{2}{2}{\AA}{b}{c} \wedge \labels{b} = \labels{c}\}|=i & \text{if } i \le M-1\\
|\{c \in A \mid \relsaa{1}{1}{\AA}{b}{c} \wedge \relsaa{2}{2}{\AA}{b}{c} \wedge \labels{b} = \labels{c}\}|\geq M & \text{otherwise.}
\end{cases}
\]
  \item For all $j \in [1,M]$, there exists at most one $i$ such that $S^j_{a,\rels{2}{2}{}{},i} \neq \emptyset$.
   
  \item For all $j \in [2,M+1]$ and $i \in [1,M]$ such that $S^j_{a,\rels{2}{2}{}{},i}\neq \emptyset$, there exists $k \in [1,M]$ such $S^{j-1}_{a,\rels{2}{2}{}{},k}\neq \emptyset$.
  \end{enumerate}
\end{lemma}

Now that we can count on a consistent labeling with predicates from $\Lambda_M$,
let us see how we can exploit it to express
$\newcc{\USet}{\BSet}{\ge}{m} \in \myCountConst{}{}{M}$,
with additional help from Lemma~\ref{lem:new-env-with-no-diag},
as a formula $\phi_{\USet,\BSet,m}(x) \in \ndFOvar{2}{\extSigma \cup \{\eqx,\ediag\} \cup \Lambda_M,\Binary_{df}}{2}$
applied to \emph{non-diagonal} elements (outside $P_\ediag$).
Let us look at two sample cases according to the case distinction
done in Lemma~\ref{lem:new-env-with-no-diag}.
Hereby, we will use, for $U \subseteq \Sigma$, the formula
$\philabels{\USet}(y) = \bigwedge_{\sigma \in U} \sigma(y) \wedge
\bigwedge_{\sigma \in \Sigma \setminus U} \neg \sigma(y)$. We now
provide the definition of the formulas $\phi_{\USet,\BSet,m}$ using a case
analysis on the shape of $\BSet$ and the result of Lemma \ref{lem:new-env-with-no-diag}:

\begin{enumerate}
\item {\bf Case $\BSet=\{\relsaord{1}{1},\relsaord{2}{2},\relsaord{1}{2}\}$:} in this simple case, we need to say that (i) the element $a$ under consideration is in $P_\eqx$, and (ii) there is an intersection of size at least $m$ (i..e., it contains a $\gammarel{m}$-labeled element) whose elements $b$ satisfy $\rels{1}{1}{a}{b}$, $\rels{2}{2}{a}{b}$, and $\labels{b}=U$:
  $$
  \begin{array}{lcl}
    \phi_{\USet,\BSet,m}(x)&:= &\displaystyle  \eqx(x) \wedge \exists y.(\phi_\USet(y) \wedge \rels{1}{1}{x}{y} \wedge \rels{2}{2}{x}{y} \wedge \gamma_m(y))
  \end{array}
  $$
\item {\bf Case $\BSet=\{\relsaord{1}{1},\relsaord{2}{2}\}$:}
  $$
  \begin{array}{lcl}
    \phi_{\USet,\BSet,m} (x)&:=&\displaystyle \neg \eqx(x) \wedge \exists y.(\phi_\USet(y) \wedge \rels{1}{1}{x}{y} \wedge \rels{2}{2}{x}{y} \wedge \gamma_m(y))
  \end{array}
  $$
\item  {\bf Case  $\BSet=\{\relsaord{1}{1},\relsaord{1}{2}\}$:}
  $$
  \begin{array}{lcl}
    \phi_{\USet,\BSet,m}(x)&:=&\displaystyle \neg \eqx(x) \wedge \exists y.(\phi_\USet(y) \wedge \eqx(y) \wedge \rels{1}{1}{x}{y} \wedge \gamma_m(y))\\
    
  \end{array}
  $$
\item  {\bf Case $\BSet=\{\relsaord{2}{2},\relsaord{1}{2}\}$:} For this case, we first
  need an extra definition. For $m \in [1,M]$, we define
  $\mathcal{S}_{\beta,m}$ the set of subsets of $\Lambda_M^\alpha$ as
  follows:
  $\mathcal{S}_{\beta,m}=\{\{\beta^{j_1}_{i_1},\ldots,\beta^{j_k}_{i_k}\}
  \mid i_1+\ldots+i_k\geq m \mbox{ and } j_1 < j_2 < \ldots < j_k
  \}$. It corresponds to the sets of element $\beta^j_i$ whose sum of
  $i$ is greater than or equal to $m$.
  We then have:
  $$
  \begin{array}{lcl}
    \phi_{\USet,\BSet,m} (x)&:=&   \displaystyle \eqx(x) \wedge \bigvee_{S \in\mathcal{S}_{\beta,m}}  \bigwedge_{\beta \in S} \exists y.\big(\phi_\USet(y) \wedge \neg \eqx(y) \wedge \beta(y) \wedge \rels{2}{2}{x}{y}\big)
  \end{array}
  $$
\item  {\bf Case $\BSet=\{\relsaord{2}{2}\}$:} we use again the set $\mathcal{S}_{\beta,m}$ introduced previously.
  $$
  \begin{array}{lcl}
    \phi_{\USet,\BSet,m} (x)&:=& \displaystyle \neg \eqx(x) \wedge \bigvee_{S \in\mathcal{S}_{\beta,m}}  \bigwedge_{\beta \in S} \exists y.\big(\phi_\USet(y) \wedge \beta(y) \wedge \neg(\rels{1}{1}{x}{y}) \wedge \rels{2}{2}{x}{y}\big)    
  \end{array}
  $$
\item {\bf Case $\BSet=\{\relsaord{1}{1}\}$:} Similar to Case 4., we first need
  an extra definition.
For $m \in \{1,\ldots,M\}$, we define the set $\mathcal{S}_{\alpha,m}$ of subsets
of $\Lambda_M^\alpha$ as follows: $\mathcal{S}_{\alpha,m}=\{\{\alpharel{i_1}{j_1},\ldots,\alpharel{i_k}{j_k}\} \mid i_1+\ldots+i_k \geq m  \mbox{ and } j_1 < j_2 < \ldots < j_k \}$. It corresponds to the sets of elements $\alpharel{i}{j}$ whose sum of $i$ is greater than or equal to $m$.
  We then have:
\[\phi_{\USet,\BSet,m}(x) := \displaystyle \bigvee_{S \in\mathcal{S}_{\alpha,m}}  \bigwedge_{\alpha \in S} \exists y.\big(\phi_\USet(y) \wedge \alpha(y) \wedge \neg \eqx(y) \wedge \rels{1}{1}{x}{y} \wedge \neg (\rels{2}{2}{x}{y})\big)\]
\item {\bf Case $\BSet=\{\relsaord{1}{2}\}$:} We use here again the set $\mathcal{S}_{\beta,m}$ introduced in Case 4.
    
    
  $$
  \begin{array}{lcl}
    \phi_{\USet,\BSet,m}(x) &:=& \displaystyle \neg\eqx(x) \wedge \exists y.\bigg(\ediag(y) \wedge \rels{1}{1}{x}{y} \;\wedge\\
                           &&   \displaystyle            \bigvee_{S \in\mathcal{S}_{\beta,m}}  \bigwedge_{\unary \in S} \exists x.\big(\phi_\USet(x) \wedge \unary(x) \wedge \neg (\rels{1}{1}{y}{x}) \wedge \rels{2}{2}{y}{x} \big)\bigg)\\
    
  \end{array}
  $$
\end{enumerate}

Finally, it remains to say that all elements are labeled with the suitable counting constraints. So we let
$\phi_{cc} = \forall x. \neg\ediag(x) \to \bigwedge_{\newcc{\USet}{\BSet}{\ge}{m}\, \in\, \myCountConst{}{}{M}} \newcc{\USet}{\BSet}{\ge}{m}(x) \leftrightarrow \phi_{\USet,\BSet,m}(x)$.

\begin{lemma}
\label{lem:new-correct-typing}
  Let $\AA=(A,(P_{\unary}),\ifunct,\ofunct) \in \nData{2}{\Sigma \cup \{\eqx\} \cup \myCountConst{}{}{M} \cup \Lambda_M}$
  be eq-respecting.
  If $\addediag{\AA} \models \phi_\alpha \wedge \phi_\beta \wedge \phi_\gamma \wedge \phi_{cc}$, then $\AA$ is cc-respecting.
\end{lemma}

\begin{proof}
  Let $\AA=(A,(P_{\unary}),\ifunct,\ofunct) \in \nData{2}{\Sigma \cup \{\eqx\} \cup \myCountConst{}{}{M} \cup \Lambda_M}$
  be eq-respecting and such that $\addediag{\AA} \models \phi_\alpha
  \wedge \phi_\beta \wedge \phi_\gamma \wedge \phi_{cc}$. We need to
  show that for all $a \in A$ and all $\newcc{U}{R}{\ge}{m} \in \myCountConst{}{}{M}$ ,
we have $a \in P_{\newcc{U}{R}{\ge}{m}}$ iff
$|\myEnv{\AA}{\Sigma}{\Binary}{a}{\USet}{\BSet}| \ge m$. We consider
$a \in A$. Since $\addediag{\AA}
\models \phi_{cc}$, we deduce that $a \in P_{\newcc{U}{R}{\ge}{m}}$
iff $\addediag{\AA} \models_{\Intrepl{x}{a}} \phi_{\USet,\BSet,m}(x)$. We need
hence to show that $\addediag{\AA} \models_{\Intrepl{x}{a}}
\phi_{\USet,\BSet,m}(x)$ iff $|\myEnv{\AA}{\Sigma}{\Binary}{a}{\USet}{\BSet}| \ge m$. To prove
this , we first use Lemma \ref{lem:new-env-with-no-diag} to get a
characterization of
$\myEnv{\AA}{\Sigma}{\Binary}{a}{\USet}{\BSet}$.  This
characterization is then directly translated into the formula
$\phi_{\USet,\BSet,m}(x)$ which makes use of the label in $\Lambda_M$
to count in the environment of $a$. The fact that this counting is
performed correctly is guaranteed by Lemmas
\ref{lem:gamma-count},\ref{lem:alphalab} and
\ref{lem:betalab}. Putting these arguments together, we can conclude
that the lemma holds.
\end{proof}

\newcommand{\All}{\mathsf{All}}
\newcommand{\transl}[1]{\widehat{#1}}

\noindent\textbf{Step 5: Putting it All Together}\\

Let $\All = \Sigma \cup \{\eqx,\ediag\} \cup \myCountConst{}{}{M} \cup \Lambda_M$
denote the set of all the unary predicates that we have introduced so far.
Recall that, after Step~1, we were left with $M \ge 1$ and
a formula $\phi \in \ndFO{2}{\Unary \cup \{\eqx\} \cup \myCountConst{\Unary}{\Binary}{M},\emptyset}$.
The question is now whether $\phi$ has a well-typed model (i.e., a model that is eq-respecting
and cc-respecting). Altogether, we get the following reduction:

\begin{proposition}
  \label{prop:count-sat}
Let $\phi \in \ndFO{2}{\Unary \cup \{\eqx\} \cup \myCountConst{\Unary}{\Binary}{M},\emptyset}$.
Then, $\phi$ has a well-typed model iff $\transl{\phi} := \tred{\phi} \wedge \phied^{\All \setminus \{\eqx,\ediag\}} \wedge \phi_\alpha \wedge \phi_\beta \wedge \phi_\gamma \wedge \phi_{cc} \in \extndFOvar{2}{\All,\Binary_{df}}{2}$ is satisfiable.
\end{proposition}

\begin{proof}
  Suppose $\transl{\phi}$ is satisfiable. Then, there is $\BB \in \nData{2}{\All}$ such that $\BB \models \transl{\phi}$. By Lemma \ref{lem-diagonalization}, there exists an eq-respecting data structure $\AA \in \nData{2}{\Unary \cup \{\eqx\} \cup \myCountConst{\Unary}{\Binary}{M} \cup \Lambda_M}$ such that $\addediag{\AA} \models \tred{\phi} \wedge \phi_\alpha \wedge \phi_\beta \wedge \phi_\gamma \wedge \phi_{cc}$. Using Lemma \ref{lem:new-correct-typing}, we deduce that $\AA$ is cc-respecting and, thus, well-typed. Furthermore, by Lemma \ref{lem:new-Tdiag}, we have $\AA \models \phi$. Note that $\AA$ belongs to $\nData{2}{\Unary \cup \{\eqx\} \cup \myCountConst{\Unary}{\Binary}{M} \cup \Lambda_M}$. However, by removing the unary predicates in $\Lambda_M$, we still have a model of $\phi$ from $\nData{2}{\Unary \cup \{\eqx\} \cup \myCountConst{\Unary}{\Binary}{M}}$ as required. Hence, $\phi$ has a well-typed model.

Assume now that there exists a well-typed data structure $\AA \in \nData{2}{\Unary \cup \{\eqx\} \cup \myCountConst{\Unary}{\Binary}{M}}$ such that $\AA \models \phi$. Using Lemma \ref{lem:new-Tdiag}, we have that $\addediag{\AA} \models \tred{\phi}$. Furthermore, using the fact that $\AA$ is well-typed, we can add the unary predicates from $\Lambda_M$ to $\addediag{\AA}$ to obtain a data structure $\AA'$ in $\nData{2}{\All}$ such that $\AA' \models \phi_\alpha \wedge \phi_\beta \wedge \phi_\gamma \wedge \phi_{cc}$.
Note that $\AA'$ is well-diagonalized.
We deduce that $\AA' \models \transl{\phi}$.
\end{proof}

We are now in the position to prove our main result:

\begin{theorem}\label{thm:datasat-rad1}
   $\DataSat{\rndFOr{1},2,\{\relsaord{1}{1},\relsaord{2}{2},\relsaord{1}{2}\}}$ is decidable.
\end{theorem}

\begin{proof}
Let $\psi \in \rndFO{2}{\Unary,\{\relsaord{1}{1},\relsaord{2}{2},\relsaord{1}{2}\}}{1}$. Using Lemma \ref{lem:new-fo-count}, we can effectively compute $M \in \N$ and $\phi \in \ndFO{2}{\Unary \cup \{\eqx\} \cup \myCountConst{\Unary}{\Binary}{M},\emptyset}$ such that $\psi$ is satisfiable iff $\phi$ has a well-typed model. By Proposition~\ref{prop:count-sat}, $\phi$ has a well-typed model iff $\transl{\phi}$ is satisfiable. Since $\transl{\phi}$ belongs to $\extndFOvar{2}{\All,\Binary_{df}}{2}$, we conclude using Proposition~\ref{prop:exttwoFO}.
\end{proof}

\subsection{Undecidability  results}
We shall now show that extending the neighborhood radius yields undecidability.
We rely on a reduction from the
domino problem  \cite{Gurevich97} and use a specific technique
presented in \cite{otto98c}.

\subsubsection{The Tiling Problem}

A \emph{domino system} $\DD$ is a triple $(D,H,V)$ where $D$ is a finite set of dominoes and $H,V\subseteq D\times D$ are two binary relations.
Let $\grid{m}$ denote the standard grid on an $m\times m$ torus, i.e.,
$\grid{m} = (\gridcarrier{m},H_{m},V_{m})$ where $H_{m}$ and
$V_{m}$  are two binary relations defined as follows: $
	\gridcarrier{m} = \Zmod{m}\times\Zmod{m}$, $H_{m} = \{
    ((i,j),(i',j)) \mid i'-i \equiv 1 \mod m \}$, and $
	V_{m} = \{ ((i,j),(i,j')) \mid i'-i \equiv 1 \mod m \}$.
In the sequel, we will suppose $\Zmod{m}=\{0,\ldots,m-1\}$ using
  the least positive member to represent residue classes.
  
A \emph{bi-binary structure} is a triple $(A,R_1,R_2)$ where $A$ is a finite set and $R_1,R_2$ are subsets of $A\times A$.
Domino systems and $\grid{m}$ for any $m$ are examples of bi-binary structures.
For two bi-binary structures $\grid{}=(G,H,V)$ and $\grid{}'=(G',H',V')$, we say that 
$\grid{}$ is \emph{homomorphically embeddable} into $\grid{}'$ if there is a morphism $\pi:\grid{}\to\grid{}'$, i.e., a mapping $\pi$ such that, for all $a,a'\in G, (a,a')\in H \donc (\pi(a),\pi(a'))\in H'$ and $(a,a')\in V \donc (\pi(a),\pi(a'))\in V'$.
For instance, $\grid{k\cdot m}$ is homomorphically embeddable into $\grid{m}$ through reduction mod $m$.
For a domino system $\DD$, a \emph{periodic tiling} is a morphism $\tau:\grid{m}\to\DD$ for some $m$ and we say that $\DD$ \emph{admits a periodic tiling} if there exists a periodic tiling of $\DD$.

The problem $\Tile$ (or \emph{periodic tiling problem}), which
is well known to be undecidable \cite{Gurevich97}, is defined as follows:
Given a domino system $\DD$, does $\DD$ admit a periodic tiling?

To use $\Tile$ in our reductions, we first use some
specific bi-binary structures, which we call grid-like and which are 
easier to manipulate in our context to encode domino systems.
A bi-binary structure $\grid{}=(A,H,V)$ is said to be $\emph{grid-like}$ if some $\grid{m}$ is homomorphically embeddable into $\grid{}$.
The logic \emph{$\FO$ over bi-binary structures} refers to the first-order logic on two binary relations $\Hform,\Vform$, and we write $\Hform xy$ to say that $x$ and $y$ are in relation for $\Hform$.
Consider the two following $\FO$ formulas over bi-binary structures: $
\vpcomplete = \forall x.\forall y.\forall x'.\forall y'.((\Hform xy\et
\Vform xx' \et \Vform yy') \to \Hform x'y')$ and $
\vpprog = \forall x.(\exists y.\Hform xy\et\exists y. \linebreak[0]\Vform xy)$.
The following lemma, first stated and proved in \cite{otto98c}, shows
that these formulas suffice to characterize grid-like
structures:

\begin{lemC}[\cite{otto98c}] \label{lemma:grid-like}
	Let $\grid{}=(A,H,V)$ be a bi-binary structure. If $\grid{}$ satisfies $\vpcomplete$ and $\vpprog$, then $\grid{}$ is grid-like.
  \end{lemC}

  Given $\AA=(A,(P_{\unary}),\ifunct,\ofunct,) \in \nData{2}{\Sign}$ and
$\vp(x,y) \in \ndFO{2}{\Unary,\Binary}$, we define the binary relation
$\query{\vp}{\AA} = \{(a,b)\in A\times A\mid
\AA\models_{\Intrepltwo{x}{a}{y}{b}} \vp(x,y)$ for some interpretation function $I\}$.
Thus, given two  $\ndFO{2}{\Unary,\Binary}$ formulas
$\vp_1(x,y),\vp_2(x,y)$ with two free variables,
$(A,\query{\vp_1}{\AA},\query{\vp_2}{\AA})$ is a bi-binary structure.

\begin{figure}[htbp]

	\begin{tikzpicture}[line cap=round,line join=round,>=triangle 45,x=1cm,y=1cm]
\begin{scope}
\clip (-.5,-.5) rectangle +( \xmax+1,\ymax+1);

\foreach \x in {1,3,...,\xmax}
\foreach \y in {1,3,...,\ymax}
\connectiononetwo{\x}{\y};

\foreach \x in {0,  2, ..., \xmax}
\foreach \y in {0,  2, ..., \xmax}
\roundedsquare{\x}{\y}{styleclasseone};

\foreach \x in {-1, 1, ..., \xmax}
\foreach \y in {-1, 1, ..., \xmax}
\roundedsquare{\x}{\y}{styleclassetwo};

\foreach \x in {0, ..., \xmax}
\foreach \y in {0, ..., \ymax}
\drawelement{\x}{\y};

\node  at (1,2) [xshift=-0.7em]{$a_1$};
\node  at (2,2) [xshift=-0.7em]{$a_2$};
\node  at (3,2) [xshift=-0.7em]{$a_3$};
\node  at (1,3) [xshift=-0.7em]{$a_4$};
\node  at (2,3) [xshift=-0.7em]{$a_5$};
\node  at (3,3) [xshift=-0.7em]{$a_6$};
\end{scope}

\draw[<->] (-.7,-1) -- (\xmax+.7,-1);
\draw[<->] (-1,-.7) -- (-1,\xmax+.7);
\foreach \i in {0,2,...,\xmax}{
    \draw  (\i,-.95) -- (\i,-1.05) node[below] {$\UH{0}$};
    \draw (-.95,\i) -- (-1.05,\i) node[left] {$\UV{0}$};
}
\foreach \i in {1,3,...,\xmax}{
    \draw (\i,-.95) -- (\i,-1.05) node[below] {$\UH{1}$};
    \draw (-.95,\i) -- (-1.05,\i) node[left] {$\UV{1}$};
}
\end{tikzpicture}
	\caption{The local pattern of $\AA_{2m}$.
	Dots denote elements. Two dots are in the same $\s{1}{1}$-equivalence class (resp.\ $\s{2}{2}$) iff they are in the same green (resp.\ purple) area.
	The thick black lines represent the relation $\s{1}{2}$ in the following way: 
	if a $\s{1}{1}$-equivalence class $C_1$ and a $\s{2}{2}$-equivalence class $C_2$ are connected with a thick black line, then for any $a\in C_1$ and $b\in C_2$, we have $\rels{1}{2}{a}{b}$.}
	\label{fig:extended grid}
  \end{figure}

As we want to reason on data structures, we build a data structure
$\AA_{2m}$ that corresponds to the grid $\grid{2m}=(\gridcarrier{2m},H_{2m},V_{2m})$. This
structure is depicted locally in Figure~\ref{fig:extended grid}.
To define $\AA_{2m}$, we use four unary predicates given by
$\Unarygrid = \{\UH{0},\UH{1},\UV{0},\UV{1}\}$. They
give us access to the coordinate modulo $2$.
We then define $\AA_{2m}=(\gridcarrier{2m},\ifunct,\ofunct,
(P_\unary)) \in \nData{2}{\Unarygrid}$ as follows:
For $k \in \{0,1\}$, we have $P_{\UH{k}} = \{ (i,j) \in
  \gridcarrier{2m} \mid i \equiv k \mod 2\}$ and 
	$P_{\UV{k}} = \{ (i,j) \in \gridcarrier{2m} \mid j \equiv k \mod 2\}$.
For all $i,j \in \{0,\ldots,2m-1\}$, we set
  $\ifunct(i,j)=((i/2)\mod m)+m*((j/2)\mod m)$ (where $/$ stands
  for the Euclidian division). Finally,
  for all $i,j \in \{1,\ldots,2m\}$, set $\ofunct(i \mod
  (2m),j \mod (2m))=\ifunct(i-1,j-1)$.

We provide below the definition of quantifier free formulas
$\vpH(x,y)$ and $\vpV(x,y)$ from the logic $\ndFO{2}{\Unarygrid,\{\relsaord{1}{1},\relsaord{2}{2}\}}$
with two free variables.
These formulas allow us to make the link between the data structure
$\AA_{2m}$ and the grid $\grid{2m}$, and we will use them later on to
ensure that a data structure has a shape 'similar' to
$\AA_{2m}$.

\begin{align*}
	\vpHzz &= \Pform{\UH{0}}{x} \et \Pform{\UH{1}}{y} \et  \Pform{\UV{0}}{x} \et \Pform{\UV{0}}{y} \et \rels{1}{1}{x}{y} \\
	\vpHoz &= \Pform{\UH{1}}{x} \et \Pform{\UH{0}}{y} \et  \Pform{\UV{0}}{x} \et \Pform{\UV{0}}{y} \et \rels{2}{2}{x}{y} \\
	\vpHzo &= \Pform{\UH{0}}{x} \et \Pform{\UH{1}}{y} \et  \Pform{\UV{1}}{x} \et \Pform{\UV{1}}{y} \et \rels{1}{1}{x}{y} \\
	\vpHoo &= \Pform{\UH{1}}{x} \et \Pform{\UH{0}}{y} \et  \Pform{\UV{1}}{x} \et \Pform{\UV{1}}{y} \et \rels{2}{2}{x}{y} \\
  \vpH &= \vpHzz \ou \vpHoz \ou \vpHzo \ou \vpHoo\\
  	\vpVzz &= \Pform{\UH{0}}{x} \et \Pform{\UH{0}}{y} \et  \Pform{\UV{0}}{x} \et \Pform{\UV{1}}{y} \et \rels{1}{1}{x}{y} \\
	\vpVoz &= \Pform{\UH{1}}{x} \et \Pform{\UH{1}}{y} \et  \Pform{\UV{0}}{x} \et \Pform{\UV{1}}{y} \et \rels{1}{1}{x}{y} \\
	\vpVzo &= \Pform{\UH{0}}{x} \et \Pform{\UH{0}}{y} \et  \Pform{\UV{1}}{x} \et \Pform{\UV{0}}{y} \et \rels{2}{2}{x}{y} \\
	\vpVoo &= \Pform{\UH{1}}{x} \et \Pform{\UH{1}}{y} \et  \Pform{\UV{1}}{x} \et \Pform{\UV{0}}{y} \et \rels{2}{2}{x}{y} \\
	\vpV &= \vpVzz \ou \vpVoz \ou \vpVzo \ou \vpVoo
\end{align*}

\begin{example}
We  draw six elements $a_1,a_2,a_3,a_4,a_5,a_6$ on Figure
\ref{fig:extended grid}. According to the definition of  $\AA_{2m}$
the data values associated to these elements are::
\begin{itemize}
\item $\ifunct(a_1)=m$ and  $\ofunct(a_1)=0$;
\item $\ifunct(a_2)=m+1$ and  $\ofunct(a_2)=0$;
\item $\ifunct(a_3)=m+1$ and  $\ofunct(a_3)=1$;
\item $\ifunct(a_4)=m$ and  $\ofunct(a_4)=m$;
\item $\ifunct(a_5)=m+1$ and  $\ofunct(a_5)=m$;
\item $\ifunct(a_6)=m+1$ and  $\ofunct(a_2)=m+1$.
\end{itemize}
We have $\ifunct(a_1)=\ifunct(a_4)$ as $a_1$ and
$a_4$ are in the same green area and $\ofunct(a_4)=\ofunct(a_5)$ as
$a_4$ and $a_5$ are in the same purple area. Furthermore we have as well
$\ifunct(a_4)=\ofunct(a_4)$ and $a_4$ lies in a green area connected to
a purple area  by a thick black line. We have then the following
assertions: $\AA_{2m} \models \vpHoz(a_1,a_2)$,  $\AA_{2m} \models
\vpHzz(a_2,a_3)$, $\AA_{2m} \models \vpHoo(a_4,a_5)$, $\AA_{2m}
\models \vpHzo(a_5,a_6)$, $\AA_{2m}
\models \vpVoz(a_1,a_4)$ , $\AA_{2m}
\models \vpVzz(a_2,a_5)$ and $\AA_{2m}
\models \vpVoz(a_3,a_6)$.

\end{example}

Using the definitions of $\gridcarrier{2m}$ and of $\AA_{2m}$ we can
show the following lemma.

\begin{lemma}
 \label{lem:grid2m}
 If $\grid{}$ is the bi-binary structure
  $(\gridcarrier{2m}, \query{\vpH}{\AA_{2m}},
  \query{\vpV}{\AA_{2m}})$, then $\grid{2m} =\grid{}$.
\end{lemma}

\begin{proof}
  We have hence to prove that $H_{2m}=\{(a,b)\mid \AA_{2m} \models_{\Intrepltwo{x}{a}{y}{b}} 
  \vpH(x,y)$ for some
  interpretation function $I\} $ and $V_{2m}=\{(a,b)\mid \AA_{2m} \models_{\Intrepltwo{x}{a}{y}{b}} 
  \vpH(x,y)$ for some
  interpretation function $I\}$.
	We first show that $H_{2m} \subseteq
    \query{\vpH}{\AA_{2m}}$.  Let $((i,j),\linebreak[0](i',j')) \in
    H_{2m}$. Hence we have $j=j'$ and $ i'-i \equiv 1 \mod 2m$. We have
    then different cases according to the parity of $j, i $ and $i'$. Assume
    $i,j$ are even. Then$(i,j),(i',j') \in P_{\UV{0}}$ and $(i,j) \in
    P_{\UH{0}}$ and  $(i',j) \in P_{\UH{1}}$ and by definition of
    $\ifunct$, we have $\ifunct(i,j)=\ifunct(i',j)$, hence 
    $((i,j),(i',j)) \in \query{\vpHzz}{\AA_{2m}}$ and $((i,j),(i',j))
    \in \query{\vpH}{\AA_{2m}}$. The other cases can be treated similarly.

We now prove that $H_{2m} \supseteq \query{\vpH}{\AA_{2m}}$. Let $(a,b)$ be such that $\AA_{2m} \models_{\Intrepltwo{x}{a}{y}{b}} 
  \vpH(x,y)$  for some interpretation  function  $I$.
	For $\vpH$ to hold on $(a,b)$, one of the $\vpHij$ must hold.
	We treat the case $\AA_{2m} \models_{\Intrepltwo{x}{a}{y}{b}} \vpHoo(a,b)$.
	Write $(a_1,a_2)$ and $(b_1,b_2)$ the coordinates of $a$ and $b$ respectively.
	As $a \in P_{\UH{0}} \cap P_{\UV{0}}$  and $b \in P_{\UH{1}} \cap  P_{\UV{0}} $), we have that
    $a_1,a_2,b_2$ are even and $b_1$ is odd.
	As $\rels{1}{1}{a}{b}$, we have  $((a_1/2)\mod m)+m*((a_2/2)\mod
    m)= ((b_1/2)\mod m)+m*((b_2/2)\mod m)$. This allows us to conclude
    that $a_2=b_2$ and that $a_1-b_2 \equiv 1 \mod m$. So we have $(a,b)\in H_{2m}$. The other cases can be treated in a
    similar way.
    
	The proof that $V_{2m} =
    \query{\vpV}{\AA_{2m}}$ follows the exact same lines.
  \end{proof}

  \subsubsection{The Radius 3 case}

We first use the previously introduced notions to show that
$\DataSat{\rndFOr{3}, 2,\{\relsaord{1}{1},\relsaord{2}{2}\}}$ is undecidable. Hence, we assume
that  $\Binary=\{\relsaord{1}{1},\relsaord{2}{2}\}$. The first step in our
reduction from \Tile~
consists in defining the formula $\vpgriddd \in \rndFO{2}{\Unarygrid,\{\relsaord{1}{1},\linebreak[0]\relsaord{2}{2}\}}{3}$
to check that a data structure corresponds to a grid
($\oplus$ stands for exclusive or):
$$
  \begin{array}{lcl}
	\vpcompleteee & =& \forall x. \locformr{ \forall y.\forall
                       x'.\forall y'.\vpH(x,y)\et \vpV(x,x') \et \vpV(y,y')\to \vpH(x',y')}{x}{3} \\[0.8ex]
	\vpproggg &=& \forall x.\locformr{\exists y.\vpH(x,y)\et\exists y. \vpV(x,y)}{x}{3} \\[0.8ex]
	\vpgriddd& =& \vpcompleteee \et  \vpproggg 
	\et \forall x.\locformr{ (\Pform{\UH{0}}{x} \oplus \Pform{\UH{1}}{x}) \et  (\Pform{\UV{0}}{x} \oplus \Pform{\UV{1}}{x}) }{x}{3}
\end{array}
  $$
  
\begin{lemma}
\label{lemma:grid}
We have $\AA_{2m}\models\vpgriddd$.
Moreover, for all  $\AA=(A,(P_{\unary}),\ifunct,\ofunct)$ in $
     \nData{2}{\Unarygrid}$, if $\AA \models \vpgriddd$, then $(A,\query{\vpH}{\AA},\query{\vpV}{\AA})$ is grid-like.
   \end{lemma}

\begin{proof}
  We first show that $\AA_{2m}\models\vpgriddd$. In the proof, we assume that $m\geq 3$. The cases $m=1$ or $2$ are treated in the same way.
	Let us prove the first conjunct, that is $\AA_{2m}\models\vpcompleteee$.
	Let $a\in\gridcarrier{2m}$. We want to prove that
    $$\vprojr{\AA_{2m}}{a}{3} \models_{\Intrepl{x}{a}}  \forall
    y.\forall x'.\forall y'.\vpH(x,y)\et \vpV(x,x') \et
    \vpV(y,y')\donc \vpH(x',y')$$ for some interpretation function
    $I$. We fix an interpretation function $I$.
	We proceed by a case analysis on the values of  $i,j\in\{0,1\}$
     such that $a \in P_{\UH{i}} \cap P_{\UV{j}}$. Assume that $(i,j)=(0,0)$.
	Then $\vprojr{\AA_{2m}}{a}{3}$ is depicted in Figure~\ref{grid_3}(A).
	Let $b,a',b'$ such that
	$$\vprojr{\AA_{2m}}{a}{3} \models_{I[x/a][y/b][x'/a'][y'/b']} \vpH(x,y)\et \vpV(x,x') \et \vpV(y,y')\,.$$ We want to show
	$$\vprojr{\AA_{2m}}{a}{3} \models_{I[x/a][y/b][x'/a'][y'/b']} \vpH(x',y')\,.$$
	By assumption on $a$ and by looking at the definition of $\vpH$,
    $$\vprojr{\AA_{2m}}{a}{3}
    \models_{I[x/a][y/b][x'/a'][y'/b']} \Pform{\UH{1}}{y}\et \Pform{\UV{0}}{y}\et\rels{1}{1}{x}{y}\,.$$
	So by elimination we have that $b$ is the element pointed by Figure~\ref{grid_3_00}.
	In a similar way, $a'$ and $b'$ are indeed the elements pointed by Figure \ref{grid_3}(B).
	Hence, we deduce $$\vprojr{\AA_{2m}}{a}{3} \models_{I[x/a][y/b][x'/a'][y'/b']} \vpH(x',y')\,.$$
	
	The case $(i,j)=(1,0)$ is depicted in Figure \ref{grid_3_10} and
    is proven in the same way just as the cases when $(i,j)= (1,0)$ or $(i,j)=(1,1)$.
	\begin{figure}
		\centering
		\begin{subfigure}{0.45\textwidth}
			\begin{tikzpicture}[line cap=round,line join=round,>=triangle 45,x=1cm,y=1cm]
	\roundedsquare{0}{0}{styleclasseone};
	\roundedsquare{2}{0}{styleclasseone};
	\roundedsquare{0}{2}{styleclasseone};
	\roundedsquare{2}{2}{styleclasseone};
	\roundedsquare{1}{1}{styleclassetwo};
	\roundedsquare{1}{3}{styleclassetwo};
	\roundedsquare{3}{1}{styleclassetwo};
	\roundedsquare{3}{3}{styleclassetwo};
	
	\foreach \i in {1,2,3}{
		\classesingleton{\i}{0}{styleclassetwo};
		\classesingleton{0}{\i}{styleclassetwo};
		\classesingleton{\i}{4}{styleclasseone};
		\classesingleton{4}{\i}{styleclasseone};
	}
	\classesingleton{0}{0}{styleclassetwo};
	\classesingleton{4}{4}{styleclasseone};
	
	\foreach \x in {0,1,2,3}
		\drawelement{\x}{0};
	\foreach \x in {0,1,2,3,4}
		\foreach \y in {1,2,3}
			\drawelement{\x}{\y};
	\foreach \x in {1,2,3,4}
		\drawelement{\x}{4};
	
	\draw (2,2) node[below left] {$a$};
	\draw (3,2) node[below right] {$b$};
	\draw (2,3) node[above left] {$a'$};
	\draw (3,3) node[above right] {$b'$};
	
	\draw[style={[-]}] (-.7,-1) -- (4.7,-1);
	\draw[style={[-]}] (-1,-.7) -- (-1,4.7);
	\foreach \i in {0,2,4}{
	    \draw  (\i,-.95) -- (\i,-1.05) node[below] {$\UH{0}$};
	\draw (-.95,\i) -- (-1.05,\i) node[left] {$\UV{0}$};}
	\foreach \i in {1,3}{
	    \draw (\i,-.95) -- (\i,-1.05) node[below] {$\UH{1}$};
	\draw (-.95,\i) -- (-1.05,\i) node[left] {$\UV{1}$};}
\end{tikzpicture}
			\caption{$a \in P_{\UH{0}} \cap P_{\UV{0}}$}
			\label{grid_3_00}
		\end{subfigure}
		\quad\quad\quad
		\begin{subfigure}{0.45\textwidth}
			\begin{tikzpicture}
	\roundedsquare{2}{0}{styleclasseone}
	\roundedsquare{4}{0}{styleclasseone}
	\roundedsquare{2}{2}{styleclasseone}
	\roundedsquare{4}{2}{styleclasseone}
	\roundedsquare{1}{1}{styleclassetwo}
	\roundedsquare{1}{3}{styleclassetwo}
	\roundedsquare{3}{1}{styleclassetwo}
	\roundedsquare{3}{3}{styleclassetwo}
	
	\foreach \i in {1,2,3}{
	\classesingleton{\i+1}{0}{styleclassetwo}
	\classesingleton{5}{\i}{styleclassetwo}
	\classesingleton{\i+1}{4}{styleclasseone}
	\classesingleton{1}{\i}{styleclasseone}}
	\classesingleton{5}{0}{styleclassetwo}
	\classesingleton{1}{4}{styleclasseone}
	
	\foreach \x in {2,3,4,5},5
		\drawelement{\x}{0};
	\foreach \x in {1,2,3,4,5}
		\foreach \y in {1,2,3}
			\drawelement{\x}{\y};
	\foreach \x in {1,2,3,4}
		\drawelement{\x}{4};
	
	\draw (3,2) node[below right] {$a$};
	\draw (4,2) node[below left] {$b$};
	\draw (3,3) node[above right] {$a'$};
	\draw (4,3) node[above left] {$b'$};
	
	\begin{scope}[xshift=1cm]
		\draw[style={[-]}] (-.7,-1) -- (4.7,-1);
		\draw[style={[-]}] (-1,-.7) -- (-1,4.7);
		\foreach \i in {0,2,4}{
		    \draw  (\i,-.95) -- (\i,-1.05) node[below] {$\UH{1}$};
		    \draw (-.95,\i) -- (-1.05,\i) node[left] {$\UV{0}$};}
		\foreach \i in {1,3}{
		    \draw (\i,-.95) -- (\i,-1.05) node[below] {$\UH{0}$};
		    \draw (-.95,\i) -- (-1.05,\i) node[left] {$\UV{1}$};}
	\end{scope}
\end{tikzpicture}
			\caption{$a \in P_{\UH{1}} \cap P_{\UV{0}}$}
			\label{grid_3_10}
		\end{subfigure}
		\caption{Some 3-local views of $\AA_{2m}$ for
          $\Binary=\{(1,1),(2,2)\}$.}
        \label{grid_3}
	\end{figure}
	
	Showing that $\AA_{2m}\models\vpproggg$ is done in the same way as showing that $\AA_{2m}\models\vpcompleteee$.
	
	Finally, it is obvious that   $\AA_{2m}$ satisfies the last conjunct of
    $\vpgriddd$.

    We now show that for all  $\AA=(A,(P_{\unary}),\ifunct,\ofunct)$ in $
     \Data{\Unarygrid}$, if $\AA \models \vpgriddd$ then $(A,\query{\vpH}{\AA},\query{\vpV}{\AA})$ is grid-like. By Lemma~\ref{lemma:grid-like}, we just have to prove that $(A,\query{\vpH}{\AA},\query{\vpV}{\AA})$ satisfies $\vpcomplete$ and $\vpprog$.
	Let us prove that $$(A,\query{\vpH}{\AA},\query{\vpV}{\AA})\models
    \forall x.\forall y.\forall x'.\forall y'.((\Hform xy\et \Vform
    xx' \et \Vform yy') \donc \Hform x'y')\,.$$ By definition of
    $(A,\query{\vpH}{\AA},\query{\vpV}{\AA})$, this amounts to
    verifying that
    $$\AA \models \forall x. \forall y.\forall
                       x'.\forall y'.\vpH(x,y)\et \linebreak[0]\vpV(x,x') \et \vpV(y,y')\donc \vpH(x',y')\,.$$
	Let $a,b,a',b'\in A$ and let $I$ be an interpretation function such that
    $\AA
    \models_{I[x/a][y/b][x'/a'][y'/b']}\vpH(x,y)\et\vpV(x,x')\et\vpV(y,y')$. Let us
    show $\AA \models_{I[x/a][y/b][x'/a'][y'/b']}\vpH(x',y')$.
	We do a case analysis on  $i,j\in\{0,1\}$ such that $a \in
    P_{\UH{i}}\cap P_{\UV{j}}$.
	We only perform the proof for the case $(i,j)=(1,0)$, the other three case can
    be treated similarly.
	By looking at $\vpH$ and $\vpV$, we have
	\begin{align*}
		\AA \models_{I[x/a][y/b][x'/a'][y'/b']}&\Pform{\UH{0}}{y} \et \Pform{\UV{0}}{y} \et \rels{2}{2}{x}{y} \\
		\AA \models_{I[x/a][y/b][x'/a'][y'/b']}&\Pform{\UH{0}}{y'} \et \Pform{\UV{1}}{y'} \et \rels{1}{1}{y}{y'} \\
		\AA \models_{I[x/a][y/b][x'/a'][y'/b']}&\Pform{\UH{1}}{x'} \et \Pform{\UV{1}}{x'} \et \rels{1}{1}{x}{x'}.
	\end{align*}
	So $b,a,b'$ are elements of $\vprojr{\AA}{a}{3}$ and
	$$\vprojr{\AA}{a}{3} \models_{I[x/a][y/b][x'/a'][y'/b']} \vpH(x,y)\et \vpV(x,x') \et \vpV(y,y')\,.$$
	Since by assumption $\AA\models\vpcompleteee$, we deduce that
    $\vprojr{\AA}{a}{3} \models_{I[x/a][y/b][x'/a'][y'/b']}
    \vpH(x',y')$. This allows us to conclude that  $\AA
    \models_{I[x/a][y/b][x'/a'][y'/b']} \vpH(x',y')$.

    We can prove in a similar way that $(A,\query{\vpH}{\AA},\query{\vpV}{\AA})\models \vpprog$.
  \end{proof}

Given a domino system $\DD=(D,H_\DD,V_\DD)$, we now provide
a formula $\vpdo$ from the logic
$\rndFO{2}{D,\{\relsaord{1}{1},\linebreak[0]\relsaord{2}{2}\}}{3}$ that guarantees that, if a
data structure corresponding to a grid satisfies $\vpdo$, then
it can be embedded into $\DD$:
$$
  \begin{array}{lcl}
	\vpDD &:= &\forall x. \locformr{ \Ou_{d\in D} \Big(\Pform{d}{x} \et \Et_{d\neq d'\in D} \neg (\Pform{d}{x}\et\Pform{d'}{x})\Big)}{x}{3} \\[0.5ex]
			& &\et  \forall x. \locformr{ \forall y. \vpH(x,y) \to \Ou_{(d,d')\in H_\DD} \Pform{d}{x} \et \Pform{d'}{y}}{x}{3} \\[0.5ex]
			& &\et  \forall x. \locformr{ \forall y. \vpV(x,y) \to \Ou_{(d,d')\in V_\DD} \Pform{d}{x} \et \Pform{d'}{y}}{x}{3}
  \end{array}
  $$

\begin{proposition}
\label{prop:3loc reduction}
	Given $\DD=(D,H_\DD,V_\DD)$ a domino system, 
	$\DD$ admits a periodic tiling iff the $\rndFO{2}{\Unarygrid \uplus D,\{\relsaord{1}{1},\linebreak[0]\relsaord{2}{2}\}}{3}$ formula $\vpgriddd\et\vpdo$ is satisfiable.
  \end{proposition}

\begin{proof}
First  assume that $\DD$ admits a periodic tiling and let $\tau:\grid{m}\to D$ be one.
	As with Lemma~\ref{lemma:grid} we already have that
    $\AA_{2m}\models\vpgriddd$. From $\AA_{2m}$ we build another data
    structure $\AA'_{2m} \in \Data{\Unarygrid \uplus D}$ by adding the
    predicates $(P_d)_{d\in D}$ as follow: 
	for any $i,j\in \{0,2m-1\}$ and $d\in\DD$ we set $P_d((i,j))$ to
    hold iff $\tau((i \mod m,j \mod m))=d$.
	We can then show that $\AA_{2m}\models\vpDD$.
	
Assume now that there exists $\AA= (A,,(P_{\unary}),\ifunct,\ofunct)$
in  $\Data{\Unarygrid \uplus D}$ such
that $\AA \models \vpgriddd\et\vpdo$.
By Lemma \ref{lemma:grid}, there exists $m>0$ and a morphism
$\pi:\grid{m}\to(\AA,\query{\vpH}{\AA},\query{\vpV}{\AA})$. It remains
hence to show that there is a morphism
$\tau:(\AA,\query{\vpH}{\AA},\query{\vpV}{\AA})\to\DD$. For any $a\in A$,  we set $\tau(a)$ to be a domino such that $P_{\tau(a)}(a)$ holds.
	Thanks to the first line of $\vpDD$, $\tau$ is well defined.
	Then thanks to the second and third line of $\vpDD$, we have that
    $\tau$ is a morphism. We deduce that $\tau\circ\pi$ is a periodic tiling of $\DD$.
\end{proof}

As a corollary of the proposition, we obtain the main result
of this section.

\begin{theorem}\label{theorem:3-loc}
$\DataSat{\rndFOr{3},2, \{\relsaord{1}{1},\relsaord{2}{2}\}}$ is undecidable.
\end{theorem}

\subsubsection{The Radius 2 case}

We can also reduce \Tile~to
$\DataSat{\rndFOr{2},2. \{\relsaord{1}{1},\relsaord{2}{2},\linebreak[0]\relsaord{1}{2}\}}$. In that case, it is a
bit more subtle to build a formula similar to the formula $\vpcomplete
$ as we have only neighborhood of radius 2. We use the diagonal
binary relation $\relsaord{1}{2}$ to overcome this.

A \emph{tri-binary
  structure} is a triple $(A,H,V,W)$ where $A$ is a set and $H,V,W$
are three subsets of $A \times A$. Intuitively $H,V$ will capture the
horizontal and vertical adjacency relation whereas $W$ will capture the diagonal adjacency.
By an abuse of notation, $\grid{m}$ will also refer to the tri-binary structure $(\gridcarrier{m},H_m,V_m,W_m)$, were $\gridcarrier{m},H_m$ and $V_m$ are the same as before and:
$$ W_m = \{ ((i,j),(i+1,j+1)) \mid i,j\in \Zmod{m} \} .$$
The logic \emph{$\FO$ over tri-binary structures} is the same as  $\FO$ over bi-binary structures with the addition of the binary symbol $\Wform$.
Let $\vpcompletebis$ be the following FO formula over tri-binary structure:
\begin{align*}
	\vpcompletebis =& \phantom{\wedge} \forall x.\forall y.\forall y'. (\Hform xy \et \Vform yy' \donc \Wform xy')
							~~\et~~ \forall x.\forall x.\forall 'y'.(\Wform xy'\et \Vform xx' \donc \Hform x'y')\,.
\end{align*}

\begin{lemma}
  \label{lem:tribinstruct-grid}
	Let $\grid{}=(A,H,V,W)$ be a tri-binary structure. If $\grid{}$ satisfies $\vpcompletebis$ and $\vpprog$, then $(A,H,V)$ is grid-like.
\end{lemma}
\begin{proof}
	We simply remark that $\vpcompletebis$ implies $\vpcomplete$ and then we apply Lemma \ref{lemma:grid-like}.
  \end{proof}

As in the previous subsection, we will consider data structures in
$\nData{2}{\Unarygrid}$ to encode domino systems and we will use $\rndFO{2}{\Unarygrid,\{\relsaord{1}{1},\linebreak[0]\relsaord{2}{2},\relsaord{1}{2}\}}{2}$ formulas in order to ensure that the
data structures are grid-like and that an embedding of a domino system
in it is feasible. In the previous subsection, to ensure that a data
structure is a grid, we used the fact that we could look in
our logical formulas to neighborhood of radius 3 (cf formula
$\vpgriddd$), but since here we want to look at neighborhoods of radius 2,
we use the diagonal relation and rely on the result of the previous lemma. Consequently, we will need again the two quantifier free formulas $\vpH(x,y)$ and $\vpV(x,y)$ of
$\ndFO{2}{\Unarygrid,\relsaord{1}{1},\relsaord{2}{2}}$ introduced
  previously and
we define a new quantifier free formula $\vpW(x,y)$ in $\ndFO{2}{\Unarygrid,\{\relsaord{1}{2}\}}$:
\begin{align*}
	\vpWzz &= \Pform{\UH{0}}{x} \et \Pform{\UH{1}}{y} \et  \Pform{\UV{0}}{x} \et \Pform{\UV{1}}{y} \et \rels{1}{2}{x}{y} \\
	\vpWoz &= \Pform{\UH{1}}{x} \et \Pform{\UH{0}}{y} \et  \Pform{\UV{0}}{x} \et \Pform{\UV{1}}{y} \et \rels{1}{2}{x}{y} \\
	\vpWzo &= \Pform{\UH{0}}{x} \et \Pform{\UH{1}}{y} \et  \Pform{\UV{1}}{x} \et \Pform{\UV{0}}{y} \et \rels{1}{2}{x}{y} \\
	\vpWoo &= \Pform{\UH{1}}{x} \et \Pform{\UH{0}}{y} \et  \Pform{\UV{1}}{x} \et \Pform{\UV{0}}{y} \et \rels{1}{2}{x}{y} \\
	\vpW &= \vpWzz \ou \vpWoz \ou \vpWzo \ou \vpWoo
\end{align*}

We will now define a formula $\vpgridd$ in
$\rndFO{2}{\Unarygrid,\{\relsaord{1}{1},\linebreak[0]\relsaord{2}{2},\relsaord{1}{2}\}}{2}$  which  ensures that a data structure
corresponds to a grid. This formula is given by ($\oplus$ stands for exclusive or):
$$
  \begin{array}{lcl}
	\vpcompletee &=& \forall x. \locformr{ \forall yy'.\vpH(x,y) \et \vpV(y,y')\donc \vpW(x,y')}{x}{2} \\
					&& \et \forall x. \locformr{ \forall yx'y'.\vpV(x,x') \et \vpW(x,y')\donc \vpH(x',y')}{x}{2} \\[1ex]
	\vpprogg &=& \forall x.\locformr{\exists y.\vpH(x,y)\et\exists y. \vpV(x,y)}{x}{2} \\[1ex]
	\vpgridd& =& \vpcompletee \et  \vpprogg
	\et \forall x.\locformr{ (\Pform{\UH{0}}{x} \oplus \Pform{\UH{1}}{x}) \et  (\Pform{\UV{0}}{x} \oplus \Pform{\UV{1}}{x}) }{x}{2}
  \end{array}
  $$

\begin{figure}[htbp]
	\centering
	\begin{subfigure}{0.45\textwidth}
		\begin{tikzpicture}
		\connectiononetwo{3}{3}
		\draw[shift={(2,2)},line width=3,line cap = butt]  (-\singletonradius,1) arc [start angle=90, end angle=180, radius=1cm-\singletonradius];
		\draw[shift={(2,2)},line width=3,line cap = butt]  (\singletonradius,-1) arc [start angle=-90, end angle=0, radius=1cm-\singletonradius];

	 \roundedsquare{2}{2}{styleclasseone}
	 \roundedsquare{1}{1}{styleclassetwo}
	 \roundedsquare{3}{3}{styleclassetwo}
	 
	 \classesingleton{1}{1}{styleclasseone}
	 \classesingleton{2}{1}{styleclasseone}
	 \classesingleton{1}{2}{styleclasseone}
	 \classesingleton{2}{3}{styleclassetwo}
	 \classesingleton{3}{2}{styleclassetwo}
	 \classesingleton{3}{4}{styleclasseone}
	 \classesingleton{4}{3}{styleclasseone}
	 \classesingleton{4}{4}{styleclasseone}

		\foreach \i in {1,2,3}{
			\drawelement{\i}{\i};
			\drawelement{\i}{\i+1};
			\drawelement{\i+1}{\i};
			}
		\drawelement{4}{4};
		
		\node at  (2,2) [above right] {$a$};
        \node at  (2,3) [above right] {$a'$};
        \node at  (3,2) [above right] {$b$};
        \node at  (3,3) [above right] {$b'$};
	
	\begin{scope}[shift={(1,1)}]
	\draw[style={[-]}] (-.7,-1) -- (3.7,-1);
	\draw[style={[-]}] (-1,-.7) -- (-1,3.7);
	\foreach \i in {0,2}{
	    \draw  (\i,-.95) -- (\i,-1.05) node[below] {$\UH{1}$};
	    \draw (-.95,\i) -- (-1.05,\i) node[left] {$\UV{1}$};}
	\foreach \i in {1,3}{
	    \draw (\i,-.95) -- (\i,-1.05) node[below] {$\UH{0}$};
	    \draw (-.95,\i) -- (-1.05,\i) node[left] {$\UV{0}$};}
	\end{scope}
\end{tikzpicture}
		\caption{If $\UH{0}(a)$ and $\UV{0}(a)$ hold.}
	\end{subfigure}
	\begin{subfigure}{0.45\textwidth}
		\begin{tikzpicture}
		\connectiononetwo{3}{3}
		\draw[line width=3,line cap = butt]  (2.6,1) 
				.. controls (2.2,1) and (2.2,1.5) .. (3-\gap,1.5);
	
		\roundedsquare{2}{2}{styleclasseone}
		\roundedsquare{3}{1}{styleclassetwo}
		\roundedsquare{3}{3}{styleclassetwo}
		
		\classesingleton{3}{1}{styleclasseone}
		\classesingleton{4}{1}{styleclasseone}
		\classesingleton{2}{2}{styleclassetwo}
		\classesingleton{2}{3}{styleclassetwo}
		\classesingleton{3}{4}{styleclasseone}
		\classesingleton{4}{4}{styleclasseone}
		\draw[styleclasseone] (4,2) ++(\gap,0) 
													-- ++(0,1)
													arc [start angle= 0, end angle= 180, radius = \gap]
													-- ++(0,-1)
													arc [start angle= -180, end angle= 0, radius = \gap];

		\foreach \i in {1,2,3,4}{
			\drawelement{3}{\i};
			\drawelement{4}{\i};
			}
		\drawelement{2}{2};
		\drawelement{2}{3};
		
		\node at  (3,2) [below right] {$a$};
        \node at  (3,3) [below right] {$a'$};
        \node at  (4,2) [below right] {$b$};
        	\node at  (4,3) [below right] {$b'$};
	
	\begin{scope}[shift={(2,1)}]
	\draw[style={[-]}] (-.7,-1) -- (2.7,-1);
	\draw[style={[-]}] (-1,-.7) -- (-1,3.7);
	\foreach \i in {0,2}{
	\draw  (\i,-.95) -- (\i,-1.05) node[below] {$\UH{0}$};
	\draw (-.95,\i) -- (-1.05,\i) node[left] {$\UV{1}$};
	}
	\foreach \i in {1,3}{
	\draw (-.95,\i) -- (-1.05,\i) node[left] {$\UV{0}$};
	}
	\draw (1,-.95) -- (1,-1.05) node[below] {$\UH{1}$};
	\end{scope}
 
\end{tikzpicture}
			\caption{If $\UH{1}(a)$ and $\UV{0}(a)$ hold.}
	\end{subfigure}
	\caption{Some 2-local views of $\AA_{2m}$ for
      $\Binary=\{\relsaord{1}{1},\relsaord{2}{2},\relsaord{1}{2}\}$.}
    \label{grid_2}
  \end{figure}

  We can then establish the following result.

  \begin{lemma} \label{lemma:grid2} The following statements hold:
  \begin{enumerate}
  \item $\AA_{2m}\models\vpgridd$, and
  \item for all $\AA=(A,\ifunct,\ofunct,(P_{\unary})) \in
     \Data{\Unarygrid}$, if $\AA \models \vpgridd$, then $(A,\query{\vpH}{\AA},\query{\vpV}{\AA})$ is grid-like.
  \end{enumerate}
\end{lemma}
\begin{proof}
	The proof is similar to the of Lemma \ref{lemma:grid}.  For the
    first point, Figure \ref{grid_2} provides some representation of
    $\vprojr{\AA_{2m}}{a}{2}$ for some elements $a \in
    \gridcarrier{2m}$. For instance, if we look at the case presented
    on Figure \ref{grid_2}(A), we have that $\vprojr{\AA_{2m}}{a}{2}
  \models_{I[x/a][x'/a'][y/b][y'/b']}\vpH(x,y)$ and $\vprojr{\AA_{2m}}{a}{2}
  \models_{I[x/a][x'/a'][y/b][y'/b']}\vpV(y,y')$ and $\vprojr{\AA_{2m}}{a}{2}
  \models_{I[x/a][x'/a'][y/b][y'/b']}\vpV(x,x'')$ and $\vprojr{\AA_{2m}}{a}{2}
  \models_{I[x/a][x'/a'][y/b][y'/b']} \vpWzz(x,y')$. Similarly,  if we
  look at the case presented on Figure  \ref{grid_2}(B), we have that $\vprojr{\AA_{2m}}{a}{2}
  \models_{I[x/a][x'/a'][y/b][y'/b']}\vpH(x,y)$ and $\vprojr{\AA_{2m}}{a}{2}
  \models_{I[x/a][x'/a'][y/b][y'/b']}\vpV(y,y')$ and $\vprojr{\AA_{2m}}{a}{2}
  \models_{I[x/a][x'/a'][y/b][y'/b']}\vpV(x,x'')$ and $\vprojr{\AA_{2m}}{a}{2}
  \models_{I[x/a][x'/a'][y/b][y'/b']} \vpWoz(x,y')$.
For the second point of lemma, following the same
    reasoning as in Lemma \ref{lemma:grid}.2, we first show that the
    tri-binary structure
    $(A,\query{\vpH}{\AA},\query{\vpV}{\AA},\query{\vpW}{\AA})$ satisfies $\vpcompletebis$ and $\vpprog$ and we use Lemma \ref{lem:tribinstruct-grid}
    to conclude.
  \end{proof}

As previously, we provide a formula $\vpDDbis$ of $\rndFO{2}{D,\{\relsaord{1}{1},\linebreak[0]\relsaord{2}{2},\relsaord{1}{2}\}}{2}$ for any domino system
$\DD=(D,H_\DD,V_\DD)$. This formula  is morally the same as the
formula $\vpDD$, we only restrict the neighborhood, but in fact
this does not change anything:
$$
  \begin{array}{lcl}
	\vpDDbis &: = & \phantom{\et}\forall x. \locformr{ \Ou_{d\in D} \Pform{d}{x} \et \Et_{d\neq d'\in D} \neg (\Pform{d}{x}\et\Pform{d'}{x})}{x}{2} \\
			&&\et  \forall x. \locformr{ \forall y. \vpH(x,y) \donc \Ou_{(d,d')\in H_\DD} \Pform{d}{x} \et \Pform{d'}{y}}{x}{2} \\
			&& \et  \forall x. \locformr{ \forall y. \vpV(x,y) \donc \Ou_{(d,d')\in V_\DD} \Pform{d}{x} \et \Pform{d'}{y}}{x}{2}
  \end{array}
  $$

We have the following proposition whose proof follows the same line as
Proposition \ref{prop:3loc reduction}.

\begin{proposition}\label{prop:2loc reduction}
	Given $\DD=(D,H_\DD,V_\DD)$ a domino system, we have that $\DD$
    admits a periodic tiling iff the $\rndFO{2}{\Unarygrid \uplus D,\{\relsaord{1}{1},\linebreak[0]\relsaord{2}{2},\relsaord{1}{2}\}}{2}$ formula $\vpgridd\et\vpDDbis$ is satisfiable.
  \end{proposition}

Finally, we obtain the desired undecidability result.

\begin{theorem}\label{theorem:2-loc}
$\DataSat{\rndFOr{2},2, \{\relsaord{1}{1},\relsaord{2}{2},\relsaord{1}{2}\}}$ is undecidable.
\end{theorem}

\section{Existential Fragment}
\label{sec:existential}
We present in this section results on the existential fragment of
$\rndFO{\nbd}{\Unary,\Binary_\nbd}{r}$ (with $r \geq 1$ and $\nbd \geq 1$ and $\Binary_\nbd$ to represent the
set of binary relation symbols  $\{\rels{i}{j}{}{} \mid
i,j \in \{1,\ldots,\nbd\}\}$) and
establish when it is
decidable. This fragment $\eFO{\nbd}{\linebreak[0]\Unary,\Binary_\nbd}{r}$ is given by the grammar
\[\phi ~::=~ \locformr{\psi}{x}{r} \;\mid\; x=y \;\mid\; \neg(x=y) \;\mid\; \exists x.\phi \;\mid\; \phi\ou\phi  \;\mid\; \phi\et\phi \]
where $\psi$ is a formula from $\ndFO{\nbd}{\Unary,\Binary_\nbd}$
with (at most) one free variable $x$. The quantifier free fragment $\qfFO{\nbd}{\Unary,\Binary_\nbd}{r}$ is
defined by the grammar 
$\phi ~::=~ \locformr{\psi}{x}{r} \;\mid\; x=y \;\mid\; \neg(x=y)
\;\mid\; \phi\ou\phi  \;\mid\; \phi\et\phi $.

\begin{remark}
  Note that for both these
fragments, we do not impose any restrictions on the use of quantifiers in
the formula $\psi$ located under the bracket of the form
$\locformr{\psi}{x}{r}$.
\end{remark}

\subsection{Two Data Values and Balls of Radius 2} 
We prove that the satisfiability problem for the
existential fragment of local first-order logic with two data values and balls of radius two is decidable.
To obtain this result we provide a reduction to the satisfiability
problem for first-order logic over $1$-data structures (see Theorem \ref{thm:1fo}). Our reduction is based on the following intuition. Consider a
$2$-data structure $\AA=(A,(P_{\unary}),\f{1},\f{2}) \in
\nData{2}{\Unary}$ and an element $a \in A$. If we take an
element $b$ in $\Ball{2}{a}{\AA}$, the radius-2-ball around $a$, we
know that either $\f{1}(b)$ or $\f{2}(b)$ is a common value with
$a$. In fact, if $b$ is at distance $1$ of $a$, this holds by definition and 
if $b$ is
at distance $2$ then $b$ shares an element with $c$ at distance $1$ of
$a$ and this element has to be shared with $a$ as well so $b$ ends to
be at distance $1$ of $a$. The
trick consists then in using extra-labels for elements sharing a value with
$a$ that can be forgotten and to keep only the value of $b$ not
present in $a$, this construction leading to a $1$-data structure. In
order to have a sound and complete reduction, we need as well to show
that we can enforce  the $1$-data structures that satisfied
our formula to have a 'good' shape (i.e. they morally can be obtained
from a $2$-data structures by applying the aforementioned
construction) and for this, we provide a formula of $\ndFO{1}{\Unary',\{\relsaord{1}{1}\}}$ (where
$\Unary'$ is obtained from $\Unary$ by adding extra predicates).\\

The first step for our reduction consists in providing a
characterization for the elements located in the radius-1-ball and the radius-2-ball around
another element.

\begin{lemma}\label{lem:shape-balls}
	Let $\AA=(A,(P_{\unary}),\f{1},\f{2}) \in
    \nData{2}{\Unary}$  and $a,b\in A$ and $j \in \{1,2\}$.  We have:
\begin{enumerate}
\item  $(b,j)\in \Ball{1}{a}{\AA}$ iff there is $i\in \{1,2\}$ such that $\relsaa{i}{j}{\AA}{a}{b}$.
\item $(b,j)\in \Ball{2}{a}{\AA}$ iff there exists $i,k\in \{1,2\}$ such that $\relsaa{i}{k}{\AA}{a}{b}$.
\end{enumerate}
\end{lemma}

\begin{proof} We show both statements:
\begin{enumerate}
\item Since $(b,j)\in \Ball{1}{a}{\AA}$, by definition we have either $b=a$ and in that case $\relsaa{j}{j}{\AA}{a}{b}$ holds, or $b \neq a$ and necessarily there exists $i\in \{1,2\}$ such that $\relsaa{i}{j}{\AA}{a}{b}$.
 \item First, if there exists $i,k\in \{1,2\}$ such that
$\relsaa{i}{k}{\AA}{a}{b}$, then $(b,k)\in \Ball{1}{a}{\AA}$ and $(b,j)\in \Ball{2}{a}{\AA}$ by definition. Assume now that $(b,j)\in \Ball{2}{a}{\AA}$. Hence  there exists $i\in \{1,2\}$ such that $\distaa{(a,i)}{(b,j)}{\AA}\leq 2$.
We perform a case analysis on the value of
$\distaa{(a,i)}{(b,j)}{\AA}$.
\begin{itemize}
  \item \textbf{Case $\distaa{(a,i)}{(b,j)}{\AA}=0$}. In that case
    $a=b$ and $i=j$ and we have  $\relsaa{i}{i}{\AA}{a}{b}$.
   \item \textbf{Case $\distaa{(a,i)}{(b,j)}{\AA}=1$}. In that case,
     $((a,i),(b,j))$ is an edge in the data graph $\gaifmanish{\AA}$
     of $\AA$ which means that $\relsaa{i}{j}{\AA}{a}{b}$ holds.
   \item \textbf{Case $\distaa{(a,i)}{(b,j)}{\AA}=2$}. Note that
     we have by definition $a \neq b$. Furthermore, in that case, there is
     $(c,k)\in A\times\{1,2\}$ such that $((a,i),(c,k))$ and
     $((c,k),(b,j))$ are edges in $\gaifmanish{\AA}$. If $c\neq a$ and
     $c\neq b$, this implies that $\relsaa{i}{k}{\AA}{a}{c}$ and
     $\relsaa{k}{j}{\AA}{c}{b}$, so $\relsaa{i}{j}{\AA}{a}{b}$ and
     $\distaa{(a,i)}{(b,j)}{\AA}=1$ which is a contradiction.
	If $c=a$ and $c\neq b$, this implies that $\relsaa{k}{j}{\AA}{a}{b}$.
	If $c\neq a$ and $c = b$, this implies that $\relsaa{i}{k}{\AA}{a}{b}$.
	\qedhere
\end{itemize}
\end{enumerate}
\end{proof}

We consider a formula $\phi=\exists x_1\ldots\exists
x_n.\phi_{qf}(x_1,\ldots,x_n)$ of $\eFO{2}{\Unary,\Binary_2}{2}$ in prenex normal form, i.e., such that $\phi_{qf}(x_1,\ldots,x_n)\in\qfFO{2}{\Unary,\Binary_2}{2}$. We know that there is a structure $\AA=(A,(P_{\unary}),\linebreak[0]\f{1},\f{2})$ in $\nData{2}{\Unary}$ such that $\AA\models\phi$ if and only if there are $a_1,\ldots,a_n \in A $ such that $\AA\models\phi_{qf}(a_1,\ldots,a_n)$.

Let $\AA=(A,(P_{\unary}),\f{1},\f{2})$ be a structure  in $\nData{2}{\Unary}$ and a tuple $\tuple{a} = (a_1,\ldots,a_n)$ of elements in $A^n$. We shall present the construction of a $1$-data structure
$\AAas$ in $\nData{1}{\Unaryp}$ (with $\Unary \subseteq \Unaryp$) with the same set of nodes as $\AA$, but where each node carries a single data value. In order to retrieve the data relations that hold in $\AA$ while reasoning over $\AAas$, we introduce extra-predicates in $\Unaryp$ to establish whether a node shares a common value with one of the nodes among $a_1,\ldots,a_n$ in $\AA$.
\begin{figure*}[htbp]
\centering
	\begin{subfigure}[b]{0.45\textwidth}
\begin{tikzpicture}[node distance=2cm]
	\node [data, label=below left:$a$]                (A)    {1
      \nodepart{two} 2 };
	\node [data, above left of=A,xshift=-1em,label=below:$b$]    (B)    {1 \nodepart{second} 3};
	\node [data, above right of=A,xshift=1em,label=below right:$c$]    (C)    {3 \nodepart{second} 2};
	\node [dataredred, below left of=A,label=below:$d$]    (D)    {5 \nodepart{second} 6};
	\node [dataredred, above right of=B,xshift=1em, label=below:$e$]
    (E)    {4 \nodepart{second} 3};
    \node [data, below right of=A, label=below:$f$]    (F)    {2 \nodepart{second} 7};

	\draw[line width=0.7pt,<->] (A.one north) .. controls +(0,.5) and
    +(.5,0).. (B.one east);
	\draw[line width=0.7pt,<->] (B.two east) .. controls +(2,-0.5) and
    +(-2,.5).. (C.one west);
    \draw[line width=0.7pt,<->] (E.two east) .. controls +(0,0) and
    +(0,0.5).. (C.one north);
    \draw[line width=0.7pt,<->] (E.south west) .. controls +(0,0) and
    +(0.5,.2).. (B.two east);
	\draw[line width=0.7pt,<->] (A.south east) .. controls +(1,-.5)
    and +(0,0).. (C.south west);
    \draw[line width=0.7pt,<->] (A.south) .. controls +(0,-0.5) and
    +(0,0).. (F.one west);
    \draw[line width=0.7pt,<->] (F.north) .. controls +(0,0) and
    +(0,0).. (C.south);
	\selfconnectionright{A};
	\selfconnectionleft{B};
	\selfconnectionright{C};
	\selfconnectionleft{D};
	\selfconnectionleft{E};
	\selfconnectionright{F};

\end{tikzpicture}
\caption{A data structure $\AA$ and  $\gaifmanish{\AA}$.}
\label{fig:abstract-a}
	\end{subfigure}
	\unskip\ \vrule\ \hspace{1em}
	\begin{subfigure}[b]{0.45\textwidth}
\begin{tikzpicture}[node distance=2cm]
	\node [dataone, label=below left:$a$]                (A)    {8};
	\node [dataone, above left of=A,xshift=-1em,label=below:$b$]    (B)    {3};
	\node [dataone, above right of=A,xshift=1em,label=below:$c$]    (C)    {3};
	\node [dataone, below left of=A,label=below:$d$]    (D)    {9};
	\node [dataone, above right of=B,xshift=1em, label=below:$e$]
    (E)    {10};
    \node [dataone, below right of=A, label=below:$f$]    (F)    {7};
    \node [right of=C,yshift=-3em]    (G)    {$\begin{array}{l}\uP{a[1,1]}=\{a,b\} \\
                                     \uP{a[2,2]}=\{a,c\}\\
                                     \uP{a[1,2]}=\emptyset \\
                                   \uP{a[2,1]}=\{f\}\\\end{array}$};

\end{tikzpicture}
\caption{$\sem{\AA}_{(a)}$.}
\label{fig:abstract-b}
	\end{subfigure}
	\caption{
\label{fig:abstract}}
\end{figure*}

We now explain formally how we build $\AAas$. Let $\Udeci{n}=\{\udd{p}{i}{j}\mid p\in\{1,\ldots,n\}, i,j\in\{1,2\}\}$ be a set of new unary predicates and $\Unaryp = \Unary \cup \Udeci{n}$.
For every element $b\in A$, the predicates in $\Udeci{n}$ are used to keep track of the relation between the data values of $b$ and the one of $a_1,\ldots,a_n$ in $\AA$.
Formally, we define $\uP{\udd{p}{i}{j}}=\{b\in A\mid \AA\models \rels{i}{j}{a_p}{b}\}$.
We now define a data function $f:A\to \N$.
We recall for this matter that $\Valuessub{\AA}{\tuple{a}} = \{f_1(a_1),f_2(a_1),\ldots,f_1(a_n),f_2(a_n)\}$ and let $\inj:A\to\N\setminus \Values{\AA}$ be an injection. For every $b \in A$, we set:
\[
	f(b) = \begin{cases}
		f_2(b) \text{ if } f_1(b)\in \Valuessub{\AA}{\tuple{a}} \text{ and } f_2(b)\notin \Valuessub{\AA}{\tuple{a}}\\
		f_1(b) \text{ if } f_1(b)\notin \Valuessub{\AA}{\tuple{a}} \text{ and } f_2(b)\in \Valuessub{\AA}{\tuple{a}}\\
		\inj(b) \text{ otherwise}
    	   \end{cases}
\]
Hence depending if $f_1(b)$ or $f_2(b)$ is in $\Valuessub{\AA}{\tuple{a}}$, it splits the elements of $\AA$ in four categories.
If $f_1(b)$ and $f_2(b)$ are in $\Valuessub{\AA}{\tuple{a}}$, the predicates in $\Udeci{n}$ allow us to retrieve all the data values of $b$. 
Given $j\in\{1,2\}$, if $f_j(b)$ is in $\Valuessub{\AA}{\tuple{a}}$ but $f_{3-j}(b)$ is not, the new predicates will give us the $j$-th data value of $b$ and we have to keep track of the $(3-j)$-th one, so we save it in $f(b)$.
Lastly, if neither $f_1(b)$ nor $f_2(b)$ is in $\Valuessub{\AA}{\tuple{a}}$, we will never be able to see the data values of $b$ in $\phi_{q_f}$ (thanks to Lemma \ref{lem:shape-balls}), so they do not matter to us. Finally, we have  $\AAas = (A, (\uP{\unary})_{\unary\in\Unaryp}, f) $. Figure \ref{fig:abstract-b} provides an example of  $\Valuessub{\AA}{\tuple{a}}$ for the data structures depicted on Figure \ref{fig:abstract-a} and $\tuple{a}=(a)$.
The next lemma formalizes the connection existing between $\AA$ and
$\AAas$ with $\tuple{a} = (a_1,\ldots,a_n)$.

\begin{lemma}\label{lem:r2dv2-semantique}
Let $p\in\{1,\ldots,n\}$ and assume $\vprojr{\AA}{a_p}{2} = (A',(P'_{\unary}),\fp_1,\fp_2)$ (with $A'= \{b \in A \mid (b,i) \in \Ball{2}{a_p}{\AA}\}$). For all $b,c\in A'$ and $j,k\in\{1,2\}$, the following statements hold:

	\begin{enumerate}
  \item $(b,j),(c,k) \in\Ball{2}{a_p}{\AA}$
    \item If $(b,j)\in\Ball{1}{a_p}{\AA}$ and $(c,k)\in\Ball{1}{a_p}{\AA}$ then $\vprojr{\AA}{a_p}{2}\models\rels{j}{k}{b}{c}$ iff there is $i\in\{1,2\}$ s.t. $b \in \uP{\udd{p}{i}{j}}$ and $c \in \uP{\udd{p}{i}{k}}$.
		\item If $(b,j)\in\Ball{2}{a_p}{\AA}\setminus\Ball{1}{a_p}{\AA}$ and $(c,k)\in\Ball{1}{a_p}{\AA}$ then $\vprojr{\AA}{a_p}{2}\nvDash\rels{j}{k}{b}{c}$ 
		\item If $(b,j),(c,k) \in\Ball{2}{a_p}{\AA}\setminus\Ball{1}{a_p}{\AA}$ then $\vprojr{\AA}{a_p}{2}\models\rels{j}{k}{b}{c}$ iff either                $\relsaa{1}{1}{\AAas}{b}{c}$ or there exists $p' \in \{1,\ldots,n\}$ and $\ell \in \{1,2\}$ such that $b \in \uP{\udd{p'}{\ell}{j}}$ and $c \in \uP{\udd{p'}{\ell}{k}}$ .
	\end{enumerate}
\end{lemma}

\begin{proof}
	\begin{enumerate}
  \item We show for instance that $(b,1) \in\Ball{2}{a_p}{\AA}$. Since $b$ belongs to $A'$, it means that $(b,2) \in\Ball{2}{a_p}{\AA}$. Using Lemma \ref{lem:shape-balls}.2, we have $(b,1) \in \Ball{1}{a_p}{\AA}$ or $(b,2) \in \Ball{1}{a_p}{\AA}$, hence  $(b,1) \in\Ball{2}{a_p}{\AA}$.
		\item Assume that $(b,j)\in\Ball{1}{a_p}{\AA}$ and $(c,k)\in\Ball{1}{a_p}{\AA}$.
			It implies that $\fp_j(b)=f_j(b)$ and $\fp_k(c)=f_k(c)$.
			Then assume that $\vprojr{\AA}{a_p}{2}\models\rels{j}{k}{b}{c}$. 
			As $(b,j)\in\Ball{1}{a_p}{\AA}$, thanks to Lemma \ref{lem:shape-balls}.1 it means that there is a $i\in\{1,2\}$ such that $\relsaa{i}{j}{\AA}{a_p}{b}$.
			So we have $f_k(c)=\fp_k(c)=\fp_j(b)=f_j(b)=f_i(a_p)$, that is $\relsaa{i}{k}{\AA}{a_p}{c}$. Hence by definition, $b \in \uP{\udd{p}{i}{j}}$ and $c \in \uP{\udd{p}{i}{k}}$.
			Conversely, let $i\in\{1,2\}$ such that $b \in \uP{\udd{p}{i}{j}}$ and $c \in \uP{\udd{p}{i}{k}}$. This means that $\relsaa{i}{j}{\AA}{a_p}{b}$ and $\relsaa{i}{k}{\AA}{a_p}{c}$.
			So $\fp_j(b)=f_j(b)=f_i(a_p)=f_k(c)=\fp_k(c)$, that is $\vprojr{\AA}{a_p}{2}\models\rels{j}{k}{b}{c}$. 
		\item Assume that $(b,j)\in\Ball{2}{a_p}{\AA}\setminus\Ball{1}{a_p}{\AA}$ and $(c,k)\in\Ball{1}{a_p}{\AA}$.
			It implies that $\fp_j(b)=f_j(b)$ and $\fp_k(c)=f_k(c)$.
			Thanks to Lemma \ref{lem:shape-balls}.1, $(c,k)\in\Ball{1}{a_p}{\AA}$ implies that $f_k(c)\in\{f_1(a_p),f_2(a_p)\}$ and $(b,j)\notin\Ball{1}{a_p}{\AA}$ implies that $f_j(b)\notin\{f_1(a_p),f_2(a_p)\}$.
			So  $\vprojr{\AA}{a_p}{2}\not \models\rels{j}{k}{b}{c}$.
		\item Assume that $(b,j), (c,k) \in\Ball{2}{a_p}{\AA}\setminus\Ball{1}{a_p}{\AA}$. As previously, we have that $f_j(b)\notin\{f_1(a_p),f_2(a_p)\}$ and $f_k(c)\notin\{f_1(a_p),f_2(a_p)\}$, and thanks to Lemma \ref{lem:shape-balls}.2, we have $f_{3-j}(b) \in \{f_1(a_p),f_2(a_p)\}$ and $f_{3-k}(b) \in \{f_1(a_p),f_2(a_p)\}$. There is then two cases:
    \begin{itemize}
    \item Suppose there does not exists $p' \in \{1,\ldots,n\}$ such that $f_{j}(b) \in \{f_1(a_{p'}),f_2(a_{p'})\}$ .This allows us to deduce that $\fp_j(b)=f_j(b)=f(b)$ and $\fp_k(c)=f_k(c)$. If $\vprojr{\AA}{a_p}{2}\models\rels{j}{k}{b}{c}$, then necessarily there does not exists $p' \in \{1,\ldots,n\}$ such that $f_{k}(c) \in \{f_1(a_{p'}),f_2(a_{p'})\}$ so we have $\fp_k(c)=f_k(c)=f(c)$ and  $f(b)=f(c)$, consequently $\relsaa{1}{1}{\AAas}{b}{c}$. Similarly assume that  $\relsaa{1}{1}{\AAas}{b}{c}$, this means that $f(b)=f(c)$ and either $b=c$ and $k=j$ or $b \neq c$ and by injectivity of $f$,we have $f_j(b)=f(b)=f_k(c)$. This allows us to deduce that $\vprojr{\AA}{a_p}{2}\models\rels{j}{k}{b}{c}$.
    \item  If there exists $p' \in \{1,\ldots,n\}$ such that $f_{j}(b) = f_\ell(a_{p'})$ for some $\ell \in \{1,2\}$. Then we have $b \in       \uP{\udd{p'}{\ell}{j}}$. Consequently, we have $\vprojr{\AA}{a_p}{2}\models\rels{j}{k}{b}{c}$ iff $c \in \uP{\udd{p'}{\ell}{k}}$. \qedhere
   \end{itemize}
	\end{enumerate}
\end{proof}

We shall now see how we translate the formula
$\phi_{qf}(x_1,\linebreak[0]\ldots,\linebreak[0]x_n)$ into a formula
$\phit{\phi_{qf}}(x_1,\linebreak[0]\ldots,x_n)$ in $\ndFO{1}{\Unaryp,\{\relsaord{1}{1}\}}$ such that $\AA$ satisfies $\phi_{qf}(a_1,\ldots,a_n)$ if, and only if, $\AAas$ satisfies $\phit{\phi_{qf}}(a_1,\ldots,a_n)$. Thanks to the previous lemma we know that if $\vprojr{\AA}{a_p}{2}\models\rels{j}{k}{b}{c}$ then $(b,j)$ and $(c,k)$ must belong to the  same set among $\Ball{1}{a_p}{\AA}$ and $\Ball{2}{a_p}{\AA}\setminus\Ball{1}{a_p}{\AA}$  and we can test in $\AAas$ whether $(b,j)$ is a member of  $\Ball{1}{a_p}{\AA}$ or $\Ball{2}{a_p}{\AA}$.
Indeed, thanks to Lemmas \ref{lem:shape-balls}.1 and \ref{lem:shape-balls}.2, we have $(b,j) \in \Ball{1}{a_p}{\AA}$ iff $b\in\bigcup_{i=1,2}\uP{\udd{p}{i}{j}}$ and $(b,j) \in \Ball{2}{a_p}{\AA}$ iff $b\in\bigcup_{i=1,2}^{j'=1,2} \uP{\udd{p}{i}{j'}}$. This reasoning leads to the  following formulas in $\ndFO{1}{\Unaryp,\{\relsaord{1}{1}\}}$ with $p \in \{1,\ldots,n\}$ and $j \in \{1,2\}$:
\begin{itemize}
\item $\phiBun{j}(y) := \udd{p}{1}{j}(y) \ou \udd{p}{2}{j}(y)$ to test if the $j$-th field of an element belongs to $\Ball{1}{a_p}{\AA}$
\item $\phiBdeux(y) := \phiBun{1}(y) \ou \phiBun{2}(y)$ to test if a field of an element belongs to $\Ball{2}{a_p}{\AA}$
\item $\phiBdsu{j}(y) := \phiBdeux(y) \et \neg\phiBun{j}(y)$ to test that the $j$-th field of an element belongs to $\Ball{2}{a_p}{\AA}\setminus\Ball{1}{a_p}{\AA}$
\end{itemize}
  
We shall now present how we use these formulas to translate atomic formulas of the form  $\rels{j}{k}{y}{z}$ under some $\locformr{-}{x_p}{2}$. For this matter, we rely on the two following formulas of $\ndFO{1}{\Unaryp}$:
\begin{itemize}
\item The first formula asks  for $(y,j)$ and $(z,k)$ to be in $\Ball{1}{a_p}{1}$ (where here we abuse notations, using variables for the elements they represent) and for these two data values to coincide with one data value of $a_p$, it corresponds to Lemma \ref{lem:r2dv2-semantique}.2:
  $$
  \phiun(y,z) := \phiBun{j}(y) \et \phiBun{k}(z) \et \\\Ou_ {i=1,2}\udd{p}{i}{j}(y)\et\udd{p}{i}{k}(z) 
  $$
\item The second formula asks for $(y,j)$ and $(z,k)$ to be in $\Ball{2}{a_p}{\AA}\setminus\Ball{1}{a_p}{\AA}$ and checks either whether the data values of $y$ and $z$ in $\AAas$ are equal or whether there exist $p'$ and $\ell$ such that $y$ belongs to $\udd{p'}{\ell}{j}(y)$ and $z$ belongs to $\udd{p'}{\ell}{k}(z)$, it corresponds to Lemma \ref{lem:r2dv2-semantique}.4:
  \begin{align*}
  \phideux(y,z) := & \phiBdsu{j}(y) \et \phiBdsu{k}(z) \et \\ & \big (y\sim z \ou\big(\Ou^n_{p'=1}\Ou^2_ {\ell=1}\udd{p'}{\ell}{j}(y)\et\udd{p'}{\ell}{k}(z)\big)\big)  
  \end{align*}
\end{itemize}

	Finally, here is the inductive definition of the translation $\T{-}$ which uses sub transformations $\Tp{-}$ in order to remember the centre of the ball and leads to the construction of $\phit{\phi_{qf}}(x_1,\ldots,x_n)$:
\[	\begin{array}{rcl}
		\T{\phi\ou\phi'} &=& \T{\phi} \ou \T{\phi'}\\
		\T{x_p=x_p'}  &=& x_p=x_p'            \\      
		\T{\neg\phi} &=& \neg\T{\phi}       \\   
      \T{\locformr{\psi}{x_p}{2}}  &=& \Tp{\psi}  \\
      \Tp{\rels{j}{k}{y}{z}} &=&\phiun(y,z) \ou \phideux(y,z)\\
      \Tp{\unary(x)} &=& \unary(x)  \\
      \Tp{x=y} &=& x=y \\
      \Tp{\phi\ou\phi'}&=& \Tp{\phi} \ou \Tp{\phi'} \\
      \Tp{\neg\phi} &=& \neg\Tp{\phi}\\
      \Tp{\exists x. \phi} &=& \exists x.\phiBdeux(x) \wedge \Tp{\phi}\\
	\end{array}\]

\begin{lemma} \label{lem:correct}
	We have $\AA\models\phi_{qf}(\tuple{a})$ iff $\AAas\models\phit{\phi_{qf}}(\tuple{a})$.
\end{lemma}
\begin{proof}
Since $\phit{\phi_{qf}}(x_1,\ldots,x_n)$ belongs to the quantifier free fragment of $\rndFO{2}{\Unary,\Binary_2}{2}$, the only place where quantifiers can be used are in the local modalities $\locformr{\psi}{x_p}{2}$. As a matter of fact, because of the inductive definition of $\T{\phi}$ and that only the  formulas $\exists x. \phi$ and $\rels{j}{k}{y}{z}$ change, we only have to prove that $\vprojr{\AA}{a_p}{2}\models\exists x. \phi$ iff $\AAas\models \Tp{\exists x. \phi}$ and that   given $b,c\in A$, we have $\vprojr{\AA}{a_p}{2}\models\rels{j}{k}{b}{c}$ iff $\AAas\models \Tp{\rels{j}{k}{y}{z}}(b,c)$.

To show that $\vprojr{\AA}{a_p}{2}\models\exists x. \phi$ iff
$\AAas\models \Tp{\exists x. \phi}$, we recall that $\vprojr{\AA}{a_p}{2}  =
(A',(P'_{\unary}),\f{1}',\ldots,\f{n}')$ with $A'
=\{b \in A \mid (b,i) \in \Ball{r}{a}{\AA}$ for some
$i \in \{1,2\}\}$. Hence when we quantify existentially in $\AAas$, we should only consider the element belonging to $A'$. Thanks to Lemma \ref{lem:r2dv2-semantique}.1, we know that if $b \in A'$ then $(b,1)$ and $(b,2)$ belong to $\Ball{2}{a_p}{\AA}$. Consequently the formula $\phiBdeux(x)$ allows to focus on elements in $A'$.

We now prove that  given $b,c\in A$, we have $\vprojr{\AA}{a_p}{2}\models\rels{j}{k}{b}{c}$ iff $\AAas\models \Tp{\rels{j}{k}{y}{z}}(b,c)$ and we first suppose that $\vprojr{\AA}{a_p}{2}\models\rels{j}{k}{b}{c}$ and assume $\vprojr{\AA}{a_p}{2}  =
(A',(P'_{\unary}),\f{1}',\ldots,\f{n}')$ .
    Using Lemma \ref{lem:r2dv2-semantique},  it implies that $b,c \in A'$ and $(b,j)$ and $(c,k)$ belong to same set between $\Ball{1}{a_p}{\AA}$ and  $\Ball{2}{a_p}{\AA} \setminus \Ball{1}{a_p}{\AA}$ . We proceed by a case analysis.
    \begin{itemize}
	\item If $(b,j),(c,k)\in\Ball{1}{a_p}{\AA}$ then by lemma \ref{lem:r2dv2-semantique}.2 we have that $\AAas\models\phiun(b,c)$ and thus $\AAas\models \Tp{\rels{j}{k}{y}{z}}(b,c)$.
      
	\item If $(b,j),(c,k)\in\Ball{2}{a_p}{\AA} \setminus \Ball{1}{a_p}{\AA}$ then by lemma \ref{lem:r2dv2-semantique}.4 we have that $\AAas\models\phideux(b,c)$ and thus $\AAas\models \Tp{\rels{j}{k}{y}{z}}(b,c)$.
    \end{itemize}

 We now suppose that $\AAas\models \Tp{\rels{j}{k}{y}{z}}(b,c)$.
	It means that $\AAas$ satisfies $\phiun(b,c)$ or $\phideux(b,c)$.
	If $\AAas\models\phiun(b,c)$, it implies that $(b,j)$ and $(c,k)$ are in $\Ball{1}{a_p}{\AA}$, and we can then apply lemma \ref{lem:r2dv2-semantique}.2 to deduce that $\vprojr{\AA}{a_p}{2}\models\rels{j}{k}{b}{c}$.
	If $\AAas\models\phideux(b,c)$, it implies that $(b,j)$ and $(c,k)$ are in $\Ball{2}{a_p}{\AA} \setminus \Ball{1}{a_p}{\AA}$, and we can then apply lemma \ref{lem:r2dv2-semantique}.4 to deduce that $\vprojr{\AA}{a_p}{2}\models\rels{j}{k}{b}{c}$. \qedhere	
  \end{proof}
  
  \medskip

To provide a reduction from  $\nDataSat{\eFOr{2}}{2,\Gamma_2}$ to
$\nDataSat{\ndFOr}{1,\Gamma_1}$, having the formula $\phit{\phi_{qf}}(x_1,\ldots,x_n)$
is not enough because to use the result of the previous
lemma, we need to ensure that there exists a model $\BB$ and a tuple
of elements $(a_1,\ldots,a_n)$ such that $\BB \models\
\phit{\phi_{qf}}(a_1,\ldots,a_n)$ and as well that there exists
$\AA\in \nData{2}{\Unary}$ such that $ \BB = \AAas$. We explain now how we
can ensure this last point.

 Now, we want to characterize the structures of the form $\AAas$.
Given a structure $\BB =
(A,(\uP{\unary})_{\unary\in\Unaryp},f)\in\nData{1}{\Unaryp}$ and
$\tuple{a}\in A$, we say that $(\BB,\tuple{a})$ is \emph{well formed}
iff there exists a structure $\AA\in \nData{2}{\Unary}$ such that $ \BB
= \AAas$. Hence $(\BB,\tuple{a})$ is \emph{well formed} iff there
exist two  functions $f_1,f_2:A\to\N$ such that $\AAas=\sem{(A,(\uP{\unary})_{\unary\in\Unary}, f_1,f_2)}_{\tuple{a}}$.
We state three properties on $(\BB,\tuple{a})$, and we will show that they characterize being well formed.
\begin{enumerate}
	\item (Transitivity) For all $b,c\in A$, $p,q \in\{1,\ldots,n\}$,
      $i,j,k,\ell \in\{1,2\}$ if $b\in\uP{\udd{p}{i}{j}}$, $c\in\uP{\udd{p}{i}{\ell}}$ and  $b\in\uP{\udd{q}{k}{j}}$ then $c\in\uP{\udd{q}{k}{\ell}}$.
	\item (Reflexivity) For all $p$ and $i$, we have $a_p\in\uP{\udd{p}{i}{i}}$
	\item (Uniqueness) For all $b\in A$, if $b\in\bigcap_{j=1,2}\bigcup_{p=1,\ldots,n}^{i=1,2} \uP{\udd{p}{i}{j}}$ or $b\notin\bigcup_{j=1,2}\bigcup_{p=1,\ldots,n}^{i=1,2} \uP{\udd{p}{i}{j}}$ then for any $c\in B$ such that $f(c)=f(b)$ we have $c=b$.
\end{enumerate}
Each property can be expressed by a first order logic formula, which
we respectively name $\phitran$, $\phirefl$ and $\phiuniq$ and  we
denote by $\phiwf$ their conjunction:
$$
  \hspace*{-3pt}\begin{array}{ll}
\phitran &\!= \forall y \forall z.\Et_{p,q=1}^{n}\Et_{i,j,k,\ell=1}^2 \Big(\!\udd{p}{i}{j}(y) \et \udd{p}{i}{\ell}(z) \et \udd{q}{k}{j}(y) \donc \udd{q}{k}{\ell}(z)\!\Big) \\
\phirefl(x_1,\ldots,x_n)  &\!=\Et_{p=1}^n\Et_{i=1}^2	 \udd{p}{i}{i}(x_p) \\
\phiuniq &\!= \forall y. \Big(\Et_{j=1}^2 \Ou^n_{p=1} \Ou_{i=1}^2 \udd{p}{i}{j}(y) \ou \Et_{j=1}^2 \Et^n_{p=1}\Et^2_{i=1} \neg\udd{p}{i}{j}(y)\Big) \donc\\ &\qquad\qquad (\forall z. y\sim z \donc y=z)\\
\phiwf(x_1,\ldots,x_n) &\!=\phitran \et \phirefl(x_1,\ldots,x_n) \et
                         \phiuniq
  \end{array}
  $$

The next lemma expresses that the formula $\phiwf$ allows to
characterise precisely the $1$-data structures in $\nData{1}{\Unaryp}$
which are well-formed.

\begin{lemma}\label{lem:well-formed}
	Let $\BB\in\nData{1}{\Unaryp}$ and $a_1,\ldots,a_n$ elements of
    $\BB$, then $(\BB,\tuple{a})$ is well formed iff $\BB\models\phiwf(\tuple{a})$.
  \end{lemma}

  \begin{proof}
    First, if $(\BB,\tuple{a})$ is well formed, then there there exists
$\AA\in \nData{2}{\Unary}$ such that $ \BB = \AAas$ and by
construction we have $\AAas \models\phiwf(\tuple{a})$. We now suppose
that $\BB=(A,(\uP{\unary})_{\unary\in\Unaryp},f)$ and $\BB\models\phiwf(\tuple{a})$.	
	In order to define the functions $f_1,f_2:A\to\N$, we  need
    to introduce some objects.

    We first define a function $g :
    \{1,\ldots,n\} \times \{1,2\} \to \N\setminus \im{f}$ (where
    $\im{f}$ is the image of $f$ in $\BB$) which
      verifies the following properties:
      \begin{itemize}
        \item for all $p \in \{1,\ldots,n\}$ and $i \in \{1,2\}$, we
          have 
          $a_p \in \uP{\udd{p}{i}{3-i}} $ iff $g(p,1)=g(p,2)$;
        \item for all $p, q \in \{1,\ldots,n\}$ and $i,j \in \{1,2\}$,
          we have  $a_q \in \uP{\udd{p}{i}{j}} $ iff $g(p,i)=g(q,j)$.
        \end{itemize}
We use this function to fix the two data values carried by the
elements in $\{a_1,\ldots,a_m\}$. We now explain why this function is
well founded, it is due to the fact
that $\BB\models\phitran \et \phirefl(a_1,\ldots,a_n)$. In fact, since
$\BB \models \phirefl(a_1,\ldots,a_n)$, we have for all $p \in
\{1,\ldots,n\}$ and $i \in \{1,2\}$, $a_p \in \uP{\udd{p}{i}{i}}
$. Furthermore if $a_p \in \uP{\udd{p}{i}{j}}$ then $a_p \in
\uP{\udd{p}{j}{i}}$ thanks to the formula $\phitran$; indeed since we
have $a_p \in  \uP{\udd{p}{i}{j}}$ and $a_p \in  \uP{\udd{p}{i}{i}}$
and  $a_p \in  \uP{\udd{p}{j}{j}}$, we obtain $a_p \in
\uP{\udd{p}{j}{i}}$. Next, we also have that if $a_q \in
\uP{\udd{p}{i}{j}}$ then $a_p \in
\uP{\udd{q}{j}{i}}$ again thanks to $\phitran$;  indeed since we
have $a_q \in  \uP{\udd{p}{i}{j}}$ and $a_p \in  \uP{\udd{p}{i}{i}}$
and  $a_q \in  \uP{\udd{q}{j}{j}}$, we obtain $a_p \in
\uP{\udd{q}{j}{i}}$.

We also need a natural $\dout$ belonging to $\N\setminus
    (\im{g}\cup\im{f})$. For $j \in
    \{1,2\}$, we define $f_j$ as follows for all $b \in A$:
	\[f_j(b) = \left\{\begin{array}{ll}
		g(p,i) & \text{if for some } p,i \text{ we have } b\in\uP{\udd{p}{i}{j}} \\
		f(b) &\text{if for all $p,i$ we have $b\notin\uP{\udd{p}{i}{j}}$ and for some $p,i$ we have $b\in\uP{\udd{p}{i}{3-j}}$} \\
						 \dout &\text{if for all $p,i,j'$, we have } b\notin\uP{\udd{p}{i}{j'}}
	           \end{array}\right.
	\]

Here again, we can show that since $\BB\models\phitran \et
\phirefl(a_1,\ldots,a_n)$, the functions $f_1$ and $f_2$ are well
founded. Indeed, assume that  $b\in\uP{\udd{p}{i}{j}} \cap
  \uP{\udd{q}{k}{j}}$, then we have necessarily that
  $g(p,i)=g(q,k)$. For this we need to show that $a_p \in
  \udd{q}{k}{i}$ and we use again the formula $\phitran$. This can be
  obtained because we have $b\in\uP{\udd{p}{i}{j}}$ and
  $a_p\in\uP{\udd{p}{i}{i}}$ and $b \in \uP{\udd{q}{k}{j}}$.

 We   then define $\AA$ as the $2$-data-structures
$(A,(P_{\unary})_{\unary \in \Unary},\f{1},\f{2})$. It remains to
prove that $\BB = \AAas$. 

First, note  that for all  $b\in A$, $p \in \{1,\ldots,n\}$ and
$i,j\in\{1,2\}$, we have $b\in\uP{\udd{p}{i}{j}}$ iff
$\relsaa{i}{j}{\AA}{a_p}{b}$. Indeed, we have $b\in\uP{\udd{p}{i}{j}}$,
we have that $f_j(b)=g(p,i)$ and since $a_p \in \uP{\udd{p}{i}{j}}$ we
have as well that $f_i(a_p)=g(p,i)$, as a consequence
$\relsaa{i}{j}{\AA}{a_p}{b}$. In the other direction, if
$\relsaa{i}{j}{\AA}{a_p}{b}$, it means that $f_j(b)=f_i(a_p)=g(p,i)$
and thus  $b\in\uP{\udd{p}{i}{j}}$. Now to have $\BB = \AAas$, one has
only to be careful in the choice of  function $\inj$ 
while building $\AAas$. We recall that this function is injective and is
used to give a value to the elements $b \in A$ such that neither
$f_1(b)\in \Valuessub{\AA}{\tuple{a}} \text{ and } f_2(b)\notin
\Valuessub{\AA}{\tuple{a}}$ nor $ f_1(b)\notin
\Valuessub{\AA}{\tuple{a}} \text{ and } f_2(b)\in
\Valuessub{\AA}{\tuple{a}}$. For these elements, we make $\inj$
matches with the function $f$ and the fact that we define an injection
is guaranteed by the formula $\phiuniq$.
\end{proof}

Using the results of Lemma \ref{lem:correct} and
\ref{lem:well-formed}, we deduce that the formula $\phi\!=\!\exists x_1\linebreak[0]\ldots\exists
x_n.\phi_{qf}(x_1,\linebreak[0]\ldots,x_n)$ of $\eFO{2}{\Unary,\Binary_2}{2}$ is satisfiable
iff the formula $\psi=\exists x_1\ldots\exists
x_n.\phit{\phi_{qf}}(x_1,\ldots,x_n) \linebreak[0]\wedge \phiwf(x_1,\ldots,x_n) $
is satisfiable. Note that $\psi$ can be built in polynomial time from
$\phi$ and that it belongs to $\ndFO{1}{\Unaryp,\Binary_1}$. Hence, thanks to
Theorem \ref{thm:1fo}, we obtain that $\nDataSat{\eFOr{2}}{2}$ is in
\textsc{N2EXP}. 

We can as well obtain a matching lower bound thanks to a
reduction from the problem $\nDataSat{\ndFOr}{1,\Gamma_1}$. For this matter we rely on two
crucial points. First in the formulas of $\eFO{2}{\Unary,\Binary_2}{2}$, there is no
restriction on the use of quantifiers for the formulas located under the scope
of the $\locformr{\cdot}{x}{2}$ modality and consequently we can write
inside this modality a formula of $\ndFO{1}{\Unary}$ without any
modification. Second we can extend a
model of $\ndFO{1}{\Unary,\Binary_1}$ into a $2$-data structure such that all
elements and their values are located in the same radius-$2$-ball by adding everywhere a second
data value equal to $0$. More formally, let $\phi$ be
a formula in $\ndFO{1}{\Unary,\Binary_1}$ and consider the formula $\exists
x.\locformr{\phi}{x}{2}$ where we interpret $\phi$ over $2$-data
structures (this formula simply never mentions the values located in the second
fields). We have then the following lemma.

\begin{lemma} \label{lem:hardness-radius2-2}
 There exists $\AA \in \nData{1}{\Unary}$ such that $\AA \models \phi$
 if and only if there exists $\BB \in  \nData{2}{\Unary}$ such that
 $\BB \models \exists
x.\locformr{\phi}{x}{2}$.
\end{lemma}

\begin{proof}
 Assume that there exists  $\AA=(A,(P_{\unary})_{\unary \in
   \Unary},\f{1})$ in $\nData{1}{\Unary}$ such that  $\AA \models
 \phi$. Consider the $2$-data structure $\BB=(A,(P_{\unary})_{\unary \in
   \Unary},\f{1},\f{2})$ such that $\f{2}(a)=0$ for all $a\in
 A$. Let $a \in A$. It is clear that we have $\vprojr{\BB}{a}{2}=\BB$
 and that $\vprojr{\BB}{a}{2} \models \phi$ (because $\AA \models
 \phi$ and $\phi$ never mentions the second values of the elements
 since it is a formula in $\ndFO{1}{\Unary}$ ). Consequently $\BB \models \exists
 x.\locformr{\phi}{x}{2}$.

 Assume now that there exists $\BB=(A,(P_{\unary}),\f{1},\f{2})$ in $ \nData{2}{\Unary}$ such that  $\BB \models \exists
x.\locformr{\phi}{x}{2}$. Hence there exists $a \in A$ such that
$\vprojr{\BB}{a}{2} \models \phi$, but then by forgetting the second
value in $\vprojr{\BB}{a}{2}$ we obtain a model in $\nData{1}{\Unary}$
which satisfies $\phi$.
\end{proof}
  
Since $\nDataSat{\ndFOr}{1,\Binary_1}$ is
\textsc{N2EXP}-hard (see Theorem \ref{thm:1fo}), we obtain the desired lower bound. 

\begin{theorem}\label{thm:radius2-2}
	The problem $\nDataSat{\eFOr{2}}{2,\Binary_2}$ is \textsc{N2EXP}-complete.
\end{theorem}

\subsection{Balls of radius 1 and any number of data values}
Let $\nbd \geq 1$. We first show that  $\nDataSat{\eFOr{1}}{\nbd,\Gamma_\nbd}$ is in
\textsc{NEXP} by providing a reduction towards
$\nDataSat{\ndFOr}{0,\emptyset}$. This reduction uses the characterization of
the radius-1-ball provided by Lemma  \ref{lem:shape-balls} and is very
similar to the reduction provided in the previous section.  In fact,
 for an element $b$ located in the radius-1-ball of another
element $a$, we use extra unary predicates to explicitly characterise which are the
values of $b$ that are common with the  values of $a$. We provide here
the main step of this reduction whose proof follows the  same line as
the one of Theorem \ref{thm:radius2-2}.

We consider a formula $\phi=\exists x_1\ldots\exists
x_n.\phi_{qf}(x_1,\ldots,x_n)$ of $\eFO{\nbd}{\Unary,\Gamma_1}{1}$ in prenex normal
form, i.e., such that $\phi_{qf}(x_1,\ldots,x_n)\in\qfFO{\nbd}{\Unary,\Gamma_1}{1}$. We
know that there is a structure $\AA=(A,(P_{\unary}),\linebreak[0]\f{1},\f{2},\ldots,\f{\nbd})$ in
$\nData{\nbd}{\Unary}$ such that $\AA\models\phi$ if and only if there
are $a_1,\ldots,a_n \in A $ such that
$\AA\models\phi_{qf}(a_1,\ldots,a_n)$. Let then $\AA=(A,(P_{\unary}),\f{1},\f{2},\ldots,\f{\nbd})$ in $\nData{\nbd}{\Unary}$ and a tuple $\tuple{a} = (a_1,\ldots,a_n)$ of elements in $A^n$. Let $\Omega_n=\{\udd{p}{i}{j}\mid p\in\{1,\ldots,n\}, i,j\in\{1,\ldots,\nbd\}\}$ be a set of new unary predicates and $\Unaryp = \Unary \cup \Omega_n$.
For every element $b\in A$, the predicates in $\Omega_n$ are used to keep track of the relation between the data values of $b$ and the ones of $a_1,\ldots,a_n$ in $\AA$.
Formally, we have $\uP{\udd{p}{i}{j}}=\{b\in A\mid \AA\models \rels{i}{j}{a_p}{b}\}$.
Finally, we build the $0$-data-structure
$\sem{\AA}'_{\tuple{a}}= (A, (\uP{\unary})_{\unary\in\Unaryp})
$. Similarly to Lemma \ref{lem:r2dv2-semantique}, we have the
following connection between $\AA$ and $\sem{\AA}'_{\tuple{a}}$.

\begin{lemma}\label{lem:r1-semantique}
Let $p\in\{1,\ldots,n\}$ and assume $\vprojr{\AA}{a_p}{1} =(A',(P'_{\unary}),\f{1},\f{2},\ldots,\f{\nbd})$ (with $A'= \{b \in A \mid (b,i) \in \Ball{1}{a_p}{\AA}\}$)
	For all $b,c\in A'$ and $j,k\in\{1,\ldots,\nbd\}$, the following statements hold:
	\begin{enumerate}
    \item If $(b,j)\in\Ball{1}{a_p}{\AA}$ and $(c,k)\in\Ball{1}{a_p}{\AA}$ then $\vprojr{\AA}{a_p}{1}\models\rels{j}{k}{b}{c}$ iff there is $i\in\{1,\ldots,\nbd\}$ s.t. $b \in \uP{\udd{p}{i}{j}}$ and $c \in \uP{\udd{p}{i}{k}}$.
		\item If $(b,j)\notin\Ball{1}{a_p}{\AA}$ and $(c,k)\in\Ball{1}{a_p}{\AA}$ then $\vprojr{\AA}{a_p}{1}\nvDash\rels{j}{k}{b}{c}$
		\item If $(b,j)\notin\Ball{1}{a_p}{\AA}$ and $(c,k)\notin\Ball{1}{a_p}{\AA}$ then $\vprojr{\AA}{a_p}{1}\models\rels{j}{k}{b}{c}$ iff $b=c$ and $j=k$.
	\end{enumerate}
\end{lemma}

We now translate
$\phi_{qf}(x_1,\ldots,x_n)$ into a formula
$\phit{\phi_{qf}}'(x_1,\ldots,\linebreak[0]x_n)$ in $\ndFO{0}{\Unaryp,\emptyset}$ such that $\AA$
satisfies $\phi_{qf}(a_1,\ldots,\linebreak[0]a_n)$ if, and only if,
$\sem{\AA}'_{\tuple{a}}$ satisfies
$\phit{\phi_{qf}}'(a_1,\ldots,a_n)$. As in the previous section, we
introduce the two following formulas in $\ndFO{0}{\Unaryp,\emptyset}$ with $p \in
\{1,\ldots,n\}$ and $j \in \{1,\ldots,\nbd\}$ to:
\begin{itemize}
\item $\phiBun{j}(y) := \bigvee_{i \in \{1,\ldots,\nbd\}}\udd{p}{i}{j}(y)$ to test if the $j$-th field of an element belongs to $\Ball{1}{a_p}{\AA}$
\item $\phiBunbis(y) :=  \bigvee_{j \in \{1,\ldots,\nbd\}} \phiBun{j}(y)$ to test if a field of an element belongs to  $\Ball{1}{a_p}{\AA}$
\end{itemize}

We now present how we  translate atomic formulas of the form  $\rels{j}{k}{y}{z}$ under some $\locformr{-}{x_p}{1}$. For this matter, we rely on two formulas of $\ndFO{0}{\Unaryp,\emptyset}$ which can be described as follows:
\begin{itemize}
\item The first formula asks  for $(y,j)$ and $(z,k)$ to be in $\Ball{1}{a_p}{1}$ (here we abuse notations, using variables for the elements they represent) and for these two data values to coincide with one data value of $a_p$, it corresponds to Lemma \ref{lem:r1-semantique}.1:
  $$
  \psiun(y,z) := \phiBun{j}(y) \et \phiBun{k}(z) \et \Ou^\nbd_ {i=1}\udd{p}{i}{j}(y)\et\udd{p}{i}{k}(z) 
  $$
\item The second formula asks for $(y,j)$ and $(z,k)$ to not belong to $\Ball{1}{a_p}{\AA}$ and for $y=z$, it corresponds to Lemma \ref{lem:r1-semantique}.3:
  $$
  \psideux(y,z) := \begin{cases}
		                                	\bigwedge^\nbd_{i=1} (\neg
                                            \phiBun{i}(y) \wedge \neg  \phiBun{i}(z))  \et y=z  &\text{ if } j=k \\
																			\bot  &\text{ otherwise}
		                                \end{cases}
  $$
\end{itemize}

	Finally, as before we provide an inductive definition of the
    translation $\Tbis{-}$ which uses subtransformations $\Tpbis{-}$ in
    order to remember the centre of the ball and leads to the
    construction of $\phit{\phi_{qf}}'(x_1,\ldots,x_n)$. We only
    detail the cases:
    \begin{align*}
    \Tpbis{\rels{j}{k}{y}{z}} =& \psiun(y,z) \ou
    \psideux(y,z)\\
    \Tpbis{\exists x. \phi} =& \exists x.\phiBunbis(x)\wedge \Tpbis{\phi}
    \end{align*}  
    as the other cases are identical as for the
    translation $\T{-}$ shown in the previous section. This leads to
    the following lemma (which is the pendant of Lemma
    \ref{lem:correct}).

    \begin{lemma} \label{lem:correct2}
	We have $\AA\models\phi_{qf}(\tuple{a})$ iff $\sem{\AA}'_{\tuple{a}}\models\phit{\phi_{qf}}'(\tuple{a})$.
\end{lemma}

As we had to characterise the well-formed $1$-data structure, a similar trick is necessary here. For  this matter, we use the following
formulas:
$$
  \begin{array}{ll}
\psitran &= \forall y \forall z.\Et_{p,q=1}^{n}\Et_{i,j,k,\ell=1}^D \Big(\udd{p}{i}{j}(y) \et \udd{p}{i}{\ell}(z) \et \udd{q}{k}{j}(y) \donc \udd{q}{k}{\ell}(z)\Big) \\
\psirefl(x_1,\ldots,x_n)  &=\Et_{p=1}^n\Et_{i=1}^D	 \udd{p}{i}{i}(x_p) \\

\psiwf(x_1,\ldots,x_n) &=\psitran \et \psirefl(x_1,\ldots,x_n) 
  \end{array}
  $$

Finally with the same reasoning as the one given in the previous
section, we can show that the formula $\phi=\exists x_1\ldots\exists
x_n.\linebreak[0]\phi_{qf}(x_1,\ldots,x_n)$ of $\eFO{D}{\Unary,\Binary_\nbd}{1}$ is satisfiable
iff the formula $\exists x_1\ldots\exists
x_n.\linebreak[0]\phit{\phi_{qf}}'(x_1,\ldots,x_n) \wedge \psiwf(x_1,\ldots,x_n) $
is satisfiable. Note that this latter formula  can be built in polynomial time from
$\phi$ and that it belongs to $\ndFO{0}{\Unaryp,\emptyset}$. Hence, thanks to
Theorem \ref{thm:0fo}, we obtain that $\nDataSat{\eFOr{1}}{\nbd,\Binary_\nbd}$ is in
\textsc{NEXP}. The matching lower bound is as well obtained the same
way by reducing $\nDataSat{\ndFOr}{0}$ to $\nDataSat{\eFOr{1}}{\nbd,\Gamma_\nbd}$
showing that a formula $\phi$ in $\ndFO{0}{\Unary,\emptyset}$ is satisfiable
iff the formula $\exists
x.\locformr{\phi}{x}{1}$ in $\eFO{\nbd}{\Unary,\Binary_\nbd}{1}$ is satisfiable.
  
\begin{theorem}\label{thm:exradius1}
	For all $\nbd \geq 1$, the problem $\nDataSat{\eFOr{1}}{\nbd,\Binary_\nbd}$ is \textsc{NEXP}-complete.
\end{theorem}

\subsection{Undecidability results}
We show here $\nDataSat{\eFOr{3}}{2,\{\relsaord{1}{1},\relsaord{2}{2}\}}$ and  $\nDataSat{\eFOr{2}}{3,\Binary_3}$
are undecidable. To obtain this we provide reductions from
$\nDataSat{\ndFOr}{2,\{\relsaord{1}{1},\relsaord{2}{2}\}}$ and we use the fact that any
2-data structure can be interpreted as a radius-3-ball of a 2-data structure or
respectively  as a radius-2-ball of a 3-data structure.

\subsubsection{Radius 3 and two data values}

In order to reduce $\nDataSat{\ndFOr}{2,\{\relsaord{1}{1},\relsaord{2}{2}\}}$ to the problem
$\nDataSat{\eFOr{3}}{2,\{\relsaord{1}{1},\relsaord{2}{2}\}}$, we show that 
we can transform any $2$-data structure $\AA$ into another
2-data structure $\AAge$ such that $\AAge$ corresponds to the
radius-3-ball of any element of $\AAge$  and this transformation has
some kind of inverse.  To obtain $\AAge$ from $\AA$, we add some
elements and in order to not take them into account for the
satisfaction of the formula, we label them with  a fresh unary  predicate
$\uge$. Furthermore, given a formula $\phi \in
\ndFO{2}{\Unary,\Binary_2}$, we transform it into a formula $T(\phi)$ in
$\eFO{2}{\Unary',\Binary_2}{3}$ such that $\AA$ satisfies $\phi$ iff $\AAge$
satisfies $T(\phi)$. What follows is the formalization of this reasoning.

Let  $\AA=(A,(\uP{\unary})_{\unary},\funi,\funo)$ be a $2$-data
structure in $\nData{2}{\Unary}$ and  $\uge$ be a fresh unary
predicate not in  $\Unary$. From $\AA$ we build the following $2$-data structure
$\AAge=(A',(\uP{\unary}')_{\unary},\funi',\funo')\in\nData{2}{\Unary\cup\{\uge\}}$
such that:
\begin{itemize}
	\item $A' = A \uplus \Values{\AA}\times\Values\AA$, 
	\item for $i\in\{1,2\}$ and $a\in A$, $f_i'(a)=f_i(a)$  and for $(d_1,d_2)\in  \Values{\AA}\times\Values\AA$, $f_i((d_1,d_2))=d_i$,
	\item for $\unary\in\Unary$, $\uP{\unary}'=\uP{\unary}$,
	\item $\uP{\uge}=\Values{\AA}\times\Values\AA$.
    \end{itemize}
Hence to build $\AAge$ from $\AA$ we have added to the elements of
$\AA$ all pairs of data presented in $\AA$ and in order to recognize
these new elements in the structure we use the new unary predicate
$\uge$. We add these extra elements to ensure that all the elements of
the structure are located in   the
radius-3-ball of any element of $\AAge$.
We have then the following property.

\begin{lemma}\label{lem:ge-has-radius-3}
 $\vprojr{\AAge}{a}{3}=\AAge$ for all $a \in A'$.
\end{lemma}
\begin{proof}
	Let $b\in A'$ and $i,j \in \{1,2\}$. We show that
    $\distaa{(a,i)}{(b,j)}{\AAge}\leq 3$. i.e. that there is a path of length at most 3 from $(a,i)$ to $(b,j)$ in the data graph $\gaifmanish{\AAge}$.
	By construction of $\AAge$, there is an element $c\in A'$ such that $f_1(c)=f_i(a)$ and $f_2(c)=f_j(b)$.
	So we have the path $(a,i),(c,1),(c,2),(b,j)$ of length at most 3 from $(a,i)$ to $(b,j)$ in $\gaifmanish{\AAge}$.
\end{proof}

Conversely, to $\AA=(A,(\uP{\unary})_{\unary},\funi,\funo)\in\nData{2}{\Unary\cup\{\uge\}}$, we associate $\AAminusge=(A',(\uP{\unary}')_{\unary},\linebreak[0]\funi',\funo')\in\nData{2}{\Unary}$ where:
\begin{itemize}
	\item $A' = A \setminus \uP{\uge}$, 
	\item for $i\in\{1,2\}$ and $a\in A'$, $f_i'(a)=f_i(a)$,
	\item for $\unary\in\Unary$, $\uP{\unary}'=\uP{\unary}'\setminus \uP{\uge}$.
\end{itemize}

Finally we inductively translate any formula
$\phi\in\ndFO{2}{\Unary,\{\relsaord{1}{1},\relsaord{2}{2}\}}$ into
$T(\phi)\in\ndFO{2}{\Unary\cup\{\uge\},\{\relsaord{1}{1},\relsaord{2}{2}\}}$
by making it quantify over elements not labeled with $\uge$:
$T(\unary(x)) = \unary(x)$, $T(\rels{i}{i}{x}{y})=\rels{i}{i}{x}{y}$
for $i \in \{1,2\}$, $T( x=y )= (x=y) $, $T(\exists x.\phi)=\exists x. \neg \uge(x) \wedge T(\phi)$, $T( \vp \vee \vp')=T(\vp) \vee T(\vp')$ and $T(\neg \vp)=\neg T(\vp)$. 

\begin{lemma}\label{lem:ge-vs-without}
	Let $\phi$ be a sentence in $\ndFO{2}{\Unary}$,
    $\AA\in\nData{2}{\Unary}$ and $\BB \in
    \nData{2}{\Unary\cup\{\uge\}}$. The two following properties hold:
 \begin{itemize}
	 \item $\AA\models\phi$ iff $\AAge\models T(\phi)$
  \item $\minusge{\BB} \models\phi$ iff $\BB\models T(\phi)$.
 \end{itemize}
\end{lemma}

\begin{proof}
	As for any $\AA\in\nData{2}{\Unary}$ we have $\minusge{(\addge{\AA})} = \AA$, it is sufficient to prove the second point.
	We reason by induction on $\phi$.
	Let $\AA=(A,(\uP{\unary})_{\unary},\funi,\funo)\in\nData{2}{\Unary\cup\{\uge\}}$ and let $\AAminusge=(A',(\uP{\unary}')_{\unary},\funi',\funo')\in\nData{2}{\Unary}$. 
	The inductive hypothesis is that for any formula $\phi\in\ndFO{2}{\Unary}$ (closed or not) and any context interpretation function $I: \Var \to A'$ we have $\AAminusge \models_I \phi \text{ iff } \AA \models_I T(\phi)$.
	Note that the inductive hypothesis is well founded in the sense that the interpretation $I$ always maps variables to elements of the structures.
	
	We prove two cases: when $\phi$ is a unary predicate and when
    $\phi$ starts by an existential quantification, the other cases
    being similar. First, assume that $\phi = \unary(x)$ where $\unary\in\Unary$.
	$\AAminusge \models_I \unary(x)$ holds iff $I(x)\in\uP{\unary}'$. 
	As $I(x)\in A\setminus \uP{\uge}$, we have $I(x)\in\uP{\unary}'$ iff $I(x)\in\uP{\unary}$, which is equivalent to  $\AA \models_I T(\unary(x))$ .
	Second assume $\phi = \exists x.\phi'$.
	Suppose that $\AAminusge \models_I \exists x.\phi'$.
	Thus, there is a $a\in A'$ such that $\AAminusge \models_\Intrepl{x}{a} \phi'$.
	By inductive hypothesis, we have $\AA\models_\Intrepl{x}{a} T(\phi')$.
	As $a\in A' = A\setminus \uP{\uge}$, we have $\AA\models_\Intrepl{x}{a} \neg\uge(x)$, so $\AA\models_I \exists x. \neg\uge(x)\et T(\phi')$ as desired.
	Conversely, suppose that $\AA \models_I T(\exists x.\phi') $.
	It means that there is a $a\in A$ such that $\AA \models_\Intrepl{x}{a}\neg\uge(x)\et T(\phi')$.
	So we have that $a\in A'=A\setminus \uP{\uge}$, which means that $\Intrepl{x}{a}$ takes values in $A$ and we can apply the inductive hypothesis to get that 
	$\AAminusge \models_\Intrepl{x}{a} \phi'$.
	So we have $\AAminusge \models_I \exists x.\phi'$.
\end{proof}

From Theorem \ref{thm:undec-general}, we know that
$\nDataSat{\ndFOr}{2,\{\relsaord{1}{1},\relsaord{2}{2}\}}$ is undecidable. From  a closed formula
$\phi\in\ndFO{2}{\Unary,\{\relsaord{1}{1},\relsaord{2}{2}\}}$, we build the formula $\exists
x.\locformr{T(\phi)}{x}{3}$ which belongs to $\eFO{2}{\Unary\cup\{\uge\},\{\relsaord{1}{1},\relsaord{2}{2}\}}{3}$. Now if
$\phi$ is satisfiable, it means that there exists $\AA\in
\nData{2}{\Unary}$  such that $\AA\models\phi$. By Lemma
\ref{lem:ge-vs-without}, $\AAge\models T(\phi)$. Let $a$ be an element
of $\AA$, then thanks to Lemma \ref{lem:ge-has-radius-3}, we have $\vprojr{\AAge}{a}{3}\models T(\phi)$.
 Finally by definition of our logic, $\AAge\models\exists x.\locformr{T(\phi)}{x}{3}$.
 So $\exists x.\locformr{T(\phi}{x}{3}$ is satisfiable. Now assume
 that $\exists x.\locformr{T(\phi)}{x}{3}$ is satisfiable. So there
 exist $\AA \in \nData{2}{\Unary\cup\{\uge\}}$ and an element $a$ of
 $\AA$ such that $\vprojr{\AA}{a}{3}\models T(\phi)$.
 Using  Lemma \ref{lem:ge-vs-without}, we obtain
 $(\vprojr{\AA}{a}{3})_{\setminus\uge}\models\phi$. Hence $\phi$ is
 satisfiable. This shows that we can reduce $\nDataSat{\ndFOr}{2,\{\relsaord{1}{1},\relsaord{2}{2}\}}$ to $\nDataSat{\eFOr{3}}{2,\{\relsaord{1}{1},\relsaord{2}{2}\}}$ .

\begin{theorem}\label{thm:undec-existential-r3}
	The problem $\nDataSat{\eFOr{3}}{2,\{\relsaord{1}{1},\relsaord{2}{2}\}}$ is undecidable.
\end{theorem}

\subsubsection{Radius 2 and three data values}

We give a reduction from $\nDataSat{\ndFOr}{2,\{\relsaord{1}{1},\relsaord{2}{2}\}}$ to
$\nDataSat{\eFOr{2}}{3,\{\relsaord{1}{1},\relsaord{2}{2},\relsaord{3}{3}\}}$. The idea is similar to the one used in the
proof of Lemma \ref{lem:hardness-radius2-2} to show that
$\nDataSat{\eFOr{2}}{2,\Binary_2}$ is \textsc{N2EXP}-hard by reducing
the problem $\nDataSat{\ndFOr}{1,\Binary_1}$. The intuition is that
when we add the same extra value to each node, then all the elements
of the structure are in the ball of radius $2$ of any element. Indeed we have the following lemma.

\begin{lemma} 
 Let $\phi$ be
a formula in $\ndFO{2}{\Unary,\{\relsaord{1}{1},\relsaord{2}{2}\}}$. There exists $\AA \in \nData{2}{\Unary}$ such that $\AA \models \phi$
 if and only if there exists $\BB \in  \nData{3}{\Unary}$ such that
 $\BB \models \exists
x.\locformr{\phi}{x}{2}$.
\end{lemma}

\begin{proof}
 Assume that there exists  $\AA=(A,(P_{\unary})_{\unary \in
   \Unary},\f{1},\f{2})$ in $\nData{2}{\Unary}$ such that  $\AA \models
 \phi$.Consider the $3$-data structure $\BB=(A,(P_{\unary})_{\unary \in
   \Unary},\f{1},\f{2},\f{3})$ such that $\f{3}(a)=0$ for all $a\in
 A$. Let $a \in A$. It is clear that we have $\vprojr{\BB}{a}{2}=\BB$
 and that $\vprojr{\BB}{a}{2} \models \phi$ (because $\AA \models
 \phi$ and $\phi$ never mentions the third values of the elements
 since it is a formula in $\ndFO{2}{\Unary,\{\relsaord{1}{1},\relsaord{2}{2}\}}$). Consequently $\BB
 \models \exists
 x.\locformr{\phi}{x}{2}$.

 Assume now that there exists $\BB=(A,(P_{\unary}),\f{1},\f{2},\f{3})$ in $ \nData{3}{\Unary}$ such that  $\BB \models \exists
x.\locformr{\phi}{x}{2}$. Hence there exists $a \in A$ such that
$\vprojr{\BB}{a}{2} \models \phi$, but then by forgetting the third
value in $\vprojr{\BB}{a}{2}$ we obtain a model in $\nData{3}{\Unary}$
which satisfies $\phi$.
\end{proof}

Using Theorem \ref{thm:undec-general}, we obtain the following result.
\begin{theorem}\label{thm:undec-existential-r2}
	The problem $\nDataSat{\eFOr{2}}{3,\{\relsaord{1}{1},\relsaord{2}{2},\relsaord{3}{3}\}}$ is undecidable.
\end{theorem}

\section{Conclusion}
\label{sec:conclusion}

\begin{table}[htbp]
  
  \begin{tabular}{|c|c|c|c||c|}
    \hline
  \textbf{Logic} & $\mathbf{r}$ & $\mathbf{Da}$ & $\mathbf{\Binary}$ &
     \textbf{Decidability status}\\
    \hline
     \hline
     & $-$  & $0$ & $\emptyset$
                                                                     &
                                                                       \textsc{NEXP}-complete\\
    & & & &
    {\small{\it(Thm  \ref{thm:0fo}\cite{borger-classical-springer97,etessami-first-ic02})}}\\
    
  \cline{2-5}
     \multirow{1}{*}{$\ndFO{\nbd}{\Unary,\Binary}$} & $-$  & $1$ & $\{\relsaord{1}{1}\}$
                                                                     &
                                                                       \textsc{N2EXP}-complete\\
    &&&&
    {\small{\it(Thm  \ref{thm:1fo}\cite{Mundhenk09,Fitting12})}}\\
   \cline{2-5}
     & $-$  & $2$ & $\{\relsaord{1}{1},\relsaord{2}{2}\}$
                                                                     &
    Undecidable {\small{\it(Thm  \ref{thm:undec-general}\cite{Janiczak-Undecidability-fm53})}}\\
    \hline
    & $1$ & $2$  &
                                                     $\{\relsaord{1}{1},\relsaord{2}{2},\relsaord{1}{2}\}$ & Decidable \small{\it{(Thm \ref{thm:datasat-rad1})}}\\
    \cline{2-5}
     & $1$ & $2$  &
                                                     $\{\relsaord{1}{1},\relsaord{2}{2},\relsaord{2}{1}\}$
                                                                     &
                                                                       Decidable
                                                                       \small{\it{(Thm
                                                                       \ref{thm:datasat-rad1})}}\\
      \cline{2-5}
    \multirow{2}{*}{$\rndFO{\nbd}{\Unary,\Binary}{r}$} & $1$ & $2$  &
                                                     $\Binary_2$
                                                                     &
                                                                      Open
                                                                       problem\\
    \cline{2-5}
     & $2$ & $2$  &
                                                     $\{\relsaord{1}{1},\relsaord{2}{2}\}$
                                                                     &
                                                                       Open problem\\
    \cline{2-5}
     & $2$ & $2$  &
                                                     $\{\relsaord{1}{1},\relsaord{2}{2},\relsaord{1}{2}\}$ & Undecidable \small{\it{(Thm \ref{theorem:2-loc})}}\\
 \cline{2-5}
     & $2$ & $2$  &
                                                     $\{\relsaord{1}{1},\relsaord{2}{2},\relsaord{2}{1}\}$ & Undecidable \small{\it{(Thm \ref{theorem:2-loc})}}\\
    \cline{2-5}
    & $3$ & $2$  & $\{\relsaord{1}{1},\relsaord{2}{2}\}$ & Undecidable \small{\it{(Thm \ref{theorem:3-loc})}}\\
    \hline
     & $1$ & $\geq 1$ &
                                                                    $\Gamma_\mathbf{Da}$  & \textsc{NEXP}-complete \small{\it{(Thm \ref{thm:exradius1})}}\\
    \cline{2-5}
    \multirow{2}{*}{$\eFO{\nbd}{\linebreak[0]\Unary,\Binary}{r}$} &
                                                                    $2$ & $2$ & $\Binary_2$ & \textsc{N2EXP}-complete  \small{\it{(Thm \ref{thm:radius2-2})}} \\
    \cline{2-5}
    & $3$ & $2$ &
                                                               $\{\relsaord{1}{1},\relsaord{2}{2}\}$
                                                                     &
                                                                       Undecidable
    \small{\it{(Thm \ref{thm:undec-existential-r3})}}\\
     \cline{2-5}
     & $2$ & $3$ &
                                                              $\{\relsaord{1}{1},\relsaord{2}{2},\relsaord{3}{3}\}$
                                                                     &
                                                                       Undecidable
     \small{\it{(Thm \ref{thm:undec-existential-r2})}}\\
    
    \hline
\end{tabular}
\caption{Summary of the results for the satisfiability  problem}
\label{tab:summary}
\end{table}
In this work, we have tried to pursue an exhaustive study in order to
determine when the satisfiability of the local first order logic with
data is decidable. The results we have obtained are provided in Table
\ref{tab:summary}. We observe that even if we consider the existential
fragment, as soon as the view of the elements has  a radius  bigger
than $3$, the satisfiability  problem is undecidable. This table
allows us as well to see what are the missing elements in order to
fully characterize the decidability status of this satisfiability
problem. In particular, we do not know whether the decidability result
of Theorem \ref{thm:datasat-rad1} still holds when considering two
diagonal relations, but our proof technique does not seem to directly
apply. It would be as well very interesting to establish  the relative
expressive power of these logics and to compare them  with some other
well-known fragments as for instance the guarded fragment. Finally,
another research direction would be to see how our logic could be used
to verify distributed algorithms with data (each element in our model
representing a process).

\bibliographystyle{alphaurl}
\bibliography{biblio-lmcs}

\end{document}